%% file: renesis_arxiv_v17.tex
\documentclass[journal]{IEEEtran}

\usepackage{amsmath,amssymb,amsthm}
\usepackage{booktabs}
\usepackage{graphicx}
\usepackage{url}
\usepackage{algorithm}
\usepackage{algpseudocode}
\usepackage[siunitx,americanvoltages]{circuitikz}
\usetikzlibrary{calc,positioning,arrows.meta,shapes.gates.logic.US,fit}
\usepackage{balance}

\newtheorem{lemma}{Lemma}
\newtheorem{corollary}{Corollary}
\newtheorem{proposition}{Proposition}
\theoremstyle{definition}

\theoremstyle{remark}
\newtheorem{remark}{Remark}

\newcommand{\F}{\mathbb{F}_2}
\newcommand{\Hinf}{H_\infty}

\newcommand{\tool}{\textsc{Renesis}}

\begin{document}

\title{\tool: Energy-Aware Synthesis\\ of Adiabatic Logic from Irreversible Netlists}

\author{Mitchell~A.~Thornton,~\IEEEmembership{Senior Member,~IEEE}%
\thanks{M.~A.~Thornton is with the Darwin Deason Institute for Cyber Security
and the Department of Electrical and Computer Engineering, Southern Methodist
University, Dallas, TX 75275 USA (e-mail: mitch@smu.edu).}%
\thanks{August 15, 2026}%
\thanks{\tool\ is released as open source.  The source, the benchmark
netlists, the validation procedure, and the run records behind every number
reported here are available at
\protect\url{https://github.com/mitch-thornton/renesis-low-energy-circuit-synthesis}.}}

\maketitle

\begin{abstract}
We describe \tool, an automated synthesis tool that accepts an ordinary
irreversible netlist and produces a verified, technology-mapped
energy-recovery (adiabatic) circuit, with energy rather than area or delay
as the optimization criterion.  \tool~models the netlist under consideration
with a vector-space formulation that enables simulation and justification sweeps
as forward and reverse netlist traversals that are linear with respect to 
the number of circuit components.  The netlist traversals 
populate ledgers with data `tags' that characterize switching, erasure, and
observability information at their natural R\'enyi orders. \tool~output is in the form
of a logically reversible circuit that is mapped to one of eight different
energy-recovery families with associated parameters included.  Every synthesis
transformation is equivalence-checked and must improve one of two reported
cost tables while worsening neither before it is accepted.  Across a twenty-circuit development
set, optional re-synthesis passes are demonstrated to improve fourteen circuits, and on a
held-out set of twenty circuits, fifteen of
nineteen are improved with a best-arm median of $0.91$ of the default energy resulting.
A certified optimality-gap program computes the distance between the synthesized
circuits and the provable floor of the tool's own search space.  
\end{abstract}

\begin{IEEEkeywords}
Adiabatic logic, energy-recovery logic, low-power synthesis, technology
mapping, switched capacitance, Reed--Muller expansion, switching activity,
Landauer's principle, reversible logic.
\end{IEEEkeywords}

\section{Introduction}
\label{sec:intro}

Ultra-low-energy computing is of widespread interest, particularly for pervasive processors,
wireless sensors, and other devices that run from batteries or harvest their
energy from ambient sources.  Most devices of this category are designed 
and implemented using conventional CMOS technology that expends energy 
in a structurally avoidable way where a logic
transition charges a node capacitance from the supply, and that stored
energy is discarded as heat when the node discharges during a logic-level transition.  
A circuit style known as adiabatic, or energy-recovery, logic changes the supply's shape
rather than the logic's function in a way such that voltages ramp gradually and the 
corresponding drop across any conducting path is kept small throughout the transition. 
In this way, the charge is steered back into the supply at the end of a cycle instead of being
dumped to ground~\cite{athas1994adiabatic,koller1992adiabatic}.  The
motivation for this circuit style is potentially very large energy efficiency.  The cost is that
the circuitry is considerably less convenient than typical CMOS designs because adiabatic
families generally require multi-phase ramped power clocks, careful timing
discipline, and often require additional devices and interconnect resources 
as well as more area. It is also the case that losses within the power-clock 
generator itself must be carefully considered and controlled to determine whether a
system-level saving survives.  Conventional high-performance CMOS has the benefit of
five decades of use in terms of optimizing for nearly the opposite objective, 
high-performance through hard and fast switching where the charge is dissipated
largely through heat.  Thus, a fully adiabatic design
methodology differs from a standard-cell flow in more ways than through
simply specifying a new cell library during the technology mapping process.

Constructing adiabatic circuits so that they are also \emph{logically reversible} 
addresses the information-erasure contribution to energy
loss, and the reason for this form of circuit deserves more attention 
than it usually receives, because the justification for this choice differs
from the usual thermodynamic information-erasure arguments usually cited.
Circuit designers proficient in low-power methods will note that CMOS 
switching energy is approximately six orders of magnitude higher than the 
Landauer bound, and will likely conclude that the cost of information erasure 
does not matter.  When an adiabatic logic gate erases information, however, 
the penalty is not $k_BT\ln 2$.  It is a full non-adiabatic $CV^2$ discharge
event.  Energy recovery requires steering charge back out along a conducting
path, and no such path exists unless the state is un-computed.  Logical
irreversibility therefore costs switching-scale energy rather than
Landauer-scale energy.  Reversibility in this flow is not an appeal to a
distant thermodynamic limit.  It is the structural precondition for the
$CV^2\cdot(RC/T)$ recovery scaling to continue rather than terminate at an
irreversibility floor.  Landauer's bound locates where that floor would lie
if everything else were perfect.  The practical value of reversibility
therefore appears six orders of magnitude sooner.

The tool described here makes the issue of logical reversibility one of
practicality rather than an ideological goal.  Reversible circuits incur
a measurable price in terms of required ancilla lines, garbage outputs, 
and un-compute overhead, and the \tool~flow prices that overhead
against the energy recovery gain per netlist, under an energy model, with the
result reported rather than assumed.  \tool~was designed to ledger, update,
and account for energy pricing at each optimization stage rather than
simply choosing a logically reversible target based upon general principles.

We selected reversible adiabatic CMOS as the design style for low-energy
applications and reviewed the literature for tool support
(Section~\ref{sec:gap}).  That review located few automated synthesis tools
for the style.  None of the tools we located used an optimization criterion
based on energy itself.  
Furthermore, it is often the case that an existing irreversible design 
must be converted to an ultra-low-energy target, since designs are 
rarely generated in reversible form from the start.  For these reasons we 
are motivated to build \tool, a logic synthesis tool that receives an 
irreversible netlist as input and produces a validated, technology-mapped 
reversible adiabatic circuit in any of eight different energy-recovery 
families, where the optimization criterion at every stage is energy 
rather than the customary area or performance.

The application of such circuits is unlikely to be a general-purpose
processor.  Adiabatic operation trades speed for energy roughly one for one,
so the deployments that suit it are highly repetitive computational kernels
where the hardware stays busy and the recovery is amortized over large
operation counts.  Examples include 
checksum, dispatch, and table-lookup logic of persistent sensing workloads, 
and the data-plane portions of accelerator architectures in which conventional
logic handles control while adiabatic logic carries the repetitive
arithmetic~\cite{vaire2025whitepaper}.  The set of development benchmarks
described in Section~\ref{sec:results} includes two dense substitution-table circuits
for exactly this reason.

\tool~makes several contributions with regard to ultra-low-power design methods
and electronic design automation (EDA). First, an energy-accounting
model is described that accounts for and keeps switching and erasure 
estimates separate while reporting both, and directs optimization 
goals using the term that dominates by orders of
magnitude which is the CMOS switching energy at present 
technology nodes. We justify the account of why these two
mechanisms are different via the use of R\'enyi entropy orders in
Section~\ref{sec:background} and use these metrics throughout our
description of the tool.  Second, we describe a synthesis
flow where every accepted transformation is equivalence-checked against
the original netlist and must improve one of two reported cost tables
while worsening neither, so that the flow cannot return a circuit worse
in terms of energy metrics than its input with regard to either reported figure.
The synthesis tool is built upon a foundation of the vector-space transfer-function 
formulation as described in~\cite{thornton_vsim_book} that simultaneously enables efficient
ledgering, information-theoretic, and switching function characterization
while also supporting an automated flow.  Third,
a development and validation methodology, as described in
Section~\ref{sec:development}, resulting in a tool that exists as two complete
implementations held to byte-level agreement, and under which failures,
refusals, and negative results are retained as reportable artifacts.
The development methodology applies to other software projects of
comparable size.  Fourth, we provide 
measured results using a twenty-circuit development set and a twenty-circuit held-out
set whose circuits no development decision touched. We also include a nine-family
cross-technology study, and a certified optimality-gap program that brackets
the distance between shipped circuits and the floor of the tool's own
search space (Section~\ref{sec:results}).

\subsection{Contributions and Related Work}
\label{sec:gap}

Our review of the literature did not reveal an automated synthesis
flow that transforms an irreversible circuit netlist into a reversible or
adiabatic form with energy as the principal optimization criterion.
We did find three bodies of work that are adjacent to this one that
optimizes criteria other than energy while targeting an 
energy-recovery technology.  However, reversible-logic synthesis is 
a fairly mature field with respect to different objectives.

There is a substantial and actively maintained body of work regarding
open tooling synthesis of reversible circuits that includes 
transformation-, decomposition-, BDD-, KFDD- and ESOP-based methods.
Approaches include exact synthesis, reversible pebbling and
compute--uncompute compilation resulting in tools such as Caterpillar \cite{soeken2022epfl}\cite{meuli2019pebbling}, 
RevKit \cite{soeken2012revkit},
Tweedledum \cite{schmitt2022tweedledum}, 
MQT SyReC \cite{adarsh2022syrec}\cite{wille2010syrec}, 
ReVerC \cite{amy2017reverc} and others.  The objectives 
of these approaches and applications include one or more of gate count,
various forms of quantum cost, $T$-count (total number of $T$-gates in a decomposed quantum circuit), 
reversible circuit lines, ancilla count, and other metrics
based on quantum-computing targets.  Among the tools we examined we
did not find one focused solely on switched capacitance as the technology target where
charge is recovered rather than dissipated, nor an objective in which
energy appears directly.  

A circuit that is efficient in quantum cost does not generally
result in one that is inexpensive in terms of overall transported charge when mapped
to a conventional electronic circuit. Indeed, one of the results reported here is 
that, of the circuits we measured, low quantum cost can result in an energy-expensive
circuit with respect to switched capacitance. Likewise, an electronic circuit with
relatively low switched capacitance can result in high quantum cost
when mapped from a classical electronic technology to a quantum technology.
Moreover, our analyses indicates that 
the various metrics based on quantum technology are poor proxies for
grading the energy efficiency of a conventional electronic circuit.
We therefore treat these tools as producers of a common
technology-independent model, a logically reversible abstraction, rather
than as competitors or substitutes for our objective.  

Adiabatic circuit design appears to be largely accomplished via manual layout
at the time of this writing.  The adiabatic and energy-recovery literature we 
surveyed is predominantly found in the circuit-design literature where families 
are proposed and simulated at the transistor level using hand-built blocks 
such as inverter chains, full adders, small multipliers, and cipher rounds.
Typically, these designs are compared to to static CMOS implementation
in terms of energy and dissipated power over some frequency range.  
A 2019 survey characterizes this position directly, observing that the ``majority of the
work in adiabatic circuits is implemented as a full custom design'' and
that ``robust EDA methodologies are needed, which remains as an open
problem''~\cite{dhananjay2019iscas}.  The published design flows we
examined generally proceed by schematic capture, manual layout, parasitic
extraction, and circuit simulation, with no automated synthesis,
technology-mapping, or place-and-route processing.

The automated adiabatic flows we did find tend to optimize structural proxies to energy.
Zulehner, Frank and Wille~\cite{zulehner2019adiabatic} give the one of the most
complete automated flows that we could locate. They describe a two-stage method that realizes a
function against a gate library and then maps the result to an adiabatic
circuit comprising transmission gates together with a power-clock schedule for
a 2LAL target that is then evaluated using ISCAS and IWLS benchmarks. The reported metrics
include transmission-gate count, number of power clocks, and runtime.
Their approach builds upon earlier work exploiting reversible design for adiabatic
implementation in~\cite{rauchenecker2017mixdes}.  Morrison and
Ranganathan~\cite{morrison2014tcad} synthesize dual-rail adiabatic logic,
the same target class as ours, for resistance to differential-power-analysis 
rather than for energy, so their objective focuses on balancing the
charge drawn rather than minimizing its amount.  Clark, Raffel and
Thapliyal~\cite{clark2024isvlsi} report the generation of
dual-rail adiabatic gates from a BDD specification, again with a
security objective.  Earlier, Ye~\cite{ye1997thesis} proposed an
XOR-based decomposition diagram aimed at adiabatic implementation, and
H\"anninen \emph{et al.}~\cite{hanninen2014adiabatic} outline
how adiabatic circuit generation might be inserted into a conventional
standard-cell design flow.  We were not able to locate a public source release
for any of these, which we admit could be more of a statement about our searching
method and results rather than one about the actual
availability of the corresponding tools or prototypes.

A significant body of work exists for switching-aware synthesis, 
but this is a different cost model than that we consider here.
Power-driven technology mapping is well established in the conventional
regime, founded by Tsui \emph{et al.}~\cite{tsui1993dac} and Tiwari
\emph{et al.}~\cite{tiwari1993dac} and supported by a mature body of literature 
based on activity estimation~\cite{najm1994survey}. The popular Berkeley ABC 
tool implements a switching-activity-aware cut selection in its
mappers~\cite{mishchenko2011power}, and we used it as a
conventional reference point during the implementation and design of~\tool.  
These approaches optimize $\alpha CV^2f$, where a node costs 
energy only when its value changes.  In comparison, energy-recovery
logic charges every node from the power clock on every cycle whether or
not its value changes or ``switches,'' so activity $\alpha$ enters the cost function differently
in these approaches and in our measurements, it has materially less leverage
as shown in Section~\ref{sec:results}.  Transplanting a conventional low-power
objective into an adiabatic target therefore appears not to be a suitable substitute 
and we treat it as something to be measured rather than assumed.

For completeness, we also note that the adiabatic technology with the
most developed open tooling that we could find is not based on CMOS.  The
``adiabatic
quantum-flux-parametron'' (AQFP) logic style has an actively maintained open-source
synthesis and buffer-insertion toolchain together with a shared benchmark
repository~\cite{bairamkulov2024date}.  AQFP is a superconducting technology
that is usually based on
majority-gate logic rather than a dual-rail CMOS-based transistor technology, and 
it is priced in terms of the number of Josephson junctions required 
rather than switched capacitance, so we do not treat it as a baseline. We
mention AQFP because its practice, an explicit per-cell cost model and a
public repository with stored results, is the standard we have tried to
meet.

Against this background we state our contribution in as specific terms as possible.  
The quantity minimized in \tool~is the switched capacitance that sets the energy 
requirements of a circuit rather than using a
structural proxy for it. This objective is carried consistently from
the initial covering stage through the technology mapping phase and the
resulting flow is a single released tool rather than a collection of
per-experiment results.  

\subsection{Organization}

Section~\ref{sec:background} reviews the concepts the rest of the paper
assumes.  These are energy dissipation in conventional and adiabatic CMOS
and its relation to the R\'enyi entropy family, the switching-theoretic
representations the tool manipulates, the synthesis and EDA machinery it is
built from, and the vector-space netlist model underlying all of it.
Section~\ref{sec:renesis} describes the architecture of \tool~and its underlying
algorithms, with the formal results placed next to the mechanisms they
support and a worked example closing each stage.
Section~\ref{sec:development} describes how the tool was developed and
validated with the aid of an AI agent.  Section~\ref{sec:results} reports
the measured performance of \tool.  Section~\ref{sec:discussion} interprets
these results.
\section{Background and Basis}
\label{sec:background}

This section reviews the results and concepts used in the remainder of the
paper, so that the report is self-contained.  Most of the content is
established material cited to its source, or a connection between
established materials stated once here so that later sections can depend on
it.  Readers who work in adiabatic circuit design may skim
Section~\ref{sec:adiabatic-bg}.  Readers who work in logic synthesis will
find it worth a closer reading since the cost model behaves differently
in comparison to those used in conventional low-power synthesis since the early
nineties.

\subsection{Adiabatic and energy-recovery logic}
\label{sec:adiabatic-bg}

Conventional CMOS wastes energy for a structural reason rather than an
unavoidable one.  When a gate output rises, charge flows from a supply
held at $V$ (Vdd) into a load connected to $0$ (GND).  The full supply
voltage appears across this conducting channel at the start of the
transition, and the energy dissipated in that channel over the whole
transient is
\begin{equation}
E_{\text{conv}} = \tfrac{1}{2} C V^{2}
\label{eq:cv2}
\end{equation}
regardless of how the channel is sized or how slowly the input arrives.
On the falling edge, the stored charge is dumped to ground with a
portion of it being lost to heat, so the supply recovers none of it.  Over a circuit
switching at frequency $f$ with per-node activity factors $\alpha_i$, the
average dynamic power is $P_{\text{conv}} = f V^{2}\sum_i \alpha_i C_i$.
This is why \emph{switched capacitance}, $C_{\text{sw}}=\sum_i\alpha_i
C_i$, became the standard target of low-power synthesis.  At fixed $V$ and
$f$ it is the only factor a synthesis tool can influence.  Two properties of
\eqref{eq:cv2} deserve emphasis because neither survives into the adiabatic
case.  The energy does not depend on how fast the transition happens, and it
does not depend on the resistance of the path that carries the charge.

Adiabatic switching removes the large instantaneous voltage
drop~\cite{athas1994adiabatic,koller1992adiabatic}.  The stiff supply is
replaced by a \emph{power clock} that ramps slowly, so the load voltage
tracks the source and the drop across the conducting path stays small at
every instant.  Charging $C$ through a resistance $R$ with a linear ramp
of duration $T$ dissipates
\begin{equation}
E_{\text{ad}}(T) \;=\; \frac{RC}{T}\, C V^{2}
\qquad (T \gg RC),
\label{eq:adiabatic}
\end{equation}
and on the reverse ramp the charge flows back into the clock rather than
to ground, so the system recovers rather than dissipates.  The dissipated
energy per transition can be made small by slowing the ramp, and there is no
$\tfrac12 CV^2$ floor.  The name is aspirational rather than exact.
Practical families are quasi-adiabatic.  Threshold drops, leakage, and
finite clock quality leave a residue, and combining the corrections gives
the per-cycle behavior of a real adiabatic node as
\begin{equation}
E(T) \;\approx\; \underbrace{\frac{RC}{T}CV^{2}}_{\text{adiabatic}}
      \;+\; \underbrace{E_{\text{nonadiab}}}_{\text{residue, flat}}
      \;+\; \underbrace{P_{\text{leak}}\,T}_{\text{grows with }T}.
\label{eq:total}
\end{equation}
Equation~\eqref{eq:total} has a methodological consequence that is easy to
state and easy to violate.  An energy figure for an adiabatic circuit is not
meaningful unless the throughput at which it was obtained is also given,
since any adiabatic circuit can be made to look better by slowing it
down.  Comparisons in this paper are made at matched
convention, and where a curve matters we report the curve
(Section~\ref{sec:spice}).

Three practical costs accompany this technique and all three shape 
the objectives of a synthesis tool designed to address it. Energy is 
traded for time, essentially one for one in the adiabatic regime.  
The constant-voltage DC supply is replaced by a multi-phase 
clock with slowly rising and falling, or ``ramped'' pulse waveform, 
typically with two, four, or eight phases, which must itself be generated 
efficiently or the energy recovery of an efficiently-designed processing circuit
is lost upstream. Multi-phase clock generation and distribution is 
outside this paper's scope, and we only count the charge that the 
processing logic uses from the ``power-clock.''  And, as previously 
mentioned, most families are \emph{dual-rail} in design, meaning that every
signal travels as a complementary pair, so that each output has a defined
charge and recovery path.  Dual-rail circuitry roughly
doubles the device count in comparison to a single-ended static 
implementation, and the pull networks are pass-transistor structures 
rather than restoring CMOS stacks.  This is the first place where 
intuition transfers badly from conventional synthesis because 
overall device count does not dominate the cost. Consequently,
a transformation that adds devices while reducing the transported
charge lowers the cost, which runs against the intuition carried over from
conventional static CMOS design.

Within the transmission-gate families this work targets, a logic node is
realized as two complementary pass networks between a power-clock phase
and the node's two rails.  Each literal of the node's function contributes
one pass device, whose gate terminal presents a capacitance to whatever
signal drives it. Thus, in each cycle, a node charges the gate capacitance of every
pass device it feeds.  From this point of view, the energy-relevant quantity 
can be considered as roughly equivalent to the
number of \emph{literal occurrences} in the mapped network and how many times
each signal is used as a control, rather than the number of gates or the circuit
depth.  Two mappings of the same function with identical gate counts can
differ substantially in energy, and a mapping with more gates can be
less expensive using the currency of energy. Section~\ref{sec:ledger} 
describes how we use this observation as a formal property of the
cost functional and Section~\ref{sec:results} measures it repeatedly.

The inherent series composition of the target circuit style results in
a corresponding realizability constraint.  A chain of $d$
pass devices has resistance of roughly $d\,R_{\text{on}}$ and degrades
the signal it conducts, so the pass-transistor design style bounds the chain
depth, typically at two to four devices.  We treat that bound as a design
rule rather than as a cost term, for a reason developed in
Section~\ref{sec:ledger}.

Finally, we reiterate that adiabatic operation and logical reversibility are separate ideas
although they are often introduced together.  Landauer's bound concerns the
erasure of information and is operable at a higher-level abstraction 
providing a statement about information and the logic used to 
process it, or to convert it from one form to another.  Adiabatic charging
concerns the transport of electrical charge and is a statement at a lower
level of abstraction about the circuits themselves.  The two abstractions
connect at the margin described in the introduction.  A family that must
return its nodes to a known state in order to recover charge is essentially performing
an un-compute step from an information-theoretic point of view, and that
step motivates the logic it realizes to be reversible.  The fully static,
fully adiabatic families make this motivation become a structural one.  S2LAL requires one clock
cycle per logic stage and an eight-cycle minimum period in order to schedule
the compute and un-compute phases~\cite{frank2020s2lal}.

\subsubsection{The Supported Target Families}
\label{sec:families}

The technology mapper in \tool~targets eight published energy-recovery families through one
common engine.  These families differ in the shape of the pull network, the number
of power-clock phases, the device overhead per gate, and, most
consequentially for energy, whether or not they leave a non-adiabatic residue.
A ninth target, a structurally mapped NOR network is also supported, 
and is only carried as a comparison baseline and is not an energy-recovery family.  
The per-family parameters are exactly those in the released tool's technology files, so
a reader can reproduce any figure in this paper from the same numbers the
mapper used or they may change the parameters to any other set they desire.

\IfFileExists{figures/fam_cells_owner.tex}{\input{figures/fam_cells_owner}}{}

The reference target of \tool\ is the transmission-gate family (TG), dual-rail
series-parallel networks of full CMOS transmission gates carrying charge
between a four-phase power clock and the two output rails.  Because a
transmission gate passes both polarities without a threshold drop, the
family leaves no residue, and it carries no per-gate device overhead,
which makes it the cleanest family in which to measure the cost of a 
particular synthesis decision.  Its disadvantage is that each literal requires 
a NMOS/PMOS pair.  It descends directly from the adiabatic-switching formulation of
Athas \emph{et al.}~\cite{athas1994adiabatic}.  ECRL~\cite{moon1996ecrl}
replaces the transmission gates with NMOS-only function networks under a
cross-coupled PMOS latch where a literal costs one device rather than two, at
the price of a hard threshold residue of roughly $(V_{tp}/V)^2$ per cycle
that no ramp time removes.  PFAL~\cite{vetuli1996pfal} places the
function networks in parallel with the latch devices, removing ECRL's
residue at the largest per-gate overhead of the four-phase families.
PAL~\cite{oklobdzija1997pal} runs dual-rail NMOS pass networks
from a two-phase clock, halving the clock-distribution burden.
CAL~\cite{maksimovic1995cal} adds an auxiliary clock that lets a
pipeline stall without losing recovered charge, at the largest overhead
in the set and an ECRL-class residue.  2LAL~\cite{aghoram2004twolal}
is transmission-gate pass logic in which signals are pipelined and where a value
is valid for a bounded window, so consumers more than one level
downstream require explicit relay stages, and logic depth translates
directly into hardware.  S2LAL~\cite{frank2020s2lal} is its fully static,
fully adiabatic complement, requiring an eight-phase clock and roughly
double the devices per stage.  SPGAL is a transmission-gate family with
symmetric, energy-balanced charging designed for power-analysis
resistance, from the same line of work as
EE-SPFAL~\cite{kumar2019eespfal}.\footnote{The survey underlying this
section could not locate a citable origin paper for SPGAL specifically.}

\subsubsection{Mapped Technology Families}
\label{sec:family-records}

Figure~\ref{fig:family-cells} illustrates one representative cell per family, and
readers comparing the eight panels will notice that several of them
coincide.  This point is highlighted because at the individual cell level,
some of these families are the same circuit.  Cells (g) and (h) have identical
topology, a cross-coupled PMOS pair loading a dual-rail pass network, and
separated only by what is inside the network and by how many devices
the latch carries.  Cells (e) and (g) differ only in where the network
terminates, at ground in one case and at the dual-rail inputs in the other.

Where the families separate is the parameter record, and since those
records ship as data rather than as code, we print them in
Table~\ref{tab:famrecords} exactly as the mapper reads them.

\begin{table}[t]
\caption{The eight family records as they ship, read from
\texttt{config/technology/} of the released tool.  $\phi$ is the phase
count, \emph{ovhd} the per-gate overhead devices, \emph{self} the self-load
devices on the cell's own output rails, \emph{clk} the clock-driven
devices, and \emph{dev} the device type of the function network, N for nMOS
pass devices and T for full transmission gates, with its input capacitance
beside it.  The starred row is the shipped default.}
\label{tab:famrecords}
\centering
\small
\input{tab_famrecords}
\end{table}

Two readings of the table are important in terms of the results reported 
in a following Section of this paper.  The
first reading is that the table separates families the schematic could not.
\textsc{SPGAL} carries twice \textsc{PAL}'s latch overhead and a function
network of transmission gates rather than NMOS pass devices, which is a
factor of two in capacitance per device, so the two price very differently
even though they draw identically.

The second reading is less obvious and we state it rather than have it
inferred from a results table.  \textsc{ECRL} and \textsc{PAL} have
\emph{identical} records in every field the switched-capacitance model
reads: the same overhead, the same self-load, the same NMOS device at the
same capacitance.  They differ in phase count and in the non-adiabatic
residue, and the metric of Sec.~\ref{sec:ledger} sees neither, because a
count of transported charge is blind to how many power-clock rails deliver
it and blind to a floor that does not scale with the ramp.  The
consequence appears directly in Sec.~\ref{sec:families-results}, where the
two families return the same figure to three decimal places on every
circuit, and it is the reason the win column in that table should be read
as a tie rather than as a preference.  A designer choosing between
\textsc{ECRL} and \textsc{PAL} on this evidence is choosing on the
strength of the residue and the clock budget, which our tables do not
price, rather than on the numbers where they do matter.  Section~\ref{sec:spice} is the
instrument that begins to close that gap, and the family-by-family
validation it enables provides this information.

Two further observations cut across the set of families and drive the experiments of
Section~\ref{sec:results}.  The families do not agree on what is
expensive.  ECRL and CAL pay a fixed residue and provide inexpensive literals, while TG,
PFAL, 2LAL, and S2LAL pay expensive literals and no residue, so a
transformation that trades literal count against gate count need not have
the same sign in every family.  Whether the re-synthesis passes are
family-independent is therefore an empirical question, and one this paper
does not answer.  The family studies of Section~\ref{sec:families-results}
run the default configuration across all nine targets, and the pass studies
run every arm against the release family, so no pass is observed under
another family anywhere in our measurements.  We say what was measured and
mark the cross-cases as not run.
The residue also changes the meaning of a ramp-time sweep.  A residue-free
family's $E(T)$ curve falls as $1/T$ indefinitely, while ECRL's and CAL's
flatten onto the residue, a difference Section~\ref{sec:spice} uses as a
check on the simulation itself.

\subsection{R\'enyi Entropy and Energy}
\label{sec:renyi-bg}

For a random variable $Y$ with distribution $\{p_y\}$, the R\'enyi
entropy of order $\alpha\ge0$, $\alpha\ne1$, is\footnote{Here we
use $\alpha$ as a parameter for R\'enyi entropy order in keeping with convention and we caution
the reader that $\alpha$ does not refer to its previous use where it 
denotes switching activity.}
\begin{equation}
H_\alpha(Y)\;=\;\frac{1}{1-\alpha}\,\log_2\!\sum_y p_y^{\,\alpha},
\label{eq:renyidef}
\end{equation}
with $H_1$ defined by the limit $\alpha\to1$.  It is a one-parameter
family of entropies rather than a single quantity, and the parameter
selects how heavily the most likely outcomes are weighted.  Four members
are used in this paper.  The max-entropy $H_0$ is the logarithm of the
number of outcomes with $p_y>0$ and ignores the probabilities entirely.
Shannon information entropy $H_1$ is the average surprise and is the 
same quantity as that used in
Landauer's principle.  The collision entropy $H_2 = -\log_2\Pr[\text{two
independent draws agree}]$ is algebraically the most tractable member.
The min-entropy $\Hinf = -\log_2\max_y p_y$ is governed by the single
most likely outcome.  The family is non-increasing in $\alpha$, so
$H_0\ge H_1\ge H_2\ge\Hinf$, and that ordering is not a technicality
here, because it is what makes the costs below comparable.

Landauer's principle, in its ensemble-average form, states that erasing a
logical bit whose prior carries Shannon entropy $H_1$ dissipates at least
$k_BT\ln2\cdot H_1$ into the
environment~\cite{landauer1961,reeb2014landauer}.  At $T=300$\,K the
per-bit constant evaluates to approximately $2.87\times10^{-21}$\,J.  Two
qualifications matter for a tool that intends to report the quantity.
It is a bound on an average over the prior distribution, not a guarantee
about any single erasure event.  Statements about a single shot require
the smooth min- and max-entropies rather than $H_1$, and carry correction
terms that vanish only in the many-copy limit~\cite{reeb2014landauer}.
It is our view that a tool that reports a per-run erasure figure computed from $H_1$ is
reporting an ensemble quantity, and \tool\ follows that viewpoint.  The
bound concerns information that is destroyed, which, in turn, depends on what the
circuit is required to produce rather than on how it is written. For this reason, the
erasure ledger is a property of the specification together with its
boundary and must be computed rather than assumed.
Bennett~\cite{bennett1973} showed that a computation performing no
erasure has no such bound, which is the observation the reversible
circuit design approach descends from.

It is constructive to consider how large the information erasure term is, in Joules, at present technology nodes.  
We compute this value.  One transmission gate
switching in the characterized library this work uses (Nangate45 \cite{nangate45} typical,
$C=1.70$\,fF at the inverter input pin, $V_{dd}=1.1$\,V) costs
$CV^2 = 2.06\times10^{-15}$\,J, which is about $7.2\times10^{5}$ times
$k_BT\ln2$.  Charge recovery narrows the gap without closing it.  At a
ramp factor $r=RC/T$ the dissipation is $rCV^2$, so even an aggressive
$r=10^{-3}$ leaves a factor of about $7\times10^{2}$.  Scaling the node
does not change the conclusion.  At a contemporary 28\,nm operating point
the same calculation gives $2.3\times10^{5}\,k_BT\ln2$, under one order
of magnitude below the 45\,nm figure against a gap of five to six orders.
Two consequences follow, and they point in opposite directions.  For now,
switching loss is the quantity worth optimizing, and a flow that spends
switching activity to avoid erasure is spending in the wrong currency, and
that is the premise of this tool.  The ordering is not permanent.
$CV^2$ falls with each new device technology generation while $k_BT\ln2$ does not move,
so circuits synthesized against the switching term remain valid when the
floor eventually binds, while the reverse is not true.

The connection this paper leans on is that the two energy mechanisms are
different orders of one object, and it can be stated as a one-line lemma.

\begin{lemma}[Switching activity is a R\'enyi quantity of order two]
\label{lem:renyi2}
Let $h(X)$ be an internal signal of a combinational network, and let
successive input vectors $X_{t-1},X_t$ be independent and identically
distributed.  Then the switching activity
$\alpha_h=\Pr[h(X_t)\neq h(X_{t-1})]$ satisfies
$\alpha_h = 1 - 2^{-H_2(h(X))}$.
\end{lemma}
\begin{proof}
Write $q=\Pr[h(X)=1]$.  By independence of the two draws,
$\alpha_h = \Pr[h(X_t)=1]\Pr[h(X_{t-1})=0] +
\Pr[h(X_t)=0]\Pr[h(X_{t-1})=1] = 2q(1-q) = 1-\big(q^2+(1-q)^2\big)$.
The subtracted quantity $q^2+(1-q)^2=\sum_z p_z^2$ is the collision
probability of the signal's distribution, which is $2^{-H_2}$ by
\eqref{eq:renyidef} at $\alpha=2$.
\end{proof}

The reason this is worth stating is what it says when read next to
Landauer's bound.  Erasure is governed by the order-one R\'enyi  entropy of the
erased distribution.  Switching activity, and therefore dynamic power, is
governed by the order-two R\'enyi entropy of the same distribution.  The two
mechanisms need not move together.  A distribution can be reshaped so that
$H_2$ falls while $H_1$ is essentially unchanged, which is, in our
viewpoint, the correct information-theoretic description of what a
low-switching-activity realization of a fixed function accomplishes.
A third R\'enyi order also appears in this paper, governing the minimum number of garbage
lines a reversible embedding of $f$ must add is
$v=\lceil n-\Hinf(Y)\rceil$, Maslov's bound~\cite{maslov2003} written in
entropy form, and since $\Hinf\le H_1$ the ancilla a reversible embedding
must add is at least the information the irreversible circuit would have
destroyed.  Reading the mechanisms as different R\'enyi entropy orders of one object
explains why optimizing one does not automatically buy the others, and it
tells the tool which ledger a given saving belongs in.  Whether the
single-shot thermodynamic correspondence, in which min-entropy governs
work extraction for a device evaluated once, maps exactly onto circuit
quantities is not established here.

\subsection{Relevant Switching Theory Concepts}
\label{sec:switching-bg}

A switching function $f:\F^n\rightarrow\F$ has a unique representation as
a polynomial over $\F$, its \emph{algebraic normal form} (ANF), also
known as the positive-polarity Reed--Muller form,
\[
f(x_1,\ldots,x_n)\;=\;\bigoplus_{S\subseteq\{1,\ldots,n\}} a_S
\prod_{i\in S} x_i, \qquad a_S\in\F,
\]
obtained from the truth table by the M\"obius transform over
$\F$~\cite{steinbach2016chapter,thornton_rm_chapter1}.  The \emph{degree}
of $f$ is the largest $|S|$ with $a_S=1$.  Degree-one functions are
exactly the affine functions.  Fixing a polarity for each variable and
expanding accordingly gives a \emph{fixed-polarity Reed--Muller} (FPRM)
form, one per polarity vector, each likewise canonical.  Relaxing the
requirement that all occurrences of a variable share a polarity gives an
\emph{exclusive-or sum of products} (ESOP), which is typically more
compact but no longer canonical, and is therefore a target for
minimization rather than a normal form to read structure
from~\cite{thornton_rm_chapter1,houngninou_rm_chapter2}.  The conversion
between truth-table and ANF representations is well known in the
theoretical switching-theory and cryptographic communities but appears
less often in the low-power literature, so we state the procedure the
tool uses here, although a number of different means
for finding the FPRM transform could have been used.  In \tool~all 
ANF coefficients are computed by the fast M\"obius
transform, the XOR butterfly of Algorithm~\ref{alg:mobius}, in $O(n2^n)$
coefficient operations.  Over $\F$ the transform is an involution, so the
same procedure converts in both directions.\footnote{$\mathbb{B}=\{0,1\}$ is the more
customary notation for $\mathbb{F}_2$ in the switching theory literature.}

\begin{algorithm}[t]
\caption{In-place M\"obius (positive-polarity Reed--Muller) transform
of a truth vector over $\F$}
\label{alg:mobius}
\begin{algorithmic}[1]
\Require $a[0\,..\,2^n{-}1]$ with $a[x]=f(x)$, the index $x$ read as
the assignment $(x_{n-1}\cdots x_1x_0)_2$
\Ensure $a[s]$ is the ANF coefficient of the monomial
$\prod_{i\in S}x_i$ with $S=\{\,i : \text{bit } i \text{ of } s = 1\,\}$
\For{$i \gets 0$ \textbf{to} $n-1$}
  \For{$x \gets 0$ \textbf{to} $2^n-1$}
    \If{bit $i$ of $x$ equals $1$}
      \State $a[x] \gets a[x] \oplus a[x \oplus 2^i]$
    \EndIf
  \EndFor
\EndFor
\end{algorithmic}
\end{algorithm}

One distinction that we emphasize is that the ANF is defined 
over $\F$ and is what the tool manipulates whereas the spectrum 
in which switching is bounded is the \emph{bipolar} (Walsh) spectrum 
of the $\pm1$-valued form, a different transform of the same function 
over $\mathbb{R}$, and the two are not interchangeable.  
What links these two transformations is degree, and the link can be made
quantitative.  Under temporally correlated inputs in which each input
can independently invert with probability $\epsilon$, the activity of a signal
$h$ equals its Boolean noise sensitivity $\mathrm{NS}_\epsilon(h)$, that can
be expressed as a
weighted sum of squared Walsh coefficients, and the following bound
holds.

\begin{lemma}[Low-degree spectral concentration bounds switching]
\label{lem:degree}
Suppose the bipolar form of $h$ places at least a $(1-\delta)$ fraction
of its nonconstant spectral energy on coefficient sets of size at most
$d$.  Then
\[
\mathrm{NS}_\epsilon(h)\;\leq\;
\tfrac{1}{2}\big[1-(1-2\epsilon)^d\big](1-\delta)\;+\;\tfrac{\delta}{2}.
\]
\end{lemma}
\begin{proof}
The standard Fourier expression for noise sensitivity is
$\mathrm{NS}_\epsilon(h)=\tfrac12\sum_{S\neq\varnothing}
\big[1-(1-2\epsilon)^{|S|}\big]\widehat h(S)^2$, where $\widehat h(S)$
are the coefficients of the bipolar form and the nonconstant spectral
energy $\sum_{S\ne\varnothing}\widehat h(S)^2$ is at most $1$ by
Parseval's Theorem.  Consider splitting the sum at $|S|=d$.  For $|S|\le d$ the bracket is at
most $1-(1-2\epsilon)^d$ and the corresponding spectral energy is at most
$1-\delta$ of the nonconstant total. Likewise, for $|S|>d$ the bracket is at most
$1$ and hence each tail term is at most $\tfrac12\widehat h(S)^2$, and
the tail energy is at most $\delta$.  Adding the two parts back together yields the
stated bound.
\end{proof}

These two lemmas together are what justify using ANF degree as a
\emph{selection} criterion when choosing subcircuits for re-encoding. A
region whose local function has degree one is affine and can be absorbed
exactly into a linear stage, and a region of low degree concentrates its
spectral energy at low order and is therefore, by Lemma~\ref{lem:degree},
a region whose switching behavior a linear change of coordinates has
genuine room to improve.  We do not wish to overstate this guidance.
The ANF degree is a triage heuristic, and in this flow every candidate so
nominated is ultimately accepted or rejected based upon the measured energy of
its complete realization and not based upon its spectral credentials alone.

The degree also controls which \emph{exact} method is available for a
function's output distribution, and therefore for the entropies of
Section~\ref{sec:renyi-bg}, which is what allows a netlist-native tool to
avoid enumerating $2^n$ inputs.  When the ANF has degree one, the output
is uniform on a coset of the image, every attained value has the same
multiplicity, and $\Hinf=H_2=H_1=\rho$ exactly, with $\rho$ the rank of
the linear part. A rank computation settles the whole entropy family.
At degree two the orders separate, and the collision probability is fixed
by the symplectic rank of the associated bilinear form together with the
Arf invariant~\cite{arf1941untersuchungen}, giving $H_2$ exactly by the Dickson--Arf classification
rather than by enumeration.  This is the connection between the classical
algebra of quadratic forms over $\F$ and an energy quantity that we have
found most useful in practice.  At degree three an Arf-signed directional
recursion over the quadratic case remains exact but no longer scales in
$n$, and at any degree with bounded control rank a factorization is
exact.  Collision entropy is the tractable object because $\sum_y p_y^2$ is
a quadratic functional of the distribution, precisely what Parseval over
the output space and the Arf classification can evaluate in closed form.
The min-entropy the width bound uses is then recovered through the
bracket $n-H_2\le\log_2 N_{\text{dup}}\le n-H_2/2$.  The R\'enyi order
fixes which entropy answers the question, and the ANF degree fixes which
algebra can compute it without enumeration.

\subsection{Synthesis and EDA Concepts}
\label{sec:eda-bg}

The tool is built from the standard machinery of logic synthesis, and we
review the pieces on which later sections depend.  A \emph{cut} of a node
in a gate netlist is a set of nets that separates the node from the
primary inputs.  A cut is $K$-feasible when it has at most $K$ nets, and
the function of the \emph{cone} between a node and its cut is the local
function the mapper must realize.  Cut enumeration proceeds bottom-up,
the cuts of a node being unions of its fan-ins' cuts subject to the size
bound, and a \emph{cover} selects one cone per needed node so that every
primary output is computed.  Covers are priced, and the priced cover
problem with depth terms descends from FlowMap and its
successors~\cite{cong1994flowmap}.  \emph{Reconvergence}, the
same signal reaching a gate along two paths, is the structural fact that
defeats per-net independence assumptions in probability propagation and
creates the shared structure that window-based re-synthesis exploits.
\emph{roots} and fanout-free cones are the units re-synthesis rewrites.
Binary decision diagrams~\cite{bryant1986bdd} provide canonical
function representations under a variable order and, in this tool, an
alternative realization style in which a shared diagram vertex becomes a
real gate with fanout.  Boolean satisfiability enters twice, in
equivalence checking and in the optimality-gap program of
Section~\ref{sec:certified}, where energy-optimal covering is reduced to
weighted MaxSAT.

Because the \tool~flow also computes justification-style quantities, one
contrast is worth fixing here.  A SAT formulation asks whether an
assignment exists and returns one.  It is complete, and it answers one
question per solver call.  Lattice-based justification, described next,
alternatively propagates the consequences of a partial specification through
every element at once, in one pass, returning the strongest partial
assignment the structure forces.  It is weaker, since it will not
manufacture a witness that requires case splitting, and it is uniform in
the netlist, which is the property a sweep needs.  \tool~uses SAT
where completeness is required and lattice value propagation through 
a traversal of the VSIM structure where coverage of
every line in one pass is required.

Finally, a methodological remark that shapes the whole design.  Synthesis
tools are usually built on a single conceptual framework and pursue it to
the boundary of its usefulness.  A SAT-based flow reduces its questions to
satisfiability, a BDD-based flow to diagram manipulation, an
information-theoretic treatment to entropy inequalities.  There is real
virtue in that, and accordingly a real cost, which is that every problem is
eventually attacked with the tool at hand rather than the tool that
fits.  This work deliberately avoids this convention.  Information-theoretic
reasoning is used where it is most appropriate, namely to identify which
physical mechanism a cost belongs to, to bound what re-encoding can
achieve, and to place a floor under erasure.  Conventional EDA machinery,
cut enumeration, technology mapping, structural rewriting, and
equivalence checking, is used everywhere else, because it is mature,
fast, and correct, and because an entropy argument is not a substitute
for a cover.  The seam between the two is the acceptance rule of
Section~\ref{sec:ledger}.  No candidate is admitted solely on a spectral or
entropic credential.  Each is priced by realizing it and computing
what the energy model predicts.  We consider that seam to be as much of a
methodological contribution as any individual pass or algorithm in \tool.

\subsection{The Vector-Space Information Model (VSIM)}
\label{sec:vsim-bg}

The netlist model used throughout this paper, together with the
transfer-function theory that makes both the forward probability sweep
and the backward observability machinery well defined, is the VSIM
formulation developed and described in~\cite{thornton_vsim_book}~\cite{steinbach2016chapter}~\cite{thornton2015tc}.  
A switching network is
modeled as a linear transformation on a vector space of signal states.
Each structural element is modeled by a transfer matrix. Simulation is
matrix--vector multiplication in the forward direction and justification
is the corresponding action in the reverse direction.  We use and implement VSIM
theory as the main data structure in \tool~for these capabilities.

The VSIM value domain is a four-valued lattice rather than a bit or bit word.  
A datum is a row vector over a
two-dimensional space, $\langle 0| = [0\ 1]$ and $\langle 1| = [1\ 0]$.
Two further vectors complete the domain.  The \emph{total} vector
$\langle t| = [1\ 1]$, which covers both values, and the \emph{null}
vector $\langle\emptyset| = [0\ 0]$, which covers neither.  Adjoining
them yields a complete four-valued lattice~\cite[Ch.~2]{thornton_vsim_book},
and this is not a modeling convenience.  An unknown is a genuine element
of the domain rather than a flag, a contradiction is a genuine element
rather than an error return, and every per-gate transfer relation is
monotone on the lattice, which is what allows one traversal framework to
carry information in either direction without separate don't-care
machinery.  The VSIM formulation is advantageous compared to 
two-valued switching algebra for the purpose of implementing~\tool.  
Two-valued algebra is total in the sense that every net has a definite value
under every assignment, so a question of the form ``does any output
depend on this line'' must be answered by exhaustive or symbolic case
analysis.  The lattice carries partial information as one of the
values, and because the transfer maps are monotone, a single backward pass
propagates a constraint through the network without enumerating the
assignments that witness it.  In essence, the VSIM model enables both forward and reverse
netlist traversals as efficient methods that are linear in the number of nets $N$
and it enables an efficient means for populating and updating tag values.  A forward- or 
simulation-sweep is accomplished by transforming logic values as vector-matrix
products at each structural element using a vector comprised of the outer product
of the element's lattice values with its transfer matrix $T$.

The backward operator is the transfer function transpose.  Let $T$ be the transfer matrix
of a switching network.  Recovering the inputs consistent with a required
output is in general a pseudo-inverse computation.  The formulation
removes that cost in two steps.  The Gramian
$\mathrm{gram}(T)=T^{\mathsf{T}}T$ is diagonal, so its inverse is
immediate, and the pseudo-inverse can then be dispensed with altogether
in favor of the transpose, giving the justification matrix
\begin{equation}
  T^{J} \;=\; T^{\mathsf{T}}
  \label{eq:justification}
\end{equation}
\cite[Def.~4.11]{thornton_vsim_book}.  The cost of use the transpose rather than
the pseudo-inverse of the transfer function is that scaling
coefficients are lost, so one recovers the \emph{set} of justifying
inputs and not their distribution, which is exactly the information a
synthesis flow needs and more than a simulator provides.  We draw
attention to \eqref{eq:justification} because its consequence for a tool
is the opposite of what it first appears.  A reader may expect a flow
built on this formulation to contain an elaborate backward pass.  What
the result says is that it need not.  The backward operator is the forward
one read in the other direction, so a quantity defined by a backward
question can often be obtained by a forward computation.
Section~\ref{sec:backward} reports where this flow does exactly that, and
is precise about which backward machinery the shipped flow executes and
which it does not.  
\section{The \tool\ Architecture}
\label{sec:renesis}

This section describes the tool from the point of view of an EDA
algorithm designer, stage by stage.  The formal results are placed next
to the mechanisms they support rather than collected in an appendix, so
that the mathematical justification for each pipeline stage appears
where the stage is described.  Each subsection also closes with a worked
example on a small circuit.  The examples are drawn from the tool's own
output and every number quoted in them comes from a run record produced
by a single command that a reader can repeat.

The example circuits are small on purpose and they are not all the same
circuit.  A benchmark large enough to exercise every stage is too large
to draw.  A benchmark small enough to draw exercises almost nothing.
Each stage is therefore illustrated on the smallest fixture that makes
its mechanism visible.  Every fixture except one ships with the tool,
and the text identifies the exception where it appears.

\IfFileExists{figures/fig_pipeline_v5.tex}{\input{figures/fig_pipeline_v5}}{}

A default invocation executes the following sequence.  The netlist is
parsed, then optionally prepared structurally, then passed through the
optional re-synthesis passes, then through the tag sweep, then through
the optional linear pre-filter.  It is then covered and technology
mapped, verified, buffered, and verified a second time.  The energy
report runs twice and the run record is written last.  Both verification
steps are equivalence checks on the mapped network and both run whether
or not anything was optimized.  The energy report runs twice because two
figures are reported, one before buffer insertion and one after
(Section~\ref{sec:ledger}).

Two orderings in that sequence are deliberate and neither is obvious.
The re-synthesis passes run before the tag sweep because they change the
netlist the tags would describe.  The linear pre-filter runs after the
tag sweep and replaces the mapping stage rather than feeding it, for the
reason given in Section~\ref{sec:bdec}.

\subsection{Front End and Structural Preparation}
\label{sec:frontend}

The front end accepts structural Verilog (an ISCAS-85 subset),
\texttt{.isc}, BENCH, PLA, BLIF, and AIGER in both binary and ASCII
forms, and normalizes them to a common gate netlist.  The internal
representation is deliberately plain.  A gate carries an output name, a
function drawn from the usual primitive set, and an input list.  A net
with no driver is a primary input by definition rather than by
declaration, which makes partially specified inputs behave predictably.
Cone and window functions are carried as truth tables packed into
machine words, and GF(2) matrices as row bitmasks.  Both representations
are exact and both are exponential in support, so every stage that uses
them carries an explicit width cap.

An optional preparation pass performs structural hashing, associativity
rebalancing, and cut rewriting, iterated until the gate count stops
decreasing.  It usually makes the whole run faster rather than slower,
because it hands the mapper less work.  It is nonetheless off by
default.  It changes the input netlist rather than the optimization, and
leaving it on would make every reported ratio a statement about two
changes at once.  The reported baseline is therefore the netlist as
supplied.

Determinism is a structural property of this stage rather than a
convention.  No set or dictionary iteration order reaches the output.
Every ordering is either topological or content-sorted.  That property
is also what allows the independent C implementation of
Section~\ref{sec:development} to be byte-compatible with the reference
implementation.

\IfFileExists{figures/ex_frontend.tex}{\input{figures/ex_frontend}}{}

Figure~\ref{fig:ex-frontend} shows this stage on a six-gate fixture
written to contain the two redundancies structural hashing removes, a
duplicated conjunction and a double negation.  Preparation merges the
duplicate by hash-consing on the function together with the
content-sorted input tuple, collapses the negation pair to a wire, and
returns three gates.

The instructive part is what happens next.  Mapped and priced, the
prepared and unprepared netlists are identical.  Both give twelve
devices and $0.008228$\,pJ on both tables, because the cover of
Section~\ref{sec:cover} absorbs the redundancy on its own.  This is the
concrete reason preparation is off by default.  Its usual benefit is
that the mapper is passed less work, not that the resulting circuit is
less costly.

\subsection{The forward sweep: signal probability and activity}
\label{sec:forward}

Each net carries a pair of tags.  The first, $p_1$, is the probability
the net holds one.  The second, $\alpha$, is the probability it differs
from its own previous value.  Signal probability as a synthesis quantity
is long established~\cite[Sec.~2.3]{thornton2001spectral}.  What this
subsection reports is how much it matters here.

The switching model prices a gate with $k$ controls at $k+1$ capacitance
units, contributed when it fires.  Assuming independent uniform controls
gives a firing probability of $2^{-k}$.  That assumption is wrong by
$26$ to $28\%$ against simulation of the emitted circuits, in both
directions, so it cannot be absorbed as a calibration constant.  Tagging
each gate with its measured firing probability reduces the error to
$1\%$.

Three propagation modes exist.  The analytic mode propagates per-net
marginals under an independence assumption at every gate.  It is
inexpensive and it is wrong in a specific, locatable place, namely at
reconvergent stems, where the same signal reaches a gate by two paths
and the assumption double-counts.  On \texttt{c17}, a six-net circuit
with a single reconvergent stem, it is off by $7.7\%$ on one net.

The joint mode carries the exact joint distribution over the live
frontier.  It is the forward pass of the VSIM machinery applied to
distributions rather than to values, so a reconvergent stem appears once
and its correlation survives.  Its cost is exponential in frontier
width.  When the frontier would overflow, the live net whose last use is
furthest away is factored out and treated as independent, and the
total-variation distance discarded is charged to an explicit budget.
Every reported probability therefore carries a bound rather than a
silent assumption.

The sampling mode simulates a fixed number of random vectors at a fixed
seed, and is the mode whose output feeds the cover.  Accuracy was
checked against an independent simulator rather than against itself.
The netlist is emitted as structural Verilog with a generated testbench,
run under an external simulator, and the value-change dump is parsed for
per-net statistics.  Any disagreement is treated as a defect in the
internal simulator rather than in the tags.

One scoping fact governs how much a workload can change.  The
technology-priced cover, which is the shipped default, takes no tags
argument at all.  The activity figure a default run reports comes from
the energy model's own sweep, and the trial-count and seed options are
inert on a default run.  They act only under the activity-priced cover.
Claims about workload-aware synthesis in this paper are scoped
accordingly (Section~\ref{sec:drive}).

\IfFileExists{figures/ex_c17.tex}{\input{figures/ex_c17}}{}

The running example for this stage and for the four that follow it is
\texttt{c17}, the smallest ISCAS-85 member, with five inputs, two
outputs, and six NAND gates.  It is small enough that ``measured'' can
mean exact, since enumerating all $32$ input vectors is less expensive
than sampling them.

Figure~\ref{fig:ex-tags} carries the result.  Each net is labelled with
its exact signal probability and with the value the per-gate
independence rule predicts.  On four of the six nets the two agree to
every printed digit, which is what the fan-in structure of those gates
predicts.  On the other two they disagree, and those two are the gates
that sit downstream of a reconvergent stem, net~$3$ reaching net~$22$
along two paths and net~$11$ reaching net~$23$ along two paths.  The
errors are $+5.56\%$ and $-8.33\%$.

Two things follow, and both are visible in one figure on a six-gate
circuit.  The independence assumption fails at reconvergence and not
elsewhere, so its error is structural rather than diffuse.  The two
errors also carry opposite signs, so no calibration constant absorbs
them.  That is the case for measuring the tags rather than propagating
them analytically, and it is why the sampling mode is the one whose
output reaches the cover.

\subsection{VSIM Backward Sweeps in~\tool}
\label{sec:backward}

Equation~\eqref{eq:justification} states that the backward operator is
the transpose of the forward one.  We report what that buys, and we
separate three things that are easily conflated, because the distinction
is what makes the claims checkable.

The first is a justification quantity obtained forward.  The reversible
embedding width is $v=\lceil\log_2 N_{\mathrm{dup}}\rceil$ with
$N_{\mathrm{dup}}=\max_y|f^{-1}(y)|=2^n\max_y p(y)$, that is,
$n-\Hinf(Y)$.  This is Maslov's minimum-garbage
bound~\cite{maslov2003} written in entropy form.  It is defined by a
preimage, meaning the set of inputs that produce a given output, and is
therefore a justification question of exactly the kind
\eqref{eq:justification} addresses.  Because the Gramian is diagonal,
its diagonal is a vector of per-output marginals, and those are
obtainable by bit-parallel \emph{forward} simulation.  The flow computes
the bound that way and does not form $T^{\mathsf{T}}$ at all.  The
theory here earns its keep in the most direct way available.  It shows
that the transpose is redundant, and the flow cashes that in by not
building it.

The second is a backward traversal that is load-bearing but is not
justification.  The cover is extracted by walking back from the primary
outputs over the chosen cuts, and the fanout counts recovered by that
walk feed the duplication discount in the second pricing pass.  Without
that walk there is no cover.  Dead-block elimination, run to fixpoint
afterwards, is the same kind of step.  Both are reverse reachability
over a graph, of the sort any cut-based mapper performs, and neither
propagates a value or inverts a relation.  We name them because a reader
who has just been shown \eqref{eq:justification} could reasonably read a
backward walk as a netlist justification computation, and it is not.

The third is observability, which is measured and not priced.  A firing
is unobservable when suppressing it leaves every primary output
unchanged, so the energy is spent on a value no output depends on.  The
estimator is a sampling procedure that flips a net, re-evaluates the
downstream cone, and compares the outputs.  It is not the justification
pass, and we say so because the name invites the other reading.

The measurement is substantial.  On \texttt{c432} and \texttt{c880}
respectively, $49.4\%$ and $39.1\%$ of switched capacitance is spent on
values not observable at any primary output.  Of that, $90.0\%$ and
$67.9\%$ is locally detectable, meaning every later gate reading the
line fails to fire, so the value does not propagate and the condition is
testable from the immediate consumers rather than by global don't-care
analysis.  In total energy the unobservable share is $37.5\%$ and
$31.2\%$.

In~\tool~this is an upper bound rather than a harvested
saving, and it is not harvested at all.  Exploiting it requires gating a
block on its observability condition.  The gating construct we built
produces a both-rails-low condition whose propagation through a clocked
dual-rail cascade needs a family-specific discipline we did not
establish.  We report the fraction as headroom.  On a reversible-width
target this switching would not be waste, since an unobservable
intermediate must still be un-computed there.  The same computation is
necessary under one target and removable under the other, which is a
category difference rather than a trade-off.

\subsection{Covering and its Usage in~\tool}
\label{sec:cover}

Cut enumeration is the standard bottom-up construction over the prepared
netlist, at twelve-input cuts with thirty-two retained per node.  Three
details are load-bearing and each was installed after a measurement.

Cuts are ordered by content, size then sorted tuple, rather than by
set-iteration order.  Hash-seed-dependent ordering here made covers, and
every number downstream of them, vary between runs.  Retention is
smallest-first, which corrected three failures at once.  Largest-first
retention dropped small inexpensive cuts as the limit filled, starved
wide-fanin nodes down to the trivial cut, and broke dominance pruning,
since a superset kept early is not removed when its subset arrives
later.  Under smallest-first ordering, dominance pruning is exact,
because every subset is seen before any of its supersets.

Two priced covers exist.  The default prices a cut technologically,
using a constant per block, a weight on device count, a weight on
arrival depth, a weight on the block's charged internal load, and a
duplication discount that divides a leaf's accumulated cost by its
fanout.  The shipped weights put zero weight on raw device count and the
heaviest weight on charged internal load, which is the switching term at
the heart of the objective.  The alternative cover prices switching
activity directly and is the only consumer of the forward tags.  Which
of the two is the right default is an empirical question, and
Section~\ref{sec:results} answers it.

The quantity both covers ultimately serve is the following, and its
formal properties govern how every result in this paper should be read.
The per-cycle switched capacitance of a mapped dual-rail network is
\begin{equation}
\label{eq:cost}
E \;=\; c_{\mathrm{dev}}\sum_{g\in\mathcal{G}}\mathrm{load}(g),
\end{equation}
where $\mathcal{G}$ is the set of mapped gates and $\mathrm{load}(g)$
counts the literal occurrences of $g$'s output at its readers, together
with pad load and the family's self-load.

\begin{proposition}[$E$ is a literal-occurrence functional]
\label{prop:cost}
$E$ in \eqref{eq:cost} is determined by the multiset of literal
occurrences of the mapped network, together with the fixed per-family
constants.  It is not a function of $|\mathcal{G}|$, of the device
count, or of any other pure size measure.
\end{proposition}

\begin{proof}
The first claim is read off the definition.  Every term of the sum is a
count of literal occurrences at readers, plus pad and self-load terms
that are fixed by the family record and the output boundary, so two
networks with the same occurrence multiset and boundary receive the same
$E$.  For the second claim it suffices to exhibit two networks computing
the same function whose size measures agree and whose $E$ differs, or
the reverse.  The measured witnesses of
Corollary~\ref{cor:nonmonotone} provide both directions, so no function
of size alone can determine $E$.
\end{proof}

\begin{corollary}[Device count is not monotone in energy]
\label{cor:nonmonotone}
There exist pairs of networks computing the same function for which
device count and $E$ move in opposite directions.
\end{corollary}

\begin{proof}
By exhibition, with measured rather than constructed witnesses.  On
\texttt{c880} an accepted re-encoding raises the device count from
$1202$ to $1348$, an increase of $12.1\%$, while lowering $E$ from
$0.839256$ to $0.818686$\,pJ.  In the other direction, on
\texttt{c1355} an arm of the affine pass reaches $201$ gates from $518$,
and maps to $2868$ devices against the original's $2576$.  That is fewer
gates, more devices, and higher energy than a different arm with more
gates.  Both pairs compute their original functions, as verified by the
flow's equivalence check.
\end{proof}

These are consequences of Proposition~\ref{prop:cost} rather than
anomalies.  A rewrite that removes gates while raising fanout raises the
literal-occurrence count that \eqref{eq:cost} bills.  This is the reason
the flow does not select a candidate on area or gate count, and it is
why we regard the common practice of using area as a proxy for dynamic
power in pass-transistor styles as unsafe rather than merely
approximate.

After covering, a fanout-one absorption pass folds each single-consumer
block into its consumer, removing a mapped gate and a charged internal
net.  It hill-climbs on the real objective rather than on a proxy,
because the two were measured to disagree in sign on one circuit.  It
takes the first improvement rather than the best, because
best-improvement re-pricing was measured at $3.4\times$ and $14.9\times$
baseline cost on two circuits.  The pass typically improves the uncapped
table while regressing the capped one, busting the series bound, so the
acceptance gate rejects it.  The run's gate receipts
(Section~\ref{sec:ledger}) record that outcome explicitly, which
distinguishes it from an option that had no effect.

\IfFileExists{figures/ex_cover.tex}{\input{figures/ex_cover}}{}

Figure~\ref{fig:ex-cover} shows what the cover does with the six gates
of \texttt{c17}.  Cut enumeration offers each node its $K$-feasible
cuts.  The priced cover selects two cones, one per primary output, with
cuts $\{1,2,3,6\}$ and $\{2,3,6,7\}$.  No internal net survives as a
block of its own, and all four internal NAND gates are absorbed.

The decision worth pausing on is net~$16$, which lies inside both cones.
The cover could have kept it as a separate block, computed once and read
twice.  Instead it duplicated the logic into both blocks.  That is the
duplication discount at work, dividing a shared leaf's accumulated cost
by its fanout, and the choice was priced rather than assumed.  On this
circuit duplication wins because a mapped block carries a charged output
net that two pass devices would have to charge, while the duplicated
logic charges only the leaves it reads.

The lower half of the figure is the tool's own output for the two
blocks, as series-parallel conduction expressions over signed literals.
Reading the first, the positive rail of net~$22$ conducts when
$(1\wedge3)$ or when $(2\wedge(\bar3\vee\bar6))$, which is
$\overline{10\wedge16}$ pushed through DeMorgan's relation, and the negative rail
is its dual.  Counting literal occurrences in the two expressions gives
ten pass devices in block~$22$ and twelve in block~$23$.  Those are the
twenty-two the energy model bills, arrived at without counting a gate.

\subsection{Technology Mapping and Realization Choice}
\label{sec:map}

A mapped block is a pair of series-parallel conduction trees over signed
literals, one per rail.  The construction is due to DeMorgan's relation.  Conjunction
becomes series on the positive rail and parallel on the negative,
disjunction the reverse, parity becomes the familiar pair of two-branch
parallel networks, and negation is a rail exchange costing nothing.
Wide parity is folded pairwise, each fold reading each input twice per
rail, which is why a parity consumer contributes four literal
occurrences rather than two.  The DeMorgan dual does not preserve
depth, since series sums and parallel maxes, so the two rails of one
gate can have very different series depth.  That single fact drives the
bounding pass of Section~\ref{sec:cap}.

Under the default route the mapper builds more than one candidate
realization per block and keeps the less expensive option.  The
candidates are the structural form, a small-support exact form where the
support permits it, and a shared multi-output decision-diagram form in
the manner of Lindgren \emph{et al.}~\cite{lindgren_lowpower_bdd}.  In
that third form the whole technology cover is built as a single shared
multi-output BDD forest and realized directly as a dual-rail
transmission-gate multiplexer network, so a diagram vertex shared
between several outputs becomes one real gate with fanout rather than
being duplicated per cone.

In keeping with the paper's premise the shared-forest candidate is
driven by switching energy rather than by area.  Two variable orders are
built, a node-count sift and a switching-probability sift minimizing the
vertex sum $\sum_v 2p_v(1-p_v)$ over the forest, with each vertex
probability computed exactly bottom-up from the forward tags.  The
lower-energy realization is kept.  The candidate is admitted only when
it strictly improves the capped table and leaves the uncapped table no
worse, so enabling it does not regress a published number.  Realization
arbitration prices candidates \emph{after} series bounding rather than
before.  That ordering was installed after a defect was found in which
the ranking inverts across the bound on one circuit, so the arbitration
had been choosing on a metric that is not reported.

This matters for interpreting every later result.  
\tool~already shares identical diagram vertices, so a subsequent pass that
finds sharing is competing against sharing that has already happened,
and any new transformation must be priced with this candidate enabled
rather than against a naive structural mapping.

\IfFileExists{figures/ex_map.tex}{\input{figures/ex_map}}{}

Figure~\ref{fig:ex-map} is block~$22$ of \texttt{c17} as the mapper
realizes it, drawn by the tool from the mapped netlist rather than by
hand.  The structure is the conduction expression of the previous
subsection.  The positive rail is a parallel pair of series branches,
the negative rail is the De~Morgan dual, series of parallel, and each of
the ten devices is a transmission gate driven by a named rail.

This figure is where the cost model of Section~\ref{sec:ledger} becomes
tangible.  The quantity billed is the number of gate terminals that must
be charged, which is the number of literal occurrences in the drawing.
It is visibly not the number of gates.  One mapped gate here carries ten
devices, and a different cover of the same function would carry a
different number with the same gate count.  On a circuit this small the
structural form wins, and the shared-diagram candidate is built, priced,
and declined.

\subsection{Realizability: The Series Bound}
\label{sec:cap}

A series chain of $d$ pass devices has resistance of order
$d\,R_{\mathrm{on}}$ and degrades the signal it passes, so a real
process bounds $d$.  A post-mapping pass bounds every chain to a
user-specified depth by inserting buffered stages.  Let a rail network
be a series-parallel tree $t$ over signed literals, with series depth  
$\delta(\mathrm{lit})=1$,
$\delta(\mathrm{ser}\,t_1\ldots t_k)=\sum_i\delta(t_i)$, and
$\delta(\mathrm{par}\,t_1\ldots t_k)=\max_i\delta(t_i)$.

\begin{proposition}[Series bounding]
\label{prop:cap}
For every bound $c\ge1$ and every rail network $t$, the bounding
construction below returns a network $t'$ with $\delta(t')\le c$
computing the same conduction function, together with a set of
extracted stages, and the iterated pass terminates.
\end{proposition}

\begin{proof}
The construction proceeds by structural recursion with an accumulator.
A parallel node recurses on its branches independently, which preserves
the conduction function branch-wise and cannot increase depth beyond the
branch maxima.  A series node is partitioned into consecutive segments,
each composed through an accumulator literal on the previous segment's
output.  The composed segment conducts exactly when the previous stage
conducted and the segment conducts, so the conjunction of segment
functions equals the original series function, and every segment
carrying the accumulator has budget $c-1$.  A child too deep to fit the
remaining budget is not recursed into with a reduced bound.  It is
extracted whole and replaced by a literal on its own new stage, which is
legal because a child is itself a conduction subnetwork and buffering
its output yields a literal of depth one computing the same condition.
Each extracted segment becomes a real gate on both rails, its opposite
rail the DeMorgan complement of the extracted subnetwork, so the
rail-consistency invariant, that the two rails of every gate compute
complementary conduction functions, is preserved.  The post-condition
$\delta\le c$ holds at every return by induction on the tree.  Literals
have depth one, parallel nodes return the maximum of bounded branches,
and series nodes return segments each of which was constructed within
its budget.  For termination of the iterated pass, note that because the
DeMorgan dual exchanges series and parallel, a network satisfying the
bound can have a dual that does not, so each round bounds both rails of
every gate and repeats.  Each round either changes nothing, in which
case the pass stops, or strictly reduces the multiset of rail depths
exceeding the bound in the multiset order, since every rewritten rail
ends at depth at most $c$ and inserted stages have depth at most $c$ by
construction.  A strictly decreasing sequence in a well-founded order is
finite, so the iteration terminates.
\end{proof}

We treat the bound as a design rule rather than as a cost term.  The
energy model has no delay term, so it receives an inserted stage's cost,
a new charged internal net, and none of its benefit.  Were the bound
weighted into the cover and then relaxed, modeled energy would improve,
so an optimizer given the choice would remove every buffer.  The bound
is therefore enforced after mapping, and both the unbounded and bounded
figures are reported in every table.  Reporting one number alone would
force a choice between a circuit that may not be realizable and a
circuit that is not comparable with prior work imposing no bound.

A pass-transistor chain is an RC ladder, so delay grows roughly
quadratically in chain length.  With current technology this permits two
or three transmission gates in series before a restoring buffer.  The
value of six used in the validation runs reported here is therefore not
a realistic design rule.  A more practical bound of three raises the
charged term by over an order of magnitude on two of the development
benchmarks, and natural depth is seven to ten in most of the set.  Six
is an engineering compromise rather than a technology-justified limit.
Part of the modeled advantage is purchased with series depth that a real
design rule might not permit, and we state that here so the tables are
read accordingly.  Both columns of energy estimates in those tables are
bounded by the same pass with the same parameter.

\IfFileExists{figures/ex_cap.tex}{\input{figures/ex_cap}}{}

Figure~\ref{fig:ex-cap} illustrates the bounding construction of
Proposition~\ref{prop:cap} on a chain of eight series devices with the
bound set to three.  The chain is partitioned into consecutive segments,
and each segment after the first is composed through an accumulator
literal, the buffered output of the previous segment, so that a segment
conducts exactly when the previous stage conducted and this segment
conducts.  Every segment that carries an accumulator therefore has
budget $c-1$ rather than $c$, which is why eight devices split three,
two, three rather than three, three, two.  Each extracted segment
becomes a real dual-rail gate whose opposite rail is the De~Morgan
complement, so the rail-consistency invariant the verifier checks
survives insertion.

The cost of the bound is visible on the fixture of
Section~\ref{sec:davio}, where imposing it takes the billed device count
from $42$ to $56$.  That is the bound doing its job.  It makes a circuit
realizable rather than less costly, which is why it is enforced after
mapping and not weighted into the cover, and why both figures are
reported everywhere in this paper.

\subsection{The Ledger: Tables, Receipts, and the Never-Regress Rule}
\label{sec:ledger}

Every candidate circuit reported in the synthesis flow is reduced to two
numbers.  $T_1$ is the per-cycle switched-capacitance energy of the
mapped network before buffer insertion.  $T_2$ is the same figure after
it.  Both are taken only after the mapped network has been validated, so
a candidate that cannot be validated raises an error rather than being
priced.

Two tables are reported rather than one.  The series cap makes a circuit
realizable rather than less costly, so there is no defensible exchange
rate between the two figures, and a scalarization would admit moves that
trade one against the other.

Acceptance is dominance on the pair.  Call a move $M$ \emph{admissible}
at state $s$ when $M(s)$ is equivalent to the original netlist and
\[
T_1(M(s))\le T_1(s),\qquad T_2(M(s))\le T_2(s),
\]
with at least one inequality strict.  Equivalence is checked against the
\emph{original} netlist rather than against the current incumbent.  The
check is exhaustive at ten support variables or fewer and by sampling
above that limit.  Checking against incumbents would let a slow drift
accumulate.  Checking against the original netlist means a malformed
candidate cannot be accepted at any depth of pass composition.

\begin{proposition}[Dominance]
\label{prop:dominance}
Any finite sequence of admissible moves carries a netlist to a state
that dominates the input in both tables, with at least one strict
improvement if the sequence is nonempty.  In particular, the flow cannot
return a circuit worse than its input in either reported figure.
\end{proposition}

\begin{proof}
By induction on the sequence length.  The empty sequence dominates
trivially.  If $s_k$ is reached from $s_0$ by $k$ admissible moves and
satisfies $T_i(s_k)\le T_i(s_0)$ for $i\in\{1,2\}$, then an admissible
move at $s_k$ yields $s_{k+1}$ with $T_i(s_{k+1})\le T_i(s_k)\le
T_i(s_0)$ by transitivity of $\le$ in each coordinate, and the strict
inequality granted by admissibility of the first move is preserved by
the subsequent non-strict ones.  Equivalence to the original is a
conjunct of admissibility at every step, so the final netlist computes
the original function.
\end{proof}

\begin{proposition}[Non-confluence]
\label{prop:nonconfluent}
The admissible-move relation is not confluent.  There exist states from
which different orders of admissible moves terminate at incomparable
local optima.
\end{proposition}

\begin{proof}
We exhibit a state with two admissible moves whose results are both
terminal and incomparable in the product order.  Since the two results
are distinct and neither admits any further move, they have no common
descendant, which is the negation of confluence.

Let the move set be the tool's, namely equivalence-preserving rewrites,
each defined by a source region that must be present in the current
netlist and a replacement, admitted under the two-table gate.  Consider
a netlist $s$ containing three consecutive subnetworks $g_1, g_2, g_3$,
and two rewrites.  $M_1$ replaces the region $\{g_1,g_2\}$ by a network
$h_{12}$, and $M_2$ replaces the region $\{g_2,g_3\}$ by a network
$h_{23}$, both preserving the overall function.  Choose the replacements
so that $M_1$ strictly improves $T_1$ and leaves $T_2$ unchanged, while
$M_2$ strictly improves $T_2$ and leaves $T_1$ unchanged.  Both are then
admissible at $s$.  Such shapes exist in the tool's own move set and are
the generic overlap case.  A rewrite that shortens conduction paths
without changing charged load moves $T_2$ alone.  One that removes
charged internal load without touching any chain at the cap moves $T_1$
alone.  Two window rewrites whose regions overlap in $g_2$ consume one
another's source material.  After $M_1$, the region $\{g_2,g_3\}$ is no
longer present, since $g_2$ has been consumed into $h_{12}$, so $M_2$ is
not applicable, nor is $M_1$ applicable again, its own source region
having been consumed.  The state $M_1(s)$ is therefore terminal for
$\{M_1,M_2\}$, and symmetrically so is $M_2(s)$.  Their cost pairs are
$\big(T_1(s)-a,\;T_2(s)\big)$ and $\big(T_1(s),\;T_2(s)-b\big)$ for some
$a,b>0$, which are incomparable in the product order and in particular
distinct.  Two distinct terminal states reachable from one state cannot
be joined, so the relation is not confluent.
\end{proof}

Proposition~\ref{prop:dominance} licenses the claim that the flow is in
a never-regress form.  Proposition~\ref{prop:nonconfluent} is why pass
order is a user-visible parameter of the tool rather than an
implementation detail.  It is a statement about the relation, and we do
not claim on its basis that orderings differ measurably on any
particular benchmark.  Where they do differ, we report it.  The composed
window arm and the individual-pass arms of Section~\ref{sec:results}
differ on several circuits, and the held-out sweep records circuits on
which the all-passes arm misses a single-pass result, which is pass
interaction under one budget observed in the field.

Three further pieces of accounting make negative results legible and useful.
Budgets are pure cuts.  With no budget set, every enumeration returns
what it returned before, in the same order, so no measured number can
shift merely because the budget mechanism exists.  When a budget does
truncate, the truncation is recorded and reported rather than left
silent, and the result is described as a floor rather than as a
fixpoint.

Every rejected candidate leaves a receipt.  Each gate decision is on the
record with both tables before and after and a reason.  A losing
candidate discarded in silence would make the statement that toggling an
option changed nothing ambiguous between the option being inert and the
candidate being built and lost.

The energy model's conventions are stated with every figure.  The
per-cycle convention, in which every gate swings one rail per
evaluation, is reported beside an activity-weighted figure.  Pad load is
attributed to the gate that physically drives the pad, following free
rail exchanges, so a realization whose outputs are wrappers cannot evade
its pad bill.  Primary-input drive is uncharged by default.  That is a
comparison convention rather than a model of a real part, and the
excluded quantity is always computed and reported, a median $70\%$ of
the energy, so a reader can see what the convention leaves unpriced
without re-running under the other one.

\subsection{The Re-Synthesis Passes in~\tool}
\label{sec:passes}

Six optional passes are available and all are off by default.  A default
run must be fast, and the run tells the user when better results are
available at a runtime cost.  All are gated identically by the
admissible-move rule of Section~\ref{sec:ledger}, so a pass that runs
for a long time and returns the netlist unchanged has produced a result
rather than a failure.

The governing cost identity for all of them determines how a search
should be planned.  Pass runtime is essentially the number of priced
candidates multiplied by the cost of one technology mapping, since
pricing a candidate means mapping it.  The candidate cap is therefore
the direct control and the wall-clock budget is the safety net.
Enabling every pass at once costs more than the sum of the parts,
because each pass re-prices energy against what the previous one
produced.

\subsubsection{Affine-cut extraction}
\label{sec:davio}

The first pass recognizes cuts computing affine functions over $\F$ and
re-emits them as exclusive-or trees.  Detection is by Boolean
difference.  For $f:\F^n\to\F$ write $\partial f/\partial x_i =
f|_{x_i=0}\oplus f|_{x_i=1}$.

\begin{lemma}[Affine detection]
\label{lem:affine}
$f$ is affine, that is $f=c\oplus\bigoplus_{i\in L}x_i$ for some
$c\in\F$ and $L\subseteq\{1,\dots,n\}$, if and only if
$\partial f/\partial x_i$ is a constant function for every $i$.  In
that case $L=\{i:\partial f/\partial x_i\equiv1\}$ and
$c=f(0,\dots,0)$.
\end{lemma}

\begin{proof}
($\Rightarrow$) For affine $f$, the two cofactors with respect to $x_i$
differ by the constant $1$ when $i\in L$ and are equal otherwise, so the
difference is the constant $1$ or $0$ respectively.
($\Leftarrow$) The positive Davio
expansion~\cite{thornton_rm_chapter1} is
$f=f|_{x_i=0}\oplus x_i\cdot(\partial f/\partial x_i)$.  If
$\partial f/\partial x_i$ is the constant $1$, the expansion degenerates
to $f=f|_{x_i=0}\oplus x_i$, peeling $x_i$ off as an exclusive-or term
whose residual $f|_{x_i=0}$ does not depend on $x_i$.  If it is the
constant $0$, then $f=f|_{x_i=0}$ and $f$ does not depend on $x_i$ at
all.  Applying this to each variable in turn peels off exactly the
variables in $L$ and leaves a residual independent of every variable,
hence a constant, whose value is $f(0,\dots,0)$.  The accumulated
expression is the stated affine form.
\end{proof}

\begin{remark}[Cost of the test]
\label{rem:affine-cost}
For a cut with $K$ inputs, $\partial f/\partial x_i$ is the
exclusive-or of the two cofactor halves of the truth table, and testing
it for constancy is one sweep of $2^{K-1}$ entries.  Deciding affinity
therefore costs one linear pass over the cut's truth table for each of
its $K$ variables, $O(K2^K)$ in all, and the test stops at the first
non-constant difference.  At the shipped $K\le12$ this is a fixed
tabular cost per cut, which is why the detector runs on every enumerated
cut rather than on a filtered subset.
\end{remark}

\begin{corollary}[Realization invariance]
\label{cor:realization-invariant}
The test of Lemma~\ref{lem:affine} depends only on the function a cut
computes, not on the gate network realizing it.
\end{corollary}

\begin{proof}
Every quantity in the test, the cofactors and their differences, is
defined from the function $f$ alone.  Two networks realizing the same
cut function have the same truth table and therefore the same
differences, so the test accepts one exactly when it accepts the other.
\end{proof}

Corollary~\ref{cor:realization-invariant} is the methodological point of the pass.  A detector built
on Lemma~\ref{lem:affine} recognizes an exclusive-or written as a
four-NAND cluster, as a NOR cluster, in AOI form, or as a wide tree
flattened by an earlier pass, uniformly and without a pattern library.
It also recognizes $k$-input affine cuts that no two-input template can
express.  Evaluating all $n$ differences on a cut of width $n$ costs
$O(n2^n)$ word operations, which is negligible at the widths a
$K$-feasible cover produces.

Commitment is sequential and conflict-free, since two cones can each be
rewritable and still share gates.  A cone is committed only if nothing
it needs has been dropped and nothing it drops has been committed, and
the pass ends with a structural self-check that refuses to return a
netlist containing an undefined net.  The width of the emitted tree is a
ladder rather than a constant, each width iterated to a fixpoint,
priced, and put to the gate.  The widths disagree in opposite directions
on different circuits, so hard-coding the width that happens to win
would encode one benchmark's answer into the method.

The pass's principal result is a property of the method rather than of a
circuit.  On \texttt{c1355} it recovers \texttt{c499} exactly.  The
netlist goes from $518$ gates to $174$ and the per-cycle figure falls
from $1.579776$ to $1.086096$\,pJ, which are \texttt{c499}'s own numbers
to six figures.  \texttt{c1355} is \texttt{c499} with its exclusive-ors
expanded into NAND clusters, and the pass un-expands them by measuring
the function.

That result has a withdrawn predecessor, and the reason for the
withdrawal is the important part.  An earlier version matched a fixed
four-gate structure discovered by inspecting \texttt{c1355}, and
produced a large and correct improvement on that circuit that is not
publishable as a general result.  The present pass tests a property of
the function and was justified by a census of recurring kernels rather
than by the circuit it helps.

\IfFileExists{figures/ex_davio.tex}{\input{figures/ex_davio}}{}

Figure~\ref{fig:ex-davio} shows the pass on a fixture written for this
paper, the parity $a\oplus b\oplus c$ expressed as two classic four-NAND
clusters, eight gates in total.

The pass takes the function of the cut, computes the Boolean difference
with respect to each of the three variables, finds all three differences
constant at one, and concludes by Lemma~\ref{lem:affine} that the cut is
affine over $\{a,b,c\}$ with constant $f(0,0,0)$.  It then re-emits the
cut as a parity tree over those leaves.  Because the test reads the
function rather than the netlist
(Corollary~\ref{cor:realization-invariant}), the same detection fires on
a NOR cluster, on an AOI form, or on a wide tree that an earlier pass
flattened, with no pattern library and no enumeration of shapes.

Measured, eight gates become two, $42$ devices become $20$, and the
capped figure falls from $0.006171$ to $0.004114$\,pJ, a ratio of
$0.6667$.  The uncapped figure does not move at all.  On this fixture
the entire saving is therefore the series bound no longer having work to
do, which is a distinction the two-table ledger makes visible and a
single scalar objective would have hidden.

\subsubsection{Bounded elimination and algebraic extraction}
\label{sec:elim}

The second pass runs bounded elimination that collapses a node into its
fanouts when it pays for itself in literal occurrences, followed by
algebraic division in the form of single-cube division, and optionally
multi-cube kernel extraction with rectangle covering.  Elimination is
what moves the energy in our measurements.  The kernel extractions
inside the same runs save literal occurrences without moving the priced
figures.  Single-cube division finds little improvement, because the
suites are two-input dominated and structural hashing has already merged
every identical-fanin pair.

The extraction filter follows the cost model of
Section~\ref{sec:cover}.  A divisor cube of $k$ literals used by $m$
nodes costs $k\cdot m$ literal occurrences inline and $k+m$ when
extracted, so it pays when $(k-1)(m-1)$ exceeds the configured gain
threshold.  The textbook two-by-two case is a dead break-even condition
under this accounting, because the model bills literal occurrences at
readers rather than at gates or area.  This is a concrete case of an
energy objective and an area objective disagreeing about what is worth
extracting.

The pass is not accepted everywhere, and the refusals indicate the gate
working.  On \texttt{c432} it improves $T_1$ by $3.5\%$ while regressing
$T_2$ by $4.5\%$, which is the trade the two-table objective exists to
refuse.

\IfFileExists{figures/ex_elim.tex}{\input{figures/ex_elim}}{}

Figure~\ref{fig:ex-elim} shows the elimination move.  A node whose
output is charged and read by two consumers is collapsed into both of
them.  The internal net leaves the ledger and each reader gains a
literal occurrence.  Whether that trade pays follows from the arithmetic
in Eq.~\eqref{eq:cost} rather than from gate count, which is why the
pass hill-climbs on the energy model itself.

On \texttt{c880} the pass is accepted at $0.8725$ of the uncapped figure
and $0.9722$ of the capped one.  The device columns beside those ratios
are a clean instance of Corollary~\ref{cor:nonmonotone} arising in
ordinary operation rather than in a constructed example.  Elimination
\emph{lowers} the uncapped device count from $1202$ to $1166$ and
simultaneously \emph{raises} the capped count from $1338$ to $1399$,
because the collapsed cubes are deeper and the series bound then inserts
more stages into them.  Energy falls on both tables regardless, and
energy on both tables is the only test the gate applies.

\subsubsection{Parallel-prefix restructuring}
\label{sec:prefix}

This pass targets serialized structures that compute running prefixes of
an associative combination.  Ripple carry and priority resolution are
the canonical cases, and these and similar structures are detected in a
phase-normalized view at no measurable cost during the netlist sweep.
The longest chain is then rebuilt as a Brent--Kung all-prefix
network~\cite{brent1982adder}\cite{kogge1973recurrence}.

The pass is a \emph{compound} move and it must be in this form.  Treeifying the
chain alone is measured as negative on most of the circuits we
examined.  On \texttt{c432} treeification alone takes the per-cycle
figure from $0.707608$ to $0.761090$\,pJ.  What the restructuring does
is expose window material that the serialized form had blocked, and only
after re-windowing does the whole move pay, reaching $0.580074$\,pJ, a
ratio of $0.8198$ against the unoptimized network.  Because the
intermediate state is worse than the input, the treeification and the
re-windowing are gated together as a single pass.  Gating them
separately would reject the move at the first step.

\IfFileExists{figures/ex_prefix.tex}{\input{figures/ex_prefix}}{}

Figure~\ref{fig:ex-prefix} shows the compound move on \texttt{rca8}, the
eight-bit ripple-carry adder that ships as the pass's own test vector.
The detector finds one carry chain of length $16$ in the
phase-normalized view and rebuilds it as a Brent--Kung all-prefix
network.  The pass then re-windows the result, pricing $512$ candidates,
accepting $5$, and skipping $210$ under the overlap guard.

The table beneath the figure justifies gating the move as a compound
operation.  Treeified alone, the circuit is worse on both tables,
yielding $1.1333$ and $1.1500$.  Re-windowed, it recovers to $1.0000$
and $1.0833$, which is better than the intermediate state but still not
better than the input.  The two-table gate therefore refuses the move
and the tool returns the netlist it was given.

A refusal is shown here rather than an acceptance because the refusal is
the more instructive outcome.  It is what a pass reporting $1.000$ in
the tables of Section~\ref{sec:results} looks like from the inside, and
it is a result rather than a failure.  The same mechanism is accepted on
circuits with denser reconvergence around the chain.  On \texttt{c432}
the recorded triple is $0.707608$\,pJ before, $0.761090$\,pJ treeified,
and $0.580074$\,pJ after re-windowing, a ratio of $0.8198$.  In both
cases the intermediate state is worse than the input, so gating the two
steps separately would reject the move at the first step and the
accepted result would be unreachable.

\subsubsection{Interior affine windows}
\label{sec:windows}

Two related passes re-encode the interior of a circuit region rather
than a single cut.  One is a single-output form operating inside one
reconvergent cone.  The other is a multi-output form operating on a
shared-leaf region with several roots.

This operation is a spectral translation.  An invertible affine change
of coordinates, $u = Ac\oplus m$, is performed on the cut variables and
the local function is re-expressed as $g(u)=f(A^{-1}(u\oplus m))$ in
algebraic normal form.  This is a classical construction from spectral
logic synthesis~\cite{hurst1985spectral}\cite{thornton2001spectral}, applied
here at cut granularity and read against the energy objective rather
than against diagram size.

The local score orders candidates rather than accepting them.
Acceptance is global.  The full netlist is remapped and repriced under
the release configuration and put to the two-table gate.  Encoder rows
are emitted only for coordinates in the transformed function's support.
A weight-one uncomplemented row is aliased directly to its leaf and
costs no gate.  Heavier rows become balanced parity trees.

The multi-output form keys windows by cut rather than by root, so cones
over a shared leaf set merge, and realizes them through a shared term
dictionary.  A single affine encoder serves the whole window and each
distinct monomial is built once.  An interior net that is a primary
output or has an outside reader becomes an additional root rather than
disqualifying the window.  This promotion is a necessity rather than a
refinement, since a carry-save array contains no single-output
fanout-closed cone at all, because full-adder cells are two-output.
Neither window pass subsumes the other in our measurements.

\IfFileExists{figures/ex_window.tex}{\input{figures/ex_window}}{}

Figure~\ref{fig:ex-window} shows the window passes on
\texttt{reconv24}, a fixture of three eight-input parity trees
recombined through majority logic.  It exists because the standard
benchmarks we examined offer no small window-pass acceptance.

Two properties of the implementation are visible in the figure.  Encoder
rows are emitted only for coordinates the transformed function uses, and
a weight-one uncomplemented row is aliased directly to its leaf and
costs no gate, so a coordinate change that merely permutes is free.  The
local score, which counts monomials, literals, and emitted encoder-row
weights, only \emph{orders} candidates.  Acceptance requires remapping
and repricing the whole netlist under the release configuration and
passing the two-table gate.

One window is accepted in this example.  The uncapped figure falls to
$0.6296$ of default and the capped figure to $0.9145$.  The gap between
those two ratios is the series bound reclaiming part of the win, since
the re-encoded interior is deeper than what it replaced.  A tool
reporting one number would have had to choose which of the two ratios to
publish.

\subsubsection{The linear pre-filter}
\label{sec:bdec}

The last pass re-encodes the \emph{output} space with an invertible
matrix $B\in\mathrm{GL}(m,\F)$.  The motivating idea is that an encoding
of the dependent variables in a subcircuit can map the transformed
circuit core to one that expends less energy, to the point that the
encoder and decoder energy plus the transformed core energy is a net
saving.  The core computes $h=Bf$ and a decoder network at the boundary
computes $f=B^{-1}h$.

The antecedent is the linear-transformation material of
\cite[Sec.~6.3]{thornton2001spectral}, developed there for
decision-diagram size reduction.  What differs here is the objective it
is searched against, energy, and the asymmetry that the objective
induces between $B$ and $B^{-1}$.

\begin{lemma}[Decoder cost]
\label{lem:decoder}
The composed network computes $f$ exactly, and in dual-rail logic the
decoder's mapped cost is $\sum_{r:\,w_r\ge2}(w_r-1)$ two-input
exclusive-or gates, where $w_r$ is the weight of row $r$ of $B^{-1}$.
\end{lemma}

\begin{proof}
Correctness is immediate from invertibility, since the decoder computes
$B^{-1}h=B^{-1}Bf=f$.  For the cost, row $r$ of $B^{-1}$ computes the
parity of the $w_r$ core outputs its nonzero entries select.  If
$w_r=1$ the row copies a single rail pair, possibly complemented.  In
dual-rail logic a copy is a wire and a complement is an exchange of the
two rails, so no mapped device is contributed.  If $w_r\ge2$, a parity
of $w_r$ inputs is computed by a tree of two-input exclusive-or gates.
Any such tree has exactly $w_r-1$ internal nodes, by induction on $w_r$,
since a tree combining $w_r$ leaves pairwise performs one combination to
reduce the leaf count by one and $w_r-1$ reductions in total.  Summing
over rows of weight at least two gives the stated count, and no other
devices are emitted by the decoder construction.
\end{proof}

\begin{corollary}[Asymmetry]
\label{cor:asymmetry}
Permutation matrices are free, and the decoder's cost depends on the
heavy rows of $B^{-1}$, not of $B$.  A light $B$ may have an expensive
inverse.
\end{corollary}

\begin{proof}
A permutation matrix has all row weights equal to one, so
Lemma~\ref{lem:decoder} assigns it zero gates.  The cost expression
depends only on the row weights of $B^{-1}$.  Over $\F$ the row weights
of a matrix and of its inverse are not related by any monotone function,
as the standard example of an upper-triangular $B$ with ones on and
above the diagonal shows, whose rows have weight at most $n$ while
$B^{-1}$ is bidiagonal, and conversely, so bounding $B$ alone does not
bound the decoder.
\end{proof}

Corollary~\ref{cor:asymmetry} is the reason the search bounds the row
weight of $B$ and of $B^{-1}$ simultaneously.  Bounding only $B$, the
natural thing to do since $B$ is what one manipulates, admits candidates
whose core is inexpensive and whose decoder returns the saving with
interest.  A second result bounds where the pass can act at all, and
addresses the question of why composing the specification with
invertible maps does not simply make the outputs quieter.

\begin{corollary}[Balanced outputs]
\label{cor:balanced}
Let $f:\F^n\to\F^n$ be a bijection, let the input vectors be independent
and uniform, and let $A,B\in GL(n,2)$.  Then every primary output of
$g=B\circ f\circ A^{-1}$ is exactly balanced, and consequently, by
Lemma~\ref{lem:renyi2}, has switching activity exactly $\tfrac12$.
\end{corollary}

\begin{proof}
$A^{-1}$ is a bijection of $\F^n$ and so preserves the uniform
distribution.  $f$ and $B$ are likewise bijections and do the same.
Hence $g(X)$ is uniform on $\F^n$, each coordinate of a uniform vector
is balanced, and Lemma~\ref{lem:renyi2} with $q=\tfrac12$ gives
$\alpha=2\cdot\tfrac12\cdot\tfrac12=\tfrac12$.
\end{proof}

Corollary~\ref{cor:balanced} states where the gain must live rather than
that no gain exists.  No invertible re-coordinatization alters the
marginal activity of any primary output, so whatever energy difference a
coordinate change produces is produced in the interior of the
realization that the synthesis flow finds for the transformed function.
That is what the measurements show.  The recorded search on
\texttt{crc8} leaves every primary-output activity untouched, as the
corollary requires, and nevertheless lowers the measured interior energy
on both tables, the identity pricing at $1.727880$\,pJ and the recorded
search terminating at $1.670284$\,pJ, a ratio of $0.9667$.

Mechanically, the pass required separate machinery rather than a flag.
The core and decoder are mapped separately and the two maps
concatenated, because weight-one rows of $B^{-1}$ must be dropped before
mapping.  A buffer is a free rail exchange in dual rail, and mapping one
as a gate charges an entire spurious block.  The search is a
deterministic hill climb over elementary row additions with both weight
bounds enforced.  The identity is priced first as the incumbent, and
every candidate is priced twice, once with an inexpensive ranking price
and once with a full price before anything is accepted.  Nothing is
accepted on the inexpensive number alone.

\subsubsection*{Candidate Acceptance}
\label{sec:bdecrealize}

One implementation decision inside the pre-filter machinery governed
everything the pass could do.  A candidate $B$ is priced by building a
netlist and mapping it.  Through \tool~v91.1 that netlist was the
original circuit with the $B$ bank attached to its outputs.  Row $i$ of
$B$ was emitted as the exclusive-or of the \emph{finished cones} of the
outputs it selects.  The cancellation that gives the re-encoding its
value was therefore never performed algebraically.  It was left for the
technology mapper to rediscover, and the mapper sees at most a $K$-cut.

The consequence is quantitative.  Take eight outputs that are
seventeen-input sliding-window parities,
$y_i=x_i\oplus\cdots\oplus x_{i+16}$, so that
$y_i\oplus y_{i+1}=x_i\oplus x_{i+17}$, a two-input parity.  The netlist
passed to the mapper computed that two-input parity as the exclusive-or
of two full seventeen-input trees, thirty-three gates, and asked the
mapper to discover that sixteen terms cancel.  At $K=12$ it cannot.
Pricing the same row addition three ways makes the mechanism visible.

\begin{center}
\small
\begin{tabular}{lrrr}
\toprule
realization of the same $B$ & devices & $T_1$ & $T_2$ \\
\midrule
incumbent (no re-encoding)      & $2176$ & $1.0000$ & $1.0000$ \\
core derived algebraically      & $1926$ & $\mathbf{0.9464}$ & $\mathbf{0.9205}$ \\
core as v91.1 built it          & $2198$ & $1.1964$ & $1.1250$ \\
\bottomrule
\end{tabular}
\end{center}

The candidate dominates when realized algebraically and regresses when
realized as the pass built it, so the never-regress rule refused it, and
refused it \emph{correctly}.  As constructed, the candidate was worse
than doing nothing.  With the pool widened until every one of the
fifty-six legal moves was priced rather than the default twenty-four,
the verdict did not change.  The search was not failing to consider
these candidates.  It was being handed candidates from which the saving
had already been removed.

From \tool~v91.2 the pass builds such a row directly over the primary
inputs.  Three conditions gate the construction and each is load-bearing.

The row must have weight at least two, since a weight-one row \emph{is}
an original output and has nothing to cancel.  This condition is also
what makes $B=I$ produce a byte-identical core to the previous
behavior, so every incumbent price and every circuit the pass declines
is untouched by the change.

Every output the row combines must be affine in the primary inputs.
This is tested structurally, by a walk over exclusive-or, exclusive-nor,
inversion and buffering, and is therefore exact.  There is no sampling
and no tolerance, which is what allows the two implementations to
continue owing each other byte identity.

The combination must cancel, in the sense that the symmetric difference
of the supports is strictly thinner than the thinnest support it came
from.  That condition is not decoration.  When two wide outputs have
disjoint supports the difference is their sum, so a fresh tree that wide
is pure cost, while the old form pays one exclusive-or over cones the
other rows need anyway.  Without the test the construction would hand
the search candidates worse than the ones it already had, which is the
same defect pointing the other way.  Anything failing any of the three
conditions falls back to the previous construction, so the change can
only add candidates.

What the structural test does not see is an output that is affine as a
\emph{function} but built from conjunctive clusters.  A functional test
is available and is used offline as an instrument
(Section~\ref{sec:bdecresults}).  Wiring it into the pass would require
an exactness argument that survives the byte contract between the two
implementations, and that is reserved as future work.

Section~\ref{sec:results} reports that this pass improves no circuit in
either benchmark set, which raises the question of why a pass that does
not pay on the measured circuits is in the tool.  There are three
reasons.

The first is that the pass is not idle on every input, and the inputs it
pays on are instructive rather than accidental.  On the shipped fixture
of Fig.~\ref{fig:ex-bdec}, whose outputs are overlapping wide parities,
it is accepted at $0.9091$ of the default on both tables, taking the
mapped circuit from $230$ devices to $192$.  On the same shape at scale
it reaches $0.8929$ and $0.7955$, and on a Gray-to-binary decoder it
takes ten per cent off the capped table.  Table~\ref{tab:bdecset}
reports these and the cases where it declines.  What those circuits
share is an output space whose coordinates are related by an inexpensive
linear map, and that shape is common in exactly the application domain
the energy argument is made for, namely checksums, syndrome formers, and
error-correcting logic.  Neither benchmark set here is weighted toward
that shape, and both were fixed before this pass existed.  A pass that
pays on a structure the benchmark set does not contain is a different
object from a pass that does not work, and we report the difference by
including the zeros and stating which structure is missing.

The second is that its cost is bounded and its failure mode is benign.
The identity matrix is priced first and held as the incumbent, so a
search that finds nothing returns the identity and the flow maps exactly
as it would have.  Only time is spent.  That is also why the pass is
excluded from the switch that enables the others and must be named
explicitly.  It is roughly an order of magnitude more expensive than
they are, and a user who has not asked for it should not pay for it.

The third is a matter of discipline.  Removing a mechanism because it
did not win on the circuits we happened to measure is the same selection
effect, in another form, that makes a development set unusable for a
generalization claim.  The pass is theoretically motivated, its scope is
bounded by Corollary~\ref{cor:balanced}, its cost is understood, and its
results are reported as measured, which in this case means reported as
zeros.

This pass also carries the clearest open problem in the tool.  From the
identity, every candidate is a single row addition, and all of them
score identically under every inexpensive estimator we tried, so the
ranking is uninformative in the first round and the search covers the
tie by brute force over a pool.  A pool of four reports no improvement
on the one development circuit the pass is known to help, and a pool of
twenty-four finds it.  The pass works for a reason we cannot yet
articulate as a heuristic, and two attempts at a better tie-break both
failed to predict acceptance.

An uninformative ranking is a curiosity at small $m$ and is not one at
large $m$.  The move space is $m(m-1)$ while the pool is a constant, so
coverage falls as $1/m^2$.  The last row of Table~\ref{tab:bdecset}
shows what that costs.  On a thirty-two bit Gray-to-binary decoder,
which has by construction the sparsest possible pairwise structure, the
search prices twenty-five of nine hundred and ninety-two legal moves and
finds nothing, while the sixteen-bit instance of the same circuit
accepts.  We regard this as a coverage-limited fact rather than as a
negative result about the circuit, because the two are not the same
claim and only the first is supported.  A ranking heuristic that
discriminates would be worth more to this pass than further work on the
construction.

\IfFileExists{figures/ex_bdec.tex}{\input{figures/ex_bdec}}{}

Figure~\ref{fig:ex-bdec} takes the pass through its own shipped fixture,
\texttt{bdtoy2}, whose six outputs are overlapping thirteen-input
parities, $y_i=x_i\oplus\cdots\oplus x_{i+12}$.  As written, the circuit
needs $72$ two-input XOR gates.

The pass starts from the identity, which it prices first and holds as
the incumbent, then hill-climbs over elementary row additions on the
output space.  A single row addition, $e_{i+1}\mathbin{+}=e_i$, cancels
twelve of the thirteen terms of each successive pair, so the re-encoded
core computes $h_{i+1}=y_i\oplus y_{i+1}=x_i\oplus x_{i+13}$, a
two-input parity, and only $h_0$ retains a long chain.  Measured, the
search accepts two such row additions under the shipped budget and
returns a mapped circuit at $0.9091$ of the default on both tables, $192$
devices against $230$.

The boundary then has to be restored, and this is where
Lemma~\ref{lem:decoder} and Corollary~\ref{cor:asymmetry} become
concrete.  The two matrices are printed beneath the figure.  $B$ is
light, every row of weight at most two.  Its inverse is the
lower-triangular all-ones matrix, with row weights $1$ through $6$.  By
Lemma~\ref{lem:decoder} the decoder costs
$\sum_{w_r\ge2}(w_r-1)=1+2+3+4+5=15$ two-input XOR gates, so the
composed circuit is $32$ gates against the original $72$.  A search that
bounded the weight of $B$ alone would have accepted candidates whose
inexpensive core is handed back at the boundary with interest.  The
search bounds both, and this fixture is where that decision is visible
in two printed matrices.

The fixture also shows the open problem.  From the identity, every
candidate is a single row addition, and on a structure this symmetric
they all look alike to every inexpensive estimator we tried, so the
first round is a tie that the search covers by brute force over a pool.
The search stops short of the construction the fixture admits, since the
ideal $B$ of the figure is five row additions and the recorded search
accepts two.

On the same structure one size larger the search was stopping
immediately, and not because it had run out of ranking signal.  The
candidates it was handed had the saving already removed from them, for
the reason given above.  With the construction corrected, the larger
instance accepts three row additions and reaches $0.8929$ and $0.7955$,
while this fixture is unchanged at $0.9091$.  The fixture is unchanged
because at thirteen-input windows over twenty inputs the technology
mapper was already recovering the cancellation on its own, which is why
the defect could hide behind a working example.  What remains
attributable to the search is the tie-break, and
Section~\ref{sec:bdecresults} bounds what that costs as $m$ grows.

\subsection{Input Statistics and the Drive Model}
\label{sec:drive}

Activity estimation normally assumes that a net's value is temporally
independent from cycle to cycle, which fixes $\alpha=2p_1(1-p_1)$.  How
much freedom that assumption gives away can be answered exactly.  Let
$X_t\in\{0,1\}$ be a stationary first-order Markov source with marginal
$p=\Pr[X_t=1]$ and toggle rate $\alpha=\Pr[X_{t+1}\neq X_t]$, and write
$a=\Pr[X_{t+1}{=}1\mid X_t{=}0]$, $b=\Pr[X_{t+1}{=}0\mid X_t{=}1]$.

\begin{lemma}[Admissible toggle rates]
\label{lem:alpha-range}
For $p\in(0,1)$, a stationary first-order source with marginal $p$ and
toggle rate $\alpha$ exists if and only if
$0\le\alpha\le2\min(p,1-p)$, and it is then unique, with
$a=\alpha/\big(2(1-p)\big)$ and $b=\alpha/(2p)$.
\end{lemma}

\begin{proof}
Stationarity of the two-state chain requires $p=(1-p)a+p(1-b)$, that is,
$(1-p)a=pb$.  The toggle rate is
$\alpha=\Pr[X_t{=}0]\,a+\Pr[X_t{=}1]\,b=(1-p)a+pb$, and substituting the
stationarity identity gives $\alpha=2(1-p)a=2pb$, which inverts to the
stated $a$ and $b$ and determines them uniquely.  These are
probabilities exactly when $a\le1$ and $b\le1$, that is,
$\alpha\le2(1-p)$ and $\alpha\le2p$.  Non-negativity of $\alpha$ is the
remaining constraint, and every value in the stated interval is realized
by the corresponding chain.
\end{proof}

\begin{corollary}[Independence is interior]
\label{cor:indep-interior}
For every $p\in(0,1)$ the temporally independent value
$\alpha_{\mathrm{ind}}=2p(1-p)$ lies in the open interior of the
admissible interval.
\end{corollary}

\begin{proof}
Take $p\le\tfrac12$ without loss of generality, so the upper endpoint is
$2p$.  Then $\alpha_{\mathrm{ind}}=2p(1-p)<2p$ exactly when $p>0$, and
$\alpha_{\mathrm{ind}}>0$ exactly when $p<1$.  Both hold on $(0,1)$.
\end{proof}

The consequence matters beyond this tool, because it bears on power
estimation in general.  An independence-based activity model is not
conservative.  For any marginal there exist stationary sources that
toggle strictly less than independent ones and sources that toggle
strictly more, and both are realizable.  A figure computed under
independence is an estimate at an interior point of a known interval
rather than a bound on either side of it.

Nor is this only a possibility argument.  Treating each ISCAS-89
sequential benchmark as a symbolic Markov chain and extracting the
stationary $(p,\alpha)$ pair of every flip-flop, the deviation
$\alpha-\alpha_{\mathrm{ind}}$ reaches $0.3403$ in absolute value across
the suite, with twelve of twenty machines stickier on balance and eight
more active.  Because the sign varies, the error cannot be absorbed by a
fixed derating factor.  We report the deviation rather than correcting
for it, because the correction is workload-dependent and the flow prices
what it is told.

The instrument behind those numbers is a symbolic primary-input
transition model (PITM).  Per-flip-flop next-state functions are
extracted from the driving machine's netlist as decision diagrams,
composed with an input-distribution diagram by weighted existential
abstraction, restricted to the reachable set, and quantized to interval
terminals.  At no point is the state space enumerated.

The resulting object provides two results.  Read as a \emph{sequencer}
model for the primary inputs, it drives the tag sweep with the lag-one
statistics of the deployment context, replacing the independence
assumption whose admissible range is the point of this subsection.  Read
as a \emph{simulator}, its repeated squaring yields multi-step behavior,
stationary laws, and reachability.

Before any stationary average is reported, the machine is tested for
halting behavior, since a stationary average over a machine that halts
describes the activity of a stopped machine.  The tool refuses rather
than warns, and the refusal is enforced by ordering, a caller being
unable to obtain the figure before the precondition has passed.
Reachable-state counts reproduce the published figures on all
twenty-two machines of the suite for which one exists, which is our
check that the chain construction is correct.  The construction and its
quantization-error analysis are reported at greater length in a separate
manuscript.

The synthesis flow consumes the PITM result as an activity file that
specifies the behavior of the primary inputs.  Primary-input drive
statistics change the synthesized circuit only under the
activity-priced cover, because the default cover does not consume
activity, and claims about workload-aware synthesis in this paper are
scoped to that cover.

\subsection{The Certified Optimality-Gap Instrument}
\label{sec:certified}

A typical synthesis tool evaluation considers its results and
evaluates whether another different flow is better than
its own but does not consider how much room for improvement is available.
Constructed benchmarks with planted optima have long shown that mappers
can sit far from optimal while improving on one another, but they cannot
bound the distance from optimal on an arbitrary netlist.  \tool~includes
an optimality-gap instrument that prices that distance on the real
benchmark set, with certificates that state how close the synthesized
circuit is to the floor of the tool's own search space.

The instrument rests on a separability property of the cost model.
Under the release configuration, the $T_1$ objective of a mapped network
decomposes as a sum of independent per-cover contributions plus a
calibrated affine term, $T_1 = c_u W + P$, where $W$ is the total weight
of the chosen covers.  Each cover's weight is computable from the cover
alone, and the calibration constants are recovered from three
instrumented runs with a $10^{-9}$ self-check.  Energy-optimal covering
over the tool's own $K$-feasible cover space is therefore an instance of
weighted MaxSAT, with selection variables per candidate cover,
cardinality constraints for legality, and weights from the calibrated
model.

Three consequences follow.  For small and medium instances an exact
MaxSAT solver returns the true floor of the cover space, and when the
floor equals the shipped weight the shipped circuit is proven to attain
that floor, so no selection of covers in the space improves on it.  That
is a certificate rather than a comparison.  For large instances a
core-guided solver run in an anytime mode maintains a monotonically
rising certified lower bound whose invariant survives interruption,
alongside a greedy upper bound, and the pair brackets the optimum and
tightens with budget.  Because the reduction prices this tool's own
cover space, negative gaps are informative rather than contradictory,
since they attribute a win to a mechanism outside the covering algebra.
Section~\ref{sec:results} reports the program's results over the
development set.  Three statements carry the weight of the certificates
reported there and are recorded next.  The instrument's algorithms, its
solver engineering, and its evaluation as an instrument in its own right
are developed separately.

\subsection{The Statements Behind the Certificates}
\label{sec:certified-proofs}

Fix a subject graph with primary-input set $\Pi$ and required outputs
$O$, and let $\mathcal{S}(K,N)$ be the set of legal selections of
$K$-feasible covers drawn from at most $N$ enumerated cuts per node: a
selection realizes every output, and realizes every non-primary-input
leaf of every cut it chooses.  For a node $r$ and one of its cuts $c$,
let $w(r,c)$ count the literal occurrences of names outside $\Pi$ in the
two rails of the block the family's realizer maps for $(r,c)$, and write
$W(\sigma)=\sum_{(r,c)\in\sigma}w(r,c)$.

\begin{lemma}[Calibrated separability]
\label{lem:separable}
Under the release configuration there are constants $c_u>0$ and $P$,
fixed by the family record and the design's input and output boundary,
such that $T_1(\sigma)=c_u\,W(\sigma)+P$ for every
$\sigma\in\mathcal{S}(K,N)$, with each $w(r,c)$ computable from $(r,c)$
alone.  Any two selections of distinct total weight determine $c_u$ and
$P$, and any third selection is then an independent check on the pair.
\end{lemma}

\begin{proof}
By Proposition~\ref{prop:cost}, $T_1$ is the literal-occurrence
functional \eqref{eq:cost} of the mapped network.  The release
configuration maps each chosen block from its own cut, and the block's
gates read only names produced inside the block or standing at its
leaves, so the literal occurrences a block contributes are a function of
$(r,c)$ and the family record alone and are counted by $w(r,c)$; $c_u$
is the per-occurrence device capacitance carried in that record.  The
occurrences the release configuration does not charge to a block are
exactly those of names in $\Pi$, together with the pad load at the
output boundary and the family's self-load; none of these depends on
which cuts are chosen, and their contribution is the constant $P$.
Summing over the selection gives $T_1(\sigma)=c_u\,W(\sigma)+P$.  This
map is affine in $W$ with slope $c_u\ne0$, so two selections with
$W(\sigma_1)\ne W(\sigma_2)$ determine
$c_u=(T_1(\sigma_1)-T_1(\sigma_2))/(W(\sigma_1)-W(\sigma_2))$ and
$P=T_1(\sigma_1)-c_uW(\sigma_1)$ uniquely, after which a third selection
must reproduce its own reported $T_1$ from the recovered pair.  The
implementation performs that third evaluation and refuses to convert a
weight into an energy when the two disagree by more than
$10^{-9}$\,pJ.
\end{proof}

Separability is what makes the question tractable at all: it turns
minimization of $T_1$ over $\mathcal{S}(K,N)$ into minimization of the
integer quantity $W$, which is a weighted MaxSAT instance.  Write
$W^\star$ for that minimum.  The next statement is what licenses the two
kinds of entry in the reported table, and in particular why an
interrupted run may still be quoted.

\begin{proposition}[Both certificate modes are sound]
\label{prop:certified}
Let $F$ be the shipped circuit's own weight.  Then:
\emph{(i)} $\min_{\sigma\in\mathcal{S}(K,N)}T_1(\sigma)=c_uW^\star+P$,
so a solver that closes the instance returns the exact floor of the
cover space, and $F=W^\star$ proves that no selection in that space
improves on the shipped circuit.
\emph{(ii)} If $L$ is the cost accumulated by a core-guided search at
the moment a budget stops it, then $L\le W^\star$ regardless of the
solver's remaining state, and $L$ is nondecreasing as the search
proceeds.  Hence the quantity reported from an interrupted run,
$(c_uF+P)/(c_uL+P)-1$, is a \emph{ceiling} on the true gap, and it can
only fall as the budget grows.
\end{proposition}

\begin{proof}
\emph{(i)} By Lemma~\ref{lem:separable}, $T_1$ is an affine function of
$W$ with positive slope.  Such a map preserves the argmin set and
carries the optimal value to $c_uW^\star+P$, so minimizing $W$ over
$\mathcal{S}(K,N)$ and minimizing $T_1$ over it are the same problem
and have the same optimizers.  If the shipped circuit's weight equals
that minimum, no member of $\mathcal{S}(K,N)$ has smaller $T_1$, which
is the stated certificate.

\emph{(ii)} A core-guided search maintains a working formula, initially
the instance, and a bound $L$, initially $0$, under the invariant that
the working formula's optimum equals $W^\star-L$.  Each round extracts a
set of soft clauses that is unsatisfiable together with the hard
clauses; every solution therefore falsifies at least one of them and so
pays at least that set's minimum weight $\mu>0$.  The round adds $\mu$
to $L$ and relaxes by $\mu$, which subtracts exactly $\mu$ from the
working formula's optimum and restores the invariant.  So $L$ increases
by $\mu$ each round, and since the working formula's optimum
$W^\star-L$ is a cost and hence nonnegative, $L\le W^\star$ throughout.
Neither statement mentions termination, so both hold wherever the search
is stopped.  For the ceiling, $c_uL+P\le c_uW^\star+P=
\min_\sigma T_1(\sigma)$, and the numerator $c_uF+P$ does not depend on
$L$, so the reported ratio is at least the true one and decreases as $L$
rises.  When the solver closes, $L=W^\star$ and the ceiling is the gap.
\end{proof}

The instrument prices $T_1$, but the tool reports two figures and
accepts on both.  The certificate transfers to the second figure at no
additional solver cost, which is the last statement of this section.

\begin{proposition}[A certified $T_1$ floor is a $T_2$ floor]
\label{prop:t2floor}
Every mapped network satisfies $T_2\ge T_1$, with equality exactly when
the series-bounding pass extracts no stage.  Hence
\[
\begin{aligned}
\min_{\sigma\in\mathcal{S}(K,N)}T_2(\sigma)
 &\;\ge\;\min_{\sigma\in\mathcal{S}(K,N)}T_1(\sigma)\\
 &\;=\;c_uW^\star+P\;\ge\;c_uL+P
\end{aligned}
\]
for every accumulated bound $L$: a certified floor for $T_1$ over the
cover space is also a floor for the $T_2$ optimum over that space, even
though $T_2$ itself is not separable and admits no such reduction.
\end{proposition}

\begin{proof}
$T_2$ is the functional \eqref{eq:cost} evaluated after the
series-bounding pass of Proposition~\ref{prop:cap}.  That pass performs
one kind of rewrite: it extracts a subnetwork whole and replaces it by a
literal on a new stage.  The extracted subnetwork's gates and their
literal occurrences survive inside the new stage, the replacing literal
adds one further occurrence at its reader, and the new stage carries the
family's self-load; every one of these terms is nonnegative, and no term
of \eqref{eq:cost} is removed, since the pass deletes no gate and no
reader.  Applying this to each extraction gives $T_2\ge T_1$ network by
network, with equality precisely when no extraction occurs, which is the
case exactly when the network already respects the in-flight series
bound.  Taking the minimum of both sides over $\mathcal{S}(K,N)$
preserves the inequality, and chaining with
Proposition~\ref{prop:certified} gives the display.  The
non-separability of $T_2$ is not an obstacle here and is the reason the
detour is needed: the bounding pass couples blocks, because whether a
block is extracted depends on the depth accumulated by the blocks
feeding it, so no per-block weight reproduces $T_2$ and the argument of
Proposition~\ref{prop:certified} has no direct $T_2$ analogue.
Table~\ref{tab:default} exhibits the equality case on the six
development circuits the bounding pass leaves untouched, and the strict
case on the other fourteen.
\end{proof}

\section{Development Approach}
\label{sec:development}

\tool\ was developed with an AI agent that generated the initial software
base, which was then reviewed and in a small number of cases modified by the
author.  The author specified every algorithm by supplying requirements and
specifications in natural language and, in several cases, past publications
describing related algorithms.  The generative capability of the agent was
used heavily.  The contribution reported in this section is the validation
approach that accompanied it, which established functional correctness and
also surfaced defects quickly and corrected them verifiably.

The division of responsibility rests on the respective strengths of the two
participants.  The human supplies the architecture and the expertise
regarding the methods used to implement the algorithms.  The agent supplies
rapid implementation and validation cycles.  Three procedural rules govern
the arrangement.  Full validations are performed by the human on local
hardware, and every result in this paper was obtained that way.  The ground
rules for accepting a change are fixed in advance, stated to the agent, and
applied without exception.  The tool is developed concurrently in two source
languages, which supplies an equivalence check at every iteration rather
than two alternative products.

That last rule is the operational core of the method.  The development
process benefits in the same way the operational tool benefits from
equivalence checking a newly synthesized subcircuit against the netlist it
replaced.  Here the check is carried out between two concurrently developed
implementations.  A Python reference captures the semantics of each software
fragment first.  A C implementation follows it.  The two are required to
agree not only in what they compute but in the bytes they emit.

\tool\ is therefore two complete implementations of one tool, developed in
lockstep.  The Python reference defines the tool's functionality.  The C engine
makes it fast enough for practical use.  This pair of implementations are 
not a prototype and a rewrite.  They
grew together, and the acceptance rule for every change was that the two
implementations agree at the strictest level the surface permits:  
\emph{byte-level equivalence},  \emph{numeric contracts} 
and \emph{verdict-class contracts}.  Every
surface is assigned one of three of these explicit contract classes, and every
release names the class of every surface.

Under \emph{byte contracts}, anything deterministic the tool writes must be
byte-identical between the two implementations, apart from a single
family-label header.  This covers mapped netlists, run records, SPICE decks,
and schematic exports.  The demand is stronger than functional equivalence,
because the port has to reproduce the reference's \emph{incidental}
determinism rather than only its answers.  The C implementation therefore
carries its own implementation of Python's Mersenne Twister generator,
reproducing how the interpreter turns the underlying 32-bit outputs into
random integers and floating-point draws.  It reproduces the reference's
float summation order, compiled so that the compiler cannot fuse operations
the reference performs separately.  It emulates the interpreter's hash-table
probe sequence wherever an algorithm's choices follow set-iteration order,
and it mirrors round-half-even formatting down to the semantics of rounding
recorded near-misses.  In one case a pass's activity score is reproducible
only because every term in it is a multiple of $2^{-7}$, so no partial sum
can round.  That property was proven and then engineered around with
scaled-integer accumulation rather than assumed.

Under \emph{numeric contracts}, energy figures target bit-identity and are
reported at $10^{-9}$ relative tolerance.  That band is roughly seven orders
of magnitude tighter than any physical laboratory instrument can reasonably measure,
which is the point.  A wrong constant or a dropped pad net has nowhere to
hide inside it.

Under \emph{verdict-class contracts}, wall-clock budgets are the one place
byte parity is refused, under a standing rule that the C engine is not
slowed down to make parity-agreement among the two implementations 
convenient.  Given the same budget, the C implementation completes
more work per second, which is the objective, so budgeted runs are validated by final equivalence,
verdict classes, and budget-honored receipts.  When C finishes a pass that
the reference's budget check starves, the asymmetry is documented as the
rule operating correctly.

\subsection{Lineage}

Writing a specification as two implementations and comparing them is
an instance of $N$-version programming~\cite{avizienis1977nversion}\cite{avizienis1985nversion},
here with $N=2$.  Two things differ from the classical setting.

The first is economics.  $N$-version development was largely abandoned
because writing everything twice doubles the cost.  An AI implementer that
can port tens of thousands of lines in weeks changes that arithmetic.

The second is the trust model.  The classical
objection~\cite{knight1986independence} is that independently written
versions make correlated errors, so \emph{voting} among them buys less
reliability than the independence assumption promises.  That objection does
not transfer here, because nothing votes.  The two implementations are
asymmetric, a reference that defines semantics and a port that must match
it, and disagreement is used for detection rather than for masking.  Any
byte of divergence stops the line.  In that respect the method is
differential testing~\cite{mckeeman1998differential} hardened from an
agreement criterion to a bit-identity criterion, with the reference
implementation as the oracle.

Parity cannot catch both implementations being wrong in the same way and
agreeing byte for byte on an incorrect answer.  That risk is handled by
layers that never compare the implementations to each other.  Self-verification
stages check every synthesized circuit against its source netlist by
simulation.  Brute-force cross-checks cover small instances.  A held-out
benchmark screen covers generalization.  Benchmarks used during development
are implicitly fitted, since every debugging decision conditions on them, so
the release discipline keeps a registry of circuits that no development
decision touched and screens correctness claims there.  The one shipped-path
divergence found late in the campaign was found that way, green on every
development benchmark and exposed by the held-out screen.  Parity catches
the implementations disagreeing.  The held-out set catches them agreeing
wrongly.  The claims of this paper rest on surviving both.

\subsection{Parity as the Development Loop}

Parity is the unit of development progress rather than a milestone at the
end of it.  Every ported surface follows the same cycle.  The reference is
written, or already exists, and is measured.  The port is written against
it.  Parity cells, which are concrete circuit and configuration pairs run
through both implementations and compared at the surface's contract class,
are added to a persistent matrix.  The port is complete only when every cell
is identical and every previously green cell is still green.  Because the
matrix is cumulative, it is at once the acceptance test for the newest
surface and the regression harness for every older one.

By the close of the port campaign the matrix held $237$ cells over $13$
circuits spanning every synthesis mode, technology family, and orchestration
surface.  It is backed by an $18$-stage cross-language suite that pins
parser agreement, an $82$-pair quick byte-parity check, each re-synthesis
pass, and each orchestration surface on every build.

Two properties make this loop effective rather than merely thorough.
Byte-level comparison catches a class of defect that functional testing
cannot reach.  A divergence in a tie-break, an iteration order, or a
rounding path produces outputs that are both functionally correct and
different, which is invisible to an equivalence check and fatal to
reproducibility.  Several recorded defects were of exactly this kind.  The
full check is also inexpensive enough to run on every change.  The complete
matrix executes in under three minutes on the validating workstation, $163$
seconds uncached at $14$ workers\footnote{In this example, a ``worker'' is a CPU core
executing a \tool~thread.} in the most recent release validation.
Most of that efficiency is a product of the port, since the C half of every
cell is nearly free.  The strictness ratchets one way.  A fence may be
loosened when the fence is shown to be wrong, and no surface was moved to a
weaker contract class in order to complete a release.

\subsection{Failures as Certificates}

The discipline treats every failure, refusal, and termination as an artifact
with the same evidentiary status as a pass.  Every candidate a pass rejects
leaves a machine-readable receipt carrying both cost tables and their
deltas.  Development ladder rungs are marked accepted, rejected, or skipped with a
reason, and budget expiry leaves a truncation string stating that the result
is a floor.  A sweep cell killed at its hard limit is recorded as a receipt
rather than discarded, because the kill \emph{is} the measurement, namely
that the configuration exceeds the stated budget on the stated hardware.
Defects are numbered, minimally reproduced, and kept in the record after the
fix, so the defect history is part of the artifact.

Three consequences follow.

Negative results survive into the record.  When one re-synthesis pass
produced zero accepted improvements across an entire held-out set, that
entered the record with the same receipts as the accepted results reported
in Section~\ref{sec:results}.  When the two custom hash benchmarks resisted
every optimization pass, that became evidence that the benchmark set is not
fitted rather than a deficit of the tool.

Stale evidence is adjudicated rather than overwritten.  When a failure row
in a sweep ledger contradicted a later successful rerun of the same cell,
the resolution was to establish which artifact was stale, record why, and
let the rerun supersede it explicitly, with the ledger keeping both the
error and its disposition.

The validation instruments are held to the same discipline as the tool.
Twice the fence rather than the code was the defect, and both incidents are
numbered in the same registry as code defects.  One parity fence demanded
exact float equality and reported a failure over a one-ulp difference
sitting seven orders inside the tool's own numeric contract.  A control
experiment established that the check was stricter than the contract it
claimed to enforce, and the correction went into the fence, with the receipt
kept.  A generated validation script that named itself in a command fence
recursed when mechanically re-extracted, and the correction became a
standing rule enforced mechanically thereafter.  A validation run on an
unbuilt tree showed that temporary-file comparisons could be satisfied by
artifacts left by the previous release's validation, so every temporary
artifact is now version-stamped and the validation script's first stage is
the clean build itself.  No check can now run against state the script did
not produce.

Every release ships a checkpoint document that states what changed and what
is not claimed.  The validation procedure's command blocks are mechanically
extracted into a runnable script, and the release gate refuses the version
if the script and the document have drifted, so the documentation cannot rot
relative to the procedure it documents.  The gate also asserts version
banners and the presence of a language-parity disclosure section, and it
checks, one assertion per ported surface, that every surface claimed
bilingual is wired into the C driver.  A separate C-only build gate
establishes that the engine needs no Python.  Each version is smoke-tested
from a clean extraction of the delivered artifact, a rule adopted after a
release whose working tree passed validation while its delivered archive was
unbuildable.  The record of record for a release is the full parity matrix,
run on hardware separate from where the code was written and by a different
party than the one who wrote it.

\subsection{Division of Labor}

One human domain expert supplied the architecture and the specifications,
stated in natural language at the algorithmic level.  He also supplied the
adjudications, which fixed the published gate topology a family bills,
whether a budget may throttle the fast implementation, what a validation is
allowed to claim, and whether a given disagreement is a defect in the fence
or in the code.  He ran the full-matrix validations on his own hardware.  An
AI agent wrote the large majority of the code in a cloud container, filling
in the implementation detail beneath the specifications, validating against
a reduced smoke set, and delivering every stage as a snapshot and every cut
as a checksummed archive.  The split follows comparative advantage.  The
human contribution is architecture and systems-level judgment.  The agent
contribution is syntax-level fluency, library breadth, and volume.  Neither
substitutes for the other.  An agent does not decide what a validation is
allowed to claim, and a human does not hand-port a Mersenne Twister
bit-exactly in an afternoon.

Four protocol elements make this division safe rather than merely fast, and
they transfer to projects of similar shape.

The reference implementation is the code review.  Under a byte-parity
contract the agent's ports are checked against an executable oracle rather
than by a human reading agent-written code, so review effort concentrates
on the specifications and the adjudications.

Validation labor is split by trust level, and the container's numbers never
enter the record.  The agent's environment runs a reduced smoke set,
everything it produces is treated as a development signal rather than as
evidence, and the record of record is produced by the owner on owned
hardware.  That division held when it was inconvenient, and it is why no
reported figure depends on an ephemeral environment.

The agent is instrumented as fallible.  Defects introduced by the agent
enter the same numbered registry as any other, and two entered it on the
final development day.  Both were caught by the machinery rather than by
care, which is the point.  A methodology that assumes a fallible
collaborator and instruments accordingly is robust to that collaborator
being fallible.

Progress is visible at all times.  Long-running work is detached with logs
and process records, reported on a fixed cadence, and snapshotted to durable
storage after every block.

The released tool is roughly $55{,}000$ logical lines across the two
languages, $29.6$ thousand of Python and $25.8$ thousand of C and C++.  The
basic COCOMO model~\cite{boehm1981cocomo} prices an effort of that size at
roughly $14$ to $37$ person-years across its organic-to-embedded range, and
about $22$ at the semi-detached grade.  The development calendar from the
first parity cell to the closing of the bilingual port was measured in
weeks.  The claim made for the methodology is not speed.  It is that the
development rate cost nothing in rigor.  Every accelerated line was accepted
through the same byte, numeric, or verdict-class gate, and the defects the
acceleration introduced are in the record with numbers, reproductions, and
fixes.

\section{Tool Performance}
\label{sec:results}

This section reports measured performance. The development set at the
shipped defaults, the re-synthesis passes over it, a held-out study on
twenty circuits no development decision touched, runtime, comparisons
against a published baseline convention, a nine-family cross-technology
study, the certified optimality gaps, and a device-level check of the
cost model.  Every figure is generated from run records rather than
transcribed, and the release discipline of
Section~\ref{sec:development} applies to all of them.

\subsection{The Benchmark Sets and Their Justification}
\label{sec:benchmarks}

Every development figure is generated over a fixed twenty-circuit
release set.  The set was assembled under one principle.  Aside from
carrying the complete ISCAS-85 suite~\cite{brglez1985iscas85} because
it is the standard choice, each of the other members was selected because they present
a challenging case for a specific aspect of synthesis, the covering
heuristic, the affine re-synthesis passes, the realization arbitration,
the energy accounting, or the front-end parsers.  Eighteen of the
twenty are standard or standard-construction circuits and two were
generated and chosen for this project for reasons given below.  The set also doubles as
the cross-language parity corpus. Among them the twenty circuits enter
through every front end the tool implements, so a parser divergence
cannot hide behind a single-format test set.
Table~\ref{tab:benchmarks} lists the members and the role each plays.

\begin{table}[t]
\caption{The twenty-circuit development set.  Gate counts are the
parsed-netlist counts reported by the released tool.  Two members
(\texttt{c17}, \texttt{xa}) are pad-load dominated, so their improvement
ratios are marked imprecise wherever they appear and no claim is built
on them.}
\label{tab:benchmarks}
\centering
\scriptsize
\setlength{\tabcolsep}{3.5pt}
\begin{tabular}{lrrl}
\toprule
circuit & in/out & gates & role in the set \\
\midrule
c17      & 5/2    & 6    & smallest live parity cell \\
xa       & 2/2    & 2    & minimal spectral anchor \\
reconv24 & 24/1   & 28   & smallest window-pass accept \\
c432     & 36/7   & 171  & deep reconvergence; affine true negative \\
c499     & 41/32  & 174  & XOR-dominated ECC; affine accept target \\
crc8     & 64/8   & 192  & linear-feedback chains; GF(2) target \\
8-bit hash & 8/8  & 272  & dense random table (see text) \\
c880     & 60/26  & 323  & ALU; non-monotone-in-$K$ witness \\
ctrl     & 7/26   & 349  & control-dominated; binary AIGER cell \\
c1908    & 33/25  & 479  & XOR-rich but irregular \\
c1355    & 41/32  & 518  & c499 with XORs expanded to NANDs \\
dec      & 8/256  & 624  & one-hot fanout; output-pad stress \\
c2670    & 233/140 & 790 & wide I/O, unconnected inputs \\
router   & 60/30  & 801  & mixed control/datapath \\
c3540    & 50/22  & 1043 & dense ALU with don't-cares \\
c5315    & 178/123 & 1605 & cover-runtime scale \\
c6288    & 32/32  & 2353 & $16{\times}16$ multiplier; the classic pathology \\
c7552    & 207/108 & 2381 & largest member; end-to-end scale gate \\
12-bit hash & 12/12 & 4120 & dense random table at PLA scale \\
t481     & 16/1   & 5250 & huge as an AIG, tiny decomposed \\
\bottomrule
\end{tabular}
\end{table}

Five of these circuits appear as worked examples in
Section~\ref{sec:renesis}, where the effect of a stage on them is drawn
rather than tabulated.  Several members deserve a sentence beyond the
table.  The pair
\texttt{c499} and \texttt{c1355} is a representation-robustness
instrument since the two compute the same error-correcting function, with
every exclusive-or of the former expanded into NAND clusters in the
latter, so they ask whether re-synthesis recovers the affine structure
the decomposition destroyed, and the affine pass's recovery of
\texttt{c499} from \texttt{c1355} (Section~\ref{sec:davio}) is the
receipt that it does.  \texttt{t481} is the classic decomposition
benchmark, enormous as a sum-of-products and tiny under the right
XOR-based decomposition.  \texttt{t481} is thus a standing test of whether the flow finds
non-obvious algebraic structure rather than merely tidying the obvious
kind.  \texttt{c6288} stresses the depth machinery with deep
reconvergence and is important to include in any tool that uses
decision diagrams since it is well-known to cause exponentially-sized 
decision diagrams under any variable order.

The two custom additions, \textbf{EightBitHashTable} (a full 256-row
PLA) and \textbf{TwelveBitHash} (4096 rows), are dense pseudo-random
substitution tables in the form of collision-free hashes that
obey the strict avalanche criterion in most cases. These two
non-standard custom circuits were felt necessary to include for three reasons.
First, they are good examples of anti-structure control.  Every other hard member of
the set is hard because structure is present but hidden. In contrast, these random
substitution tables have, by construction, no affine, algebraic, or
decompositional structure to find, and they provide maximum entropy 
relative to their width.  A
re-synthesis pass that reports wins on these circuits is indicative of an error
since it is expected that every affine pass should earn a clean reject which makes the
pair the benchmark set's false-positive instrument. 
Second, they discriminate between the tool's two
realization strategies precisely because no gate-level restructuring
can win. The arbitration between the structural cover and the
shared-forest realization is exposed undiluted, and the hash circuits
are where the shared-forest path earns its keep.  Third, they are the
application workload.  The tool's target regime, energy-harvesting
sensors and persistent low-intensity computing, computes through
lookup, dispatch, and hash tables at least as often as through
arithmetic.  A benchmark set drawn only from 1980s arithmetic and
control ASICs would say nothing about the workloads the energy argument
is actually designed for.

These twenty circuits were the
\emph{development set} for this work. The reported defaults, cut size,
price caps, the internal-load weight, the series bound, were all chosen
while looking at them, and no amount of care in reporting undoes a
parameter selected after seeing a result.  The generalization evidence
is the held-out study of Section~\ref{sec:heldout}, whose twenty
circuits were not used during development and which therefore did not influence
any decision made during specification, implementation and tuning of the algorithms.

All results use the transmission-gate reference family at $V=1.1$~V
with a per-device input capacitance of $1.7$~fF and a pass-device
on-resistance of $10$~k$\Omega$. These are calibrated constants with
stated provenance, the capacitance being a characterized 45\,nm library
input pin, rather than extracted values, and Section~\ref{sec:spice}
measures the consequence rather than asking the reader to take them on
trust.  The flow is deterministic, requires a pinned interpreter hash
seed and asserts it, and two runs of one command differ only in
wall-clock fields. The reference and performance implementations
produce byte-identical mapped netlists on the parity cells.

A third set appears in this version and is used for one purpose only.  The
linear pre-filter of Section~\ref{sec:bdec} consumes a property neither the
development nor the held out benchmark set
contains. The boundary-structure set
of Table~\ref{tab:bdecset} is therefore generated rather than collected
to ensure the structures that the linear prefilter were designed to optimize are present.
This third set comprises seven circuits, three built to have the property and four built not to, each from
a recorded seed and each equivalence-checked against its intended function
before use.  It is not a generalization set and no claim in
Sec.~\ref{sec:results} outside Sec.~\ref{sec:bdecresults} rests on it. Rather the third set
exists so that a negative result on the other two sets can be read as a
statement about those circuits rather than about the mechanism.  We keep it
visibly separate for the same reason the held-out set is kept separate, and
we would regard folding its circuits into either of the others as exactly the
selection effect the two-set discipline is designed to prevent.

\subsection{The Development Set and Assumed Defaults}

\begin{table}[t]
\caption{Twenty circuits at the shipped defaults.  $T_1$ is the
per-cycle switched capacitance of the mapped network.  $T_2$ is the same
quantity after the series bound.  Rows ordered by gate count.}
\label{tab:default}
\centering
\begin{tabular}{lrrrr}
\toprule
Circuit & Gates & Devices & $T_1$ (pJ) & $T_2$ (pJ) \\
\midrule
c17                &     6 &     22 & 0.008228 & 0.008228 \\
xa                 &     7 &     12 & 0.008228 & 0.008228 \\
reconv24           &    28 &   1208 & 0.111078 & 0.481338 \\
c432               &   171 &    674 & 0.707608 & 0.732292 \\
c499               &   174 &   2192 & 1.086096 & 1.151920 \\
crc8               &   192 &   1636 & 1.727880 & 1.727880 \\
8-bit hash         &   272 &   3874 & 0.131648 & 1.104609 \\
c880               &   323 &   1202 & 0.839256 & 0.888624 \\
ctrl               &   348 &    656 & 0.209814 & 0.209814 \\
c1908              &   479 &   1812 & 1.731994 & 1.849243 \\
c1355              &   518 &   2576 & 1.579776 & 1.625030 \\
dec                &   624 &   1280 & 3.159552 & 3.159552 \\
c2670              &   789 &   3242 & 1.908896 & 2.061114 \\
router             &   801 &    847 & 0.929764 & 0.929764 \\
c3540              &  1043 &   4014 & 2.772836 & 3.027904 \\
c5315              &  1605 &   7862 & 3.595636 & 4.011150 \\
c6288              &  2353 &  11806 & 15.324650 & 15.334935 \\
c7552              &  2381 &  13370 & 6.825126 & 7.508050 \\
12-bit hash        &  4120 &  93966 & 1.678512 & 24.305512 \\
t481               &  5250 &    607 & 0.123420 & 0.238612 \\
\midrule
median & & & 1.332936 & 1.388475 \\
\bottomrule
\end{tabular}
\end{table}

The spread in Table~\ref{tab:default} is four orders of magnitude in
$T_1$, and the two tables do not move together.  On \texttt{reconv24}
the series bound costs a factor of $4.3$ while on \texttt{crc8},
\texttt{dec}, and \texttt{ctrl} it costs nothing at all.  A circuit
whose mapped network already respects the bound pays nothing to have it
imposed, and one built from wide conduction paths pays a great deal.
That is the realizability cap doing what it is for, and it is why both
columns are carried everywhere rather than $T_1$ alone.

\subsection{Per-Pass Results on the Development Set}
\label{sec:perpass}

Table~\ref{tab:passes7} reports every re-synthesis arm over the full
set, namely the passes singly, the composed window arm, and the all-passes
switch.  Each entry is the ratio of $T_2$ to the circuit's own default.
Entries below $1.0$ are improvements, and a value of exactly $1.000$ is
the acceptance gate declining every candidate, a result rather than a
failure.  A dagger marks a run that reached its wall-clock budget.
Those entries are floors, not fixed points.  A double dagger marks a
cell that did not complete and produced no ratio.  Four cells carry it,
all on the \texttt{factor} arm, and they did not complete for two
different reasons.  On \texttt{t481}, \texttt{router} and \texttt{c2670}
the extraction exhausted the sweep's per-cell address-space ceiling and
the pass reported the exhaustion rather than reaching the operating
system.  On the $12$-bit hash the arm was still searching when it
reached the harness's hard limit of $16{,}500$ seconds and was
terminated, which is the measurement that this circuit exceeds that
arm's practical budget on the sweep hardware.  The \texttt{factor}
column is therefore reported over sixteen circuits rather than twenty.
The other seven arms on those four circuits are unaffected.

\begin{table*}[t]
\caption{Every re-synthesis arm over the development set, as the ratio of
$T_2$ to each circuit's own default.  Bold marks an improvement,
$\dagger$ marks a wall-clock floor, and $\ddagger$ marks a cell that did
not complete and produced no ratio, three on the address-space ceiling
and one on the hard time limit.}
\label{tab:passes7}
\centering
\input{tab_passes}
\end{table*}

Three observations organize the table.  Fourteen of the twenty circuits
improve under at least one arm, with the largest wins at $0.673$
(\texttt{c1355}, all passes: the affine and multi-output window passes
both accept, and the win is bought with fewer devices) and $0.677$
(\texttt{dec}, elimination: the classic energy-for-area trade at
$2.1\times$ the device count).  Both directions of the area-energy
relationship appear in one table, which is the concrete form of
Proposition~\ref{prop:cost}.  The multi-output window pass is the
workhorse, the winning ingredient on eight circuits alone or in
composition.  And the six circuits that decline everything are not
noise.  The two hash benchmarks resist every arm for the reason
that their components are spectrally indistinguishable from random
functions, and \texttt{t481} carries a certificate
(Section~\ref{sec:certified-results}) that no selection of covers in the
space the tool searches improves on its default mapping.

We report the default column as the primary result and the best-arm column
separately, and we label the second as what it is.  Taking the best of
several arms per circuit and presenting the envelope as ``the tool''
would be the same selection effect that makes a development set
unusable for a generalization claim, in a different costume.  The
default is what a user obtains.  The best-arm figures are what the
mechanisms can reach when one is willing to search.

\subsection{The Held-Out Study}
\label{sec:heldout}

The generalization evidence is a second twenty-circuit set, drawn from
ISCAS-85 members excluded from the development set, ISCAS-89
combinational cores, the EPFL suite, and two further locally generated
hash tables, selected once, registered before any run, and never used
in any development decision, and the shipped defaults were frozen before
the first held-out run.  Table~\ref{tab:heldout-default} reports the
set at the shipped defaults, and Table~\ref{tab:heldout-arms} the
re-synthesis arms.  Cells were budgeted at three hours of search each
with a hard kill at $16{,}500$ seconds.  $43$ of the $180$ cells reached
the search budget and are floors, and the eight arm cells of the
largest member (\texttt{arbiter}, $35{,}712$ gates) were hard-killed,
which is itself the measurement that the circuit exceeds every arm's
practical budget on the sweep hardware, its default-arm record standing
as its representation.

\begin{table}[t]
\caption{The held-out set at the with \tool~default options.}
\label{tab:heldout-default}
\centering
\footnotesize
\setlength{\tabcolsep}{4pt}
\input{tab_heldout_default}
\end{table}

\begin{table}[t]
\caption{Held-out re-synthesis arms: best arm per circuit, ratio to the
circuit's own default ($T_2$ and $T_1$); \texttt{lw+mw} is the composed
window arm and \texttt{opt-all} the all-passes switch.  \emph{(none)} means no arm
improved on both tables.  \texttt{arbiter} is excluded (default record
only; every arm hard-killed at the budget).}
\label{tab:heldout-arms}
\centering
\small
\setlength{\tabcolsep}{4pt}
\input{tab_heldout_arms}
\end{table}

The principal figures, with their denominators result in fifteen of nineteen held-out
circuits improving under at least one re-synthesis arm, with a best-arm
$T_2$ median of $0.9149$ over all nineteen, and $0.8931$ median and
$0.8408$ geometric mean over the fifteen that improve.  The deepest
single cut is the elimination pass on \texttt{cavlc}, to $0.41$ of the
default energy at $0.26$ of the devices, the one cell at scale where a
pass wins energy and area together.  The comparison of interest
is with Section~\ref{sec:perpass} where fourteen of twenty development
circuits improve as compared to fifteen of nineteen held-out circuits improve, and
the median improvements are of the same size, so the passes' value
survives on circuits their defaults never saw.  Pass interaction under
one budget is real.  The all-passes arm is not the per-circuit maximum
everywhere, missing the elimination result on \texttt{cavlc} and
\texttt{s5378}, which is why the portfolio framing, best single arm, is
the honest summary.

The negative results are reported with the same prominence.  The linear
pre-filter confirmed zero improvements in eighteen completed held-out cells,
and it improved no development circuit either under the budget of
Table~\ref{tab:passes7}.  Those zeros are unchanged by the rebuild described
in Sec.~\ref{sec:bdecrealize}, and Sec.~\ref{sec:bdecresults} explains why
they are unchanged rather than merely reporting that they are.  No circuit in
either set has the property the pass needs.  The two held-out hash tables
resist every arm, consistent with their design intent and with their
development-set counterparts, which we interpret as evidence that the
benchmark suite is not fitted to the passes.  And \texttt{max} and
\texttt{priority} show no improving arm at all, their structure offers
no carry chains, no affine cuts that price through, and no accepted
windows.

\IfFileExists{tab_bdecset.tex}{\input{tab_bdecset}}{}

\subsection{Linear Pre-Filter Usage and Results}
\label{sec:bdecresults}

The pass of Sec.~\ref{sec:bdec} improves nothing in either benchmark set,
which invites two different explanations.  Either the circuits lack the
structure, or that the pass cannot exploit it.  Until the rebuild of
Sec.~\ref{sec:bdecrealize} both were true at once, which is precisely the
situation in which a table of zeros in the Development and Held-out
Benchmark circuit results teaches nothing.  With the construction
corrected the two can be separated, and this subsection addresses
this situation.

\subsubsection{The property the pass needs}

The re-encoding optimization earns its keep when a row of $B$ combines outputs whose
combination is less expensive than the outputs themselves.  Over $\F$ that is a
statement about the support set.  Writing $S_i$ for the set of primary inputs an
affine output $i$ depends on, an elementary row addition $e_i \mathbin{+}= e_j$
produces a core row of support $S_i \triangle S_j$, and it is worth taking
exactly when
\[
  |S_i \triangle S_j| \;<\; \min(|S_i|,|S_j|),
\]
that is, when the two outputs are not merely linear but \emph{near}, meaning closer
to each other than either is to the origin.  Three conditions therefore have
to hold together, and it is worth separating them because circuits fail them
for different reasons.  The outputs must be affine.  They must be near.  And
the search must be able to reach the improvement by single moves that each
dominate on both tables.

Being linear is not sufficient, and the natural counterexample is the one a
reader will reach for first.  A cyclic redundancy check is the archetype of a
linear circuit, and \texttt{crc8}'s eight outputs are all affine, and all
wide parities of \emph{similar} support, so their pairwise differences stay
wide, the sparsest being twenty-eight against a thinnest support of
thirty-three.  Nothing collapses.  The same holds, more starkly, for a
Hamming syndrome former, whose $r$ rows have weight $2^{r-1}$ and whose
pairwise differences also have weight $2^{r-1}$, half the codeword, by the
construction of the code.  \texttt{hamsynd} in Table~\ref{tab:bdecset} is
that circuit, and it declines.

\subsubsection{Linear Prefilter Results}

Table~\ref{tab:bdecset} reports a third benchmark set, disjoint from the
development and held-out sets and built specifically to contain boundary
structure, together with controls built specifically not to.  Every member is
generated from a recorded seed by a script that ships with the tool, and every
one is equivalence-checked against its intended matrix before the tool sees
it.  Each output is emitted as its own balanced exclusive-or tree with net
names unique to that output, so no two outputs share a textual subexpression.
A structural hasher has nothing to merge, and the saving is algebraic or it
does not exist.

Two of its members are notable.  \texttt{bdslide} is a bank of
seventeen-input sliding-window parities, the shape a moving checksum or a
convolutional parity accumulator has, and the pass takes it to $0.8929$ and
$0.7955$ of the unflagged result with devices falling from $2176$ to $1432$.
\texttt{gray2bin16} is the textbook Gray-to-binary decoder,
$b_i=g_i\oplus g_{i+1}\oplus\cdots\oplus g_{N-1}$, in which consecutive
outputs differ by exactly one input, the extreme case of nearness, and
the pass takes ten per cent off the capped table.  We report its uncapped
ratio, $1.0526$, as the regression it is intended for.  The pass's own never-regress gate
measures against its internal identity pricing, which on this circuit
differs from the ordinary flow by one block, so a candidate it scores as
$1.0000$ on $T_1$ is in fact slightly worse than not using the flag.  That
seam is a defect of the pass's incumbent bookkeeping rather than of the
re-encoding, it is the only member of the set where the two disagree, and it
is recorded here rather than smoothed over.

The declining members are as much of the result as the accepting ones.
\texttt{hamsynd} is linear but not near.  \texttt{bdwin}, whose full
re-encoding cuts devices by a factor of four and a half, is refused because
it \emph{raises} uncapped energy.  A device count is not an energy, and the
two-table gate exists for exactly this case.  \texttt{gray2bin32} has the
same structure as its sixteen-bit sibling and is refused for a third reason
again coverage.  At $m=32$ there are $992$ legal moves and the pool prices
twenty-five of them.  \texttt{bdnull} is the control, same size and density
with independently drawn rows, and finds nothing, which is what it is for.

\subsubsection{Analysis of Linear Prefilter Results for Development and Held Out Benchmarks}

The remaining question is whether the benchmark sets are unlucky or
structurally excluded, and the answer is based on the circuits' structure.
Applying the support test above to every output of every circuit in both sets,
 affineness established by evaluation, over the whole cone rather than
within a cut, so the answer does not depend on how the function happens to be
written, returns no pair satisfying the inequality in any circuit of
either set.  There is nothing for this particular synthesis pass to find, and the zeros in
Sec.~\ref{sec:perpass} are the correct answer about for these circuits rather
than a failure of the search.

The same screen over a hundred circuits of the MCNC PLA suite 
was conducted.  Thirty-nine of the hundred have at
least one output that is an affine function, \texttt{xor5} is a five-input
parity outright, and \emph{none} has a collapsing pair.  Linear content is
not rare in that suite and near pairs are absent from it entirely.  The reason
is representational as much as functional.  A PLA is a two-level cover in
which each output carries its own cube set, and independence of the outputs
is precisely the absence of the property this pass consumes.  We record the
screen rather than the hundred runs it makes unnecessary.

\subsubsection{Where Linear Prefilter is Actionable}

On circuits whose outputs are affine in the primary inputs, structurally so, and
pairwise near, and few enough that a constant-size pool covers a useful
fraction of the $m(m-1)$ moves, the pass reduces capped switched capacitance
by between four and twenty per cent on the instances measured here.  Outside
that description we have no evidence and claim none.  The circuit classes
that satisfy the description are identifiable in advance and include sliding-window
and moving-sum parities, Gray and other reflected-code converters, syndrome
formers with staggered rather than balanced checks, and the quasi-cyclic
parity-check blocks of modern codes, in which consecutive rows are shifts of
one another and therefore differ in a bounded number of positions. Such circuits and structures
are efficiently identifiable since the screen above costs $n+1$ simulations per
circuit against the several seconds each candidate pricing costs.  The
practical recommendation that follows is to run the screen and enable the
pass when it fires, rather than to enable the pass and hope.

\subsection{Runtime}
\label{sec:runtime}

Every number in this paper is defined by the reference implementation,
and none of them needs to be produced by it, since all six re-synthesis
passes are carried in the C engine, byte-identical to the reference on
the mapped-netlist format, and every results table can be regenerated
from the C side alone.  Table~\ref{tab:wallclock} reports the
wall-clock cost of doing so, the full development set at the default
configuration regenerates in under five minutes of single-threaded C,
and the complete five-pass optimization chain on \texttt{c880} takes
under four minutes, where the reference implementation needs on the
order of two hours for the same chain.  A designer iterating on a
technology target reruns the whole optimization surface interactively
rather than overnight.

\begin{table}[t]
\caption{C-engine wall-clock times (seconds).  Left: the development
set at the default configuration.  Right: each re-synthesis pass alone
and the composed five-pass chain on \texttt{c880}, at the shipped
budget policy.}
\label{tab:wallclock}
\centering
\input{tab_wallclock}
\end{table}

\subsection{Comparison Conventions and the Baseline}
\label{sec:baseline}

The published convention closest to this work maps a NOR-gate netlist
structurally onto transmission gates and we reproduced it as a baseline,
both as the naive construction and with the netlist first optimized by
a conventional AIG optimizer, and verify the baseline's functional
equivalence to the original netlist on random vectors for every
circuit.  Across all twenty development circuits the tool uses
$0.003\times$ to $0.25\times$ the per-cycle switched capacitance of the
optimized-NOR baseline, median $0.096\times$, on $20$ of $20$.  Against
the optimized OR-inverter construction, mapped like-for-like through
the same energy model, the tool's covers use a median $0.885\times$ on
$14$ of $20$ circuits (both sides capped: $0.903\times$ on $12$ of
$20$).  The structural axes disagree with the energy axis, and we
report them separately.  Fewer mapped gates and fewer clock-phase levels
at roughly $1.6\times$ the devices, which is not a contradiction but
the cost decomposition restated, since most of our devices hang off
uncharged primary-input literals.  On five of the twenty circuits the
AIG optimization made the NOR baseline worse than the naive
translation, a reminder that optimizing a proxy is not optimizing the
cost.  The baseline reproduces the published method, not the authors'
implementation, which may contain refinements beyond the description.

Replacing the unit-capacitance model with characterized values calibrates the axis
into joules without changing any ratio.  But mapping both flows to the
same static-CMOS standard cells and pricing both with a conventional
engine \emph{inverts} the result.  The median ratio moves to $1.887$
against this work, with seven of twenty circuits still winning.  A
multi-control pass gate is charged only when it fires, while a static-CMOS
tree switches at every level on every transition regardless.  The
energy advantage reported above is therefore claimed for
energy-recovery pass-transistor targets specifically, and a
conventional-CMOS reading of these numbers would be wrong.

\subsection{Cross-family Behavior}
\label{sec:families-results}

The family sweep runs the release configuration over all nine mapping
targets for all twenty development circuits, $180$ runs with zero
failures and zero truncations.  Table~\ref{tab:families7} summarizes.

\begin{table}[t]
\caption{The family choice, priced: per-cycle switched capacitance
relative to the release family over twenty circuits.  ``Wins'' counts
circuits where the family is the outright best choice in terms of energy.}
\label{tab:families7}
\centering
\input{tab_families}
\end{table}

Under the \tool~cost model, ECRL and PAL are jointly the most energy-efficient families, at a
$0.612$ geometric-mean ratio to the release family.  The two are not merely
close as Section~\ref{sec:family-records} shows, their shipped records are
identical in every field this metric reads, so they return the same figure
on every circuit and the win column attributes the tie to whichever is
listed first.  The seventeen circuits credited to ECRL in the table are
seventeen ties.  The quasi-adiabatic families with the lightest keeper cells
sit at the top of the table, and the fully static, fully adiabatic families
pay for their discipline in devices.  The boundary, stated so the table
cannot be over-read.  These are calibrated switched-capacitance figures under
each family's published topology and billing.  $T_2$ does not see the
non-adiabatic residues of the quasi-adiabatic families, and it does not
credit the fully adiabatic families' asymptotic scaling at slower clocks.
The device-level instrument of Section~\ref{sec:spice} exists to
close that boundary.

\subsubsection{Held-out Set Sweep Results}
\label{sec:families-heldout}

The family study was then repeated on the held-out circuits, again at frozen
defaults and again nine families per circuit, and
Table~\ref{tab:families-heldout} reports it.  Six of the $180$ cells
produced no record and all six are on the two pipelined families, four hard
kills at the budget, on \texttt{i2c} and \texttt{s838}, and two on
\texttt{adder} where the sweep's per-cell memory ceiling stopped a decision
diagram construction cleanly rather than letting it reach the operating
system.  No cell on any of the other seven families failed anywhere in the
sweep.  The two pipelined columns are therefore over seventeen circuits
rather than twenty, and since the three missing circuits are among the
harder members, those two figures are if anything optimistic.

\begin{table}[t]
\caption{The family choice on the held-out set: per-cycle switched
capacitance after the series bound, relative to the release family, over
twenty circuits at frozen defaults.  $\dagger$ marks the two families whose
column is over seventeen circuits because three cells did not complete, and
``wins'' counts circuits where the family attains the minimum, ties
included.}
\label{tab:families-heldout}
\centering
\small
\setlength{\tabcolsep}{4pt}
\input{tab_families_heldout}
\end{table}

The result we would put weight on is that \textbf{the ordering reproduces
exactly, nine families of nine}.  ECRL and PAL tied in terms of minimal energy, then
PFAL, CAL, TG, the release family, SPGAL, and the two pipelined families
last.  A designer choosing a family on the evidence of
Table~\ref{tab:families7} would make the same choice on circuits the
defaults never saw.  The magnitudes drift modestly, between $9$ and $15$ per
cent, on the seven non-pipelined families, and by $75$ per cent on the two
pipelined ones, by the same factor on both, which is what one expects when
the second is the first under a stated multiplier rather than an independent
measurement.

The held-out sweep also settles the identity claimed in
Sec.~\ref{sec:family-records} from the parameter files alone.  Across all
twenty circuits, the largest difference between the ECRL and PAL figures is
$0.0$\,pJ and the largest difference in device count is zero.  They are not
close.  They are one column printed twice, and every win either is credited
with is a tie with the other.

One measurement in the sweep asks a question of the shipped default rather
than of the alternatives, and we would rather print it than leave it in a
data file.  The cost of imposing the series bound, $T_2/T_1$, has a median
of $1.004$ to $1.005$ on every family in the set except the release family,
whose median is $1.056$ and whose maximum reaches $12.63$, on the two dense
hash tables.  On \texttt{TenBitHash} the release family maps to
$0.411400$\,pJ uncapped, a tenth of what the pipelined family needs, and the
bounding pass then takes it to $5.195982$\,pJ at $58{,}193$ devices against
$19{,}745$.  That is the entire source of the only two cells in the sweep
where a pipelined family appears to win.  It is the baseline losing rather
than 2LAL winning.  The same shape appears on the development set's own
dense table, where the release family runs $1.678512$ to $24.305512$\,pJ
(Table~\ref{tab:default}), so the effect reproduces on both benchmark sets
on the circuit class that provokes it.  The release family differs from
plain TG in exactly one field, an in-flight series limit of six against
four, and plain TG's worst cap penalty anywhere in the sweep is $1.049$.
We report the measurement and the coincidence and stop there.  Isolating the
mechanism is a small experiment we have not run, and until it is run this
belongs in the paper as a property of the default configuration on dense
wide logic rather than as an explained result.

\subsection{Certified Optimality Gaps}
\label{sec:certified-results}

Table~\ref{tab:certified} gives the results of the optimality-gap
program of Section~\ref{sec:certified} over the development set.

\begin{table}[t]
\caption{Gap from the shipped circuit's $T_1$ weight to the
certified floor of the tool's own cover space.  ``Exact'' rows are
solved to optimality.  The ``anytime bound'' rows report the certified
bracket's ceiling at the run's budget, so the true gap is at most the
figure shown.  The two negative rows are circuits whose shipped realization
is reached by a mechanism the cover space does not contain, so the floor of
that space is not a lower bound on them.}
\label{tab:certified}
\centering
\input{tab_certified}
\end{table}

Three circuits are solved exactly at gap zero, and \texttt{t481},
$5{,}250$ gates in its source netlist, carries a full certificate.  Its
shipped circuit attains the certified floor of the tool's cover space, so
no selection of covers in that space improves on it.  We state the
certificate that way rather than calling the circuit optimal without
qualification, because the claim is bounded by the space searched.  It says
cover selection has nothing left to optimize on this circuit, not that another
circuit computing the function is less costly.  The bounded rows are ceilings rather than measurements.
\texttt{crc8}'s bound, for instance, fell from $72\%$ to $35\%$ as its
anytime budget grew, and further budget moves only the number, not the
circuit.  The two hash benchmarks return large negative gaps, and this
is the instrument earning its keep.  The shipped wins on those circuits
belong to the shared-forest realizer, not to cover selection, an
attribution the reduction makes precise.  Together with the spectral
profiles of the hash components, flat at the random-function ceiling,
the program separates three kinds of hard circuit that a results table
alone cannot distinguish, namely structure-absent, structure-exhausted with
certificate in hand, and budget-limited, where the bracket is honest
about what remains unknown.

\subsection{Device-Level Validation of the Cost Model}
\label{sec:spice}

The energy model is an accounting of charge, and an accounting can be
checked.  A SPICE deck is generated directly from the mapped netlist,
and the charge the network draws from the power clock is compared, on
the same circuit, against the switched capacitance the model summed.
An independent reimplementation of the load accounting, written from
the mapped-netlist format specification rather than derived from the
tool, reproduces the tool's figures exactly on \texttt{c17}
($0.008228$~pJ) and \texttt{c432} ($0.732292$~pJ).  That is evidence
about the accounting rather than about the constants, which are inputs
to both.

\IfFileExists{figures/energy_vs_ramp.pdf}{%
\begin{figure}[t]
\centering
\includegraphics[width=\columnwidth]{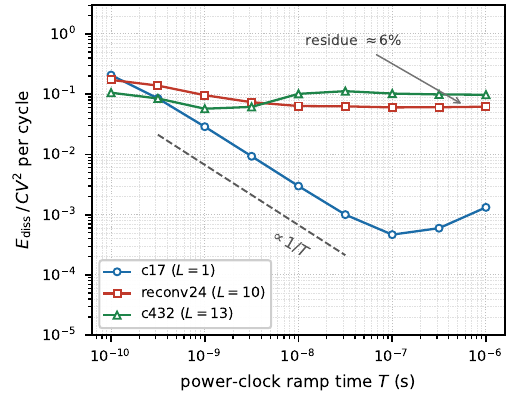}
\caption{Dissipation per cycle against power-clock ramp time for three
representative circuits, normalized by the per-cycle figure the model
bills.  $L$ is the logic depth of the mapped network.  The shallow
circuit follows the $1/T$ law over about two decades.  The two deeper
ones settle onto a floor at roughly $6\%$ of the billed energy.}
\label{fig:ramp}
\end{figure}}{}

Sweeping the ramp time requires one precaution.  A $k$-phase network $L$
levels deep needs about $\lceil L/k\rceil$ cycles simply to move data
from inputs to outputs, and measuring before that reads the charging of
nodes that have never held charge, which costs on the order of
$\tfrac12 CV^2$ whatever the ramp time and is easily mistaken for
leakage or a threshold residue.  We run eight cycles, measure the last,
and confirm the per-cycle energy is unchanged across the preceding
cycles.  What the sweep then shows is that on \texttt{c17} the measured
log-log slope lies between $-0.94$ and $-0.99$ across about two decades
of ramp time, the $1/T$ behavior the family is supposed to exhibit and
the check that the deck operates adiabatically at all, while the two
deeper circuits settle onto a ramp-independent floor at about $6\%$ of
their billed energy.  We looked first for a numerical cause and did not
find one.  Moving the deck's shunt resistance by four orders of
magnitude moves the figure by under $0.1\%$.  We therefore report the
floor as a property of the mapped circuit rather than of the deck.  One
reading consistent with the observation is that the reference family
carries no explicit pipelining discipline, so on a deep network a
control rail may return toward rest before the node it gates has
recovered, leaving charge unreturned.  The families that do carry such a
discipline would then behave differently.  We have not isolated that
mechanism, and we state the measurement and the candidate explanation
separately rather than presenting the second as established.  For a
target that leaves a ramp-independent remainder, switched capacitance
orders designs correctly to the extent the remainder is common to the
designs compared, which is a narrower claim than the proxy condition
suggests, and Section~\ref{sec:discussion} records it as such.
\section{Discussion and Summary}
\label{sec:discussion}

Conventional synthesis
optimizes area or delay and reports power as a consequence.  Every
optimization in \tool~is made with respect to energy, and a pass
that saves area at the cost of energy is refused.  The measured record
says this inversion is not cosmetic.  Gate count, device count, and
per-cycle energy moved independently on real circuits throughout the
campaign.  The affine pass at fixpoint produced the fewest gates of any
\texttt{c1355} variant and mapped to the most devices, an accepted
rewrite on \texttt{c880} raised devices while lowering energy, and the
elimination pass bought its largest win at twice the device count.
Proposition~\ref{prop:cost} says why.  The cost is a literal-occurrence
functional, and no structural count predicts it.  Any flow for
pass-transistor energy-recovery targets that selects candidates on area
is, on this evidence, optimizing the wrong quantity, and the common
practice of using area as a proxy for dynamic power in these styles
should be regarded as unsafe rather than approximate.

\tool~computes the Shannon-order term although it is not an optimization objective.
The ledger makes claims separable and decisions falsifiable.  A
saving obtained by reducing activity is not a saving in erasure, the
crossover at which un-computing scratch would become less costly than
erasing it was computed rather than assumed, and quasi-adiabatic CMOS
operates orders of magnitude away from it.  The introduction's
exchange-rate argument and the measured six-order gap are two forms of
one statement.  At present device parameters, switching is where the
energy is, and reversibility earns its place electrically, as the
precondition for recovery, rather than thermodynamically.  Because
$CV^2$ falls with each device generation while the Landauer constant
does not move, the ordering will not hold forever, and circuits
synthesized against the switching term remain valid when the floor
eventually binds, and the reverse does not hold.

It is often the case that a development benchmark set fits its tool.
This can happen because every default was
chosen while looking at those twenty circuits, and the paper's claims
would be circular if they ended there.  They do not.  On twenty
held-out circuits that no development decision touched, with defaults
frozen in advance, fifteen of nineteen improved under at least one
re-synthesis arm at a best-arm median of $0.9149$, against fourteen of
twenty and comparable magnitudes on the development set.  The passes'
value, on this evidence, is a property of the methods rather than of
the circuits they were tuned on.  The same study kept the negatives.
The linear pre-filter confirmed nothing on the held-out set in eighteen
completed cells, the hash tables resisted every arm exactly as their
design intended, and two circuits offered the passes no structure at
all.  The first of those negatives is the one we have since been able to
explain rather than merely report.  A support screen over both sets finds no
pair of outputs with the property that pass consumes, so the zeros are a
statement about the circuits, and on a set built to have the property the
same pass reduces capped switched capacitance by up to twenty per cent
(Sec.~\ref{sec:bdecresults}).  

Generalization has a second half that yielded a clear  result.  
The choice of technology family is the largest single decision a
user of this tool makes, and its ranking was established on the development
set.  Repeated on the held-out circuits at frozen defaults, that ranking
reproduces exactly, nine families of nine, with the non-pipelined
magnitudes drifting by under a sixth.  A designer picking a family from
Table~\ref{tab:families7} would pick the same one on circuits the defaults
never saw.  The same study also turned the ECRL and PAL coincidence from a
reading of two parameter files into a measurement.  Across twenty held-out
circuits the two differ by exactly zero, on energy and on device count
alike, which is a reminder that a metric can rank targets it cannot
actually distinguish, and that the ranking is only as fine as the model
underneath it.

Results tables,
including ours, answer only ``better than the default.''  The
optimality-gap program changes the kind of statement a synthesis paper
can make.  Three circuits are solved exactly, one large benchmark is
proven to attain the floor of the tool's entire cover space, the
bounded rows carry certified ceilings that tighten with budget rather
than opinions, and the negative gaps on the hash circuits localize
those wins to a mechanism outside the covering algebra.  The template is
exactness where it is affordable and a stated bracket where it is not.

\tool~is a bilingual tool
under a byte contract, with an executable oracle in place of
line-by-line review, a cumulative parity matrix inexpensive enough to run at
every change, failures retained as certificates, and a held-out screen
for the failure mode parity cannot see, allowed a domain expert and an
AI implementer to produce, in weeks, a released artifact that a
standard effort model prices at a decade or more of person-effort, with
every reported number reproducible from the release.  The speed is not
the claim.  The claim is that the development rate cost nothing in rigor, and the
numbered defect record, including the defects the fast collaborator
introduced, is the evidence.  The elements that made it safe are stated
in Section~\ref{sec:development} because they transfer.  They are not
specific to synthesis, and we suspect they are close to a minimal
protocol for building verification-grade software with an implementer
whose volume exceeds its own reliability.

The capacitance and resistance constants are calibrated with
stated provenance, not extracted from a characterized cell library, and
the device-level study of Section~\ref{sec:spice} closes only part of
that distance.  The deep-circuit floor it measured is reported as a
property of the mapped circuits on the reference family, with its
mechanism un-isolated, and it narrows the proxy claim to comparisons in
which the ramp-independent remainder is common.  The series bound of
six is an engineering compromise whose value a review of the
pass-transistor literature did not support, and part of the modeled
advantage is purchased with depth a stricter design rule would forbid.
Both sides of every comparison are bounded identically, which is what
licenses the comparisons rather than the absolute figures.  The default
charging convention leaves primary-input drive unpriced, a median
$70\%$ of the energy, for comparability with the published baseline.  The
excluded quantity is always computed and reported, and one of the two
principal comparisons survives the other convention while the second does
not.  Workload-aware synthesis is scoped to the activity-priced
cover, since the default cover consumes no activity.  Observability is
reported as headroom, $31$ to $38\%$ of total energy on the two
circuits measured, not as a saving, because the gating discipline that
would harvest it in a clocked dual-rail family was not established.
The linear pre-filter's first-round ranking is uninformative and the pass is
consequently expensive.  Because the move space grows as $m^2$ while the
pool does not, its reach shrinks with output count, which we bound with a
measurement rather than leave implicit.  Its scope is a stated property of a
circuit's output supports and not a claim about circuits in general.  And the temporal-independence assumption in the default drive is
an interior point of a characterized family, not a bound, with the
instrument for replacing it built and its synthesis-side leverage still
to be measured properly.

\tool\ receives an
ordinary irreversible netlist and synthesizes it 
into a verified, technology-mapped
energy-recovery circuit with energy as the criterion of every stage.  A
vector-space netlist model makes simulation and justification two
directions of one linear action.  A tag ledger prices switching, erasure, and
observability at their natural R\'enyi orders and keeps the electrical
and thermodynamic accounts separate.  Eight published energy-recovery
families are mapped through one engine from parameters that ship as
data, under a realizability bound enforced as a design rule.  Every
accepted transformation is equivalence-checked against the original netlist
and dominates on two cost tables, so the flow cannot return a circuit worse
than its input.  The acceptance gate's refusals are receipted, its negative
results are published, and its distance from the floor of its own search
space is certified where certification is affordable and bracketed where it
is not.

The tool, its reference and performance implementations, its benchmark sets,
the validation procedure, and the run records behind every figure in this
paper are released together at
\url{https://github.com/mitch-thornton/renesis-low-energy-circuit-synthesis}.
Reproduction, correction, and extension of these results by others are
invited.


\balance
\bibliographystyle{IEEEtran}
\bibliography{renesis_arxiv}

\end{document}

%% file: figures/fam_cells_owner.tex
\providecommand{\famscale}{0.66}

\newcommand{\fnbox}[2]{%
  \draw[line width=0.5pt,rounded corners=1.2pt,fill=black!6]
       ($#1+(-0.36,-0.27)$) rectangle ($#1+(0.36,0.27)$);
  \node[font=\scriptsize] at #1 {$#2$};}

\newcommand{\xcpair}{%
  \node[pmos,anchor=D,xscale=-1] (p1) at (-1.3,2.05) {};
  \node[pmos,anchor=D]           (p2) at ( 1.3,2.05) {};
  \draw (p1.S) -- ++(0,0.26) coordinate (ra);
  \draw (p2.S) -- ++(0,0.26) coordinate (rb);
  \draw (ra) -- (rb);
  \draw ($(ra)!0.5!(rb)$) -- ++(0,0.30) node[anchor=south,font=\scriptsize] {$\phi$};
  \draw (-1.3,2.05) -- (-1.3,1.15) coordinate (yn);
  \draw ( 1.3,2.05) -- ( 1.3,1.15) coordinate (ybn);
  \draw (p1.G) -- (p1.G |- 0,1.52) -- (1.3,1.52);
  \draw (p2.G) -- (p2.G |- 0,1.78) -- (-1.3,1.78);}

\newcommand{\outtaps}[1]{%
  \draw (-1.3,#1) -- (-2.02,#1) node[ocirc]{}
        node[anchor=east,font=\tiny,xshift=-2.5pt] {$y$};
  \draw ( 1.3,#1) -- ( 2.02,#1) node[ocirc]{}
        node[anchor=west,font=\tiny,xshift=2.5pt] {$\bar{y}$};}

\begin{figure*}[t]
\centering
\setlength{\tabcolsep}{1pt}
\renewcommand{\arraystretch}{1.0}
\begin{tabular}{cccc}

\begin{circuitikz}[scale=\famscale,transform shape,line width=0.55pt]
  \node[nmos,anchor=D]           (n) at (-0.55,2.55) {};
  \node[pmos,anchor=S,xscale=-1] (p) at ( 0.55,2.55) {};
  \draw (-0.55,2.55) -- (-0.55,2.95) -- (0.55,2.95) -- (0.55,2.55);
  \draw (0,2.95) -- (0,3.28) node[anchor=south,font=\scriptsize] {$\phi$};
  \draw (n.S) -- (-0.55,1.55) -- (0.55,1.55) -- (p.D);
  \draw (0,1.55) -- (0,1.10) node[ocirc]{}
        node[anchor=north,font=\tiny,yshift=-2.5pt] {$y$};
  \node[font=\tiny,anchor=east] at ($(n.G)+(-0.05,0)$) {$x$};
  \node[font=\tiny,anchor=west] at ($(p.G)+( 0.05,0)$) {$\bar{x}$};
  \node[font=\tiny,gray!55!black] at (0,0.52) {one rail shown, $F_{\mathrm T}$};
\end{circuitikz}
&
\begin{circuitikz}[scale=\famscale,transform shape,line width=0.55pt]
  \node[nmos,anchor=D]           (n) at (-0.55,2.55) {};
  \node[pmos,anchor=S,xscale=-1] (p) at ( 0.55,2.55) {};
  \draw (-0.55,2.55) -- (-0.55,2.95) -- (0.55,2.95) -- (0.55,2.55);
  \draw (0,2.95) -- (0,3.28) node[anchor=south,font=\scriptsize] {$\phi_{k}$};
  \draw (n.S) -- (-0.55,1.55) -- (0.55,1.55) -- (p.D);
  \node[font=\tiny,anchor=east] at ($(n.G)+(-0.05,0)$) {$x$};
  \node[font=\tiny,anchor=west] at ($(p.G)+( 0.05,0)$) {$\bar{x}$};
  \draw (0,1.55) -- (0,1.22);
  \draw[dashed,line width=0.5pt,rounded corners=1.5pt,fill=black!4]
        (-0.86,0.55) rectangle (0.86,1.22);
  \node[font=\tiny] at (0,0.99) {buffer};
  \node[font=\tiny] at (0,0.72) {$\phi_{k+1}$};
  \draw (0,0.55) -- (0,0.18) node[ocirc]{}
        node[anchor=north,font=\tiny,yshift=-2.5pt] {$y$};
\end{circuitikz}
&
\begin{circuitikz}[scale=\famscale,transform shape,line width=0.55pt]
  \node[nmos,anchor=D]           (n) at (-0.95,2.55) {};
  \node[pmos,anchor=S,xscale=-1] (p) at ( 0.95,2.55) {};
  \draw (-0.95,2.55) -- (-0.95,3.00) node[anchor=south,font=\scriptsize] {$\phi_{k}$};
  \draw ( 0.95,2.55) -- ( 0.95,3.00) node[anchor=south,font=\scriptsize] {$\bar{\phi}_{k}$};
  \draw (n.S) -- (-0.95,1.55) -- (0.95,1.55) -- (p.D);
  \draw (0,1.55) -- (0,1.10) node[ocirc]{}
        node[anchor=north,font=\tiny,yshift=-2.5pt] {$y$};
  \node[font=\tiny,anchor=east] at ($(n.G)+(-0.05,0)$) {$x$};
  \node[font=\tiny,anchor=west] at ($(p.G)+( 0.05,0)$) {$\bar{x}$};
  \node[font=\tiny,gray!55!black] at (0,0.52) {static pair};
\end{circuitikz}
&
\begin{circuitikz}[scale=\famscale,transform shape,line width=0.55pt]
  \xcpair
  \fnbox{(-2.15,2.32)}{F_{\mathrm N}}
  \fnbox{( 2.15,2.32)}{\bar{F}_{\mathrm N}}
  \draw (-2.15,2.59) -- ($(ra)+(-0.85,0)$) -- (ra);
  \draw ( 2.15,2.59) -- ($(rb)+( 0.85,0)$) -- (rb);
  \draw (-2.15,2.05) -- (-2.15,1.15) -- (-1.3,1.15);
  \draw ( 2.15,2.05) -- ( 2.15,1.15) -- ( 1.3,1.15);
  \draw (-1.3,1.15) -- (-1.3,0.72) node[ocirc]{}
        node[anchor=north,font=\tiny,yshift=-2.5pt] {$y$};
  \draw ( 1.3,1.15) -- ( 1.3,0.72) node[ocirc]{}
        node[anchor=north,font=\tiny,yshift=-2.5pt] {$\bar{y}$};
\end{circuitikz}
\\[-2pt]
\scriptsize (a) \textsc{tgate}, 4$\phi$ & \scriptsize (b) \textsc{2lal}, 4$\phi$
& \scriptsize (c) \textsc{s2lal}, 8$\phi$ & \scriptsize (d) \textsc{pfal}, 4$\phi$
\\[7pt]

\begin{circuitikz}[scale=\famscale,transform shape,line width=0.55pt]
  \xcpair
  \outtaps{1.15}
  \fnbox{(-1.3,0.66)}{F_{\mathrm N}}
  \fnbox{( 1.3,0.66)}{\bar{F}_{\mathrm N}}
  \draw (-1.3,1.15) -- (-1.3,0.93);
  \draw ( 1.3,1.15) -- ( 1.3,0.93);
  \draw (-1.3,0.39) -- (-1.3,0.12) -- (1.3,0.12) -- (1.3,0.39);
  \draw (0,0.12) -- (0,-0.10) node[ground,scale=0.62]{};
\end{circuitikz}
&
\begin{circuitikz}[scale=\famscale,transform shape,line width=0.55pt]
  \xcpair
  \outtaps{1.15}
  \node[nmos,anchor=D]           (c1) at (-1.3,0.92) {};
  \node[nmos,anchor=D,xscale=-1] (c2) at ( 1.3,0.92) {};
  \draw (-1.3,1.15) -- (-1.3,0.92);
  \draw ( 1.3,1.15) -- ( 1.3,0.92);
  \node[font=\tiny,anchor=east] at ($(c1.G)+(-0.05,0)$) {$cx$};
  \node[font=\tiny,anchor=west] at ($(c2.G)+( 0.05,0)$) {$cx$};
  \fnbox{(-1.3,-0.62)}{F_{\mathrm N}}
  \fnbox{( 1.3,-0.62)}{\bar{F}_{\mathrm N}}
  \draw (c1.S) -- (-1.3,-0.35);
  \draw (c2.S) -- ( 1.3,-0.35);
  \draw (-1.3,-0.89) -- (-1.3,-1.16) -- (1.3,-1.16) -- (1.3,-0.89);
  \draw (0,-1.16) -- (0,-1.38) node[ground,scale=0.62]{};
\end{circuitikz}
&
\begin{circuitikz}[scale=\famscale,transform shape,line width=0.55pt]
  \xcpair
  \outtaps{1.15}
  \fnbox{(-1.3,0.66)}{F_{\mathrm N}}
  \fnbox{( 1.3,0.66)}{\bar{F}_{\mathrm N}}
  \draw (-1.3,1.15) -- (-1.3,0.93);
  \draw ( 1.3,1.15) -- ( 1.3,0.93);
  \draw (-1.3,0.39) -- (-1.3,0.14) node[anchor=north,font=\tiny,yshift=-1pt] {$x$};
  \draw ( 1.3,0.39) -- ( 1.3,0.14) node[anchor=north,font=\tiny,yshift=-1pt] {$\bar{x}$};
  
\end{circuitikz}
&
\begin{circuitikz}[scale=\famscale,transform shape,line width=0.55pt]
  \xcpair
  \outtaps{1.15}
  \fnbox{(-1.3,0.66)}{F_{\mathrm T}}
  \fnbox{( 1.3,0.66)}{\bar{F}_{\mathrm T}}
  \draw (-1.3,1.15) -- (-1.3,0.93);
  \draw ( 1.3,1.15) -- ( 1.3,0.93);
  \draw (-1.3,0.39) -- (-1.3,0.14) node[anchor=north,font=\tiny,yshift=-1pt] {$x$};
  \draw ( 1.3,0.39) -- ( 1.3,0.14) node[anchor=north,font=\tiny,yshift=-1pt] {$\bar{x}$};
  \draw[<->,gray!60,line width=0.45pt,>=stealth] (-0.90,0.66) -- (0.90,0.66);
  \node[font=\tiny,gray!55!black,fill=white,inner sep=0.8pt] at (0,0.66) {balanced};
\end{circuitikz}
\\[-2pt]
\scriptsize (e) \textsc{ecrl}, 4$\phi$ & \scriptsize (f) \textsc{cal}, 2$\phi{+}cx$
& \scriptsize (g) \textsc{pal}, 2$\phi$ & \scriptsize (h) \textsc{spgal}, 2$\phi$
\\
\end{tabular}
\caption{One representative cell from each of the eight energy-recovery
families \tool\ can target, with the phase discipline each family requires
given beneath it.  The power clock enters at the top of every cell.  $F$ and
$\bar{F}$ denote the dual-rail series--parallel function network, which is
structurally the same in every family and is what the mapper actually
builds.  Only the devices that distinguish a family are drawn explicitly, and
the subscript gives the network's device type, $\mathrm{N}$ for nMOS pass
devices and $\mathrm{T}$ for full transmission gates.  Cells (a)--(c) are
pass networks, and one rail of the dual-rail pair is shown.  The
complementary rail is identical with $x$ and $\bar{x}$ exchanged.  Cells
(d)--(h) share a cross-coupled \textsc{pmos} pair and differ in where the
function network attaches: in parallel with the latch device in (d), as a
pull-down to ground in (e) and (f), and as a dual-rail pass network in (g)
and (h).  Cell (f) additionally gates each branch with an auxiliary clock
$cx$.  Cells (b) and (c) are pipelined, so a consumer more than one level
downstream requires the explicit buffer stage drawn dashed in (b).
\textbf{Several of these cells coincide as drawings, and deliberately so.}
Cells (g) and (h) have the same topology and are separated by the device
type inside the network and by the latch device count.  Cells (e) and (g)
differ only in where the network terminates, at ground for (e) and at the
dual-rail inputs for (g).  A cell schematic is simply not where these
families separate.  Table~\ref{tab:famrecords} is, and
Sec.~\ref{sec:family-records} reads the consequence off it.}
\label{fig:family-cells}
\end{figure*}

%% file: tab_famrecords.tex
\begin{tabular}{lrrrrcrr}
\toprule
family & $\phi$ & ovhd & self & clk & dev & fF & residue \\
\midrule
\textsc{tgate}$^\ast$ & 4 & 0 & 0 & 0 & T & 1.70 & --- \\
2\textsc{lal} & 4 & 0 & 0 & 0 & T & 1.70 & --- \\
\textsc{s2lal} & 8 & 0 & 0 & 0 & T & 1.70 & --- \\
\textsc{pfal} & 4 & 4 & 2 & 0 & N & 0.85 & --- \\
\textsc{ecrl} & 4 & 2 & 1 & 0 & N & 0.85 & 0.101 \\
\textsc{cal} & 2 & 4 & 1 & 2 & N & 0.85 & 0.101 \\
\textsc{pal} & 2 & 2 & 1 & 0 & N & 0.85 & --- \\
\textsc{spgal} & 2 & 4 & 2 & 0 & T & 1.70 & --- \\
\bottomrule
\end{tabular}

%% file: figures/fig_pipeline_v5.tex
%
%
\begin{figure}[t]
\centering
\begin{tikzpicture}[
  font=\scriptsize, line width=0.5pt,
  node distance=3.2mm,
  stage/.style={draw, rounded corners=1.2pt, align=center,
                text width=0.46\columnwidth, inner sep=2.2pt, fill=black!3},
  opt/.style={stage, dashed, fill=black!1},
  note/.style={align=left, text width=0.21\columnwidth, font=\tiny,
               text=black!62},
  ar/.style={-{Stealth[length=1.5mm]}, line width=0.45pt},
]

\node[stage] (fe) {\textbf{Front end}\\[0.4pt]
  \texttt{.v .isc .pla .aig .aag .blif .bench}};
\node[opt,   below=of fe]  (prep) {structural preparation\\[0.4pt]
  hash $\to$ balance $\to$ rewrite};
\node[opt,   below=of prep] (res) {re-synthesis passes\\[0.4pt]
  \texttt{davio $\to$ prefix $\to$ linwin $\to$ mowin}};
\node[stage, below=of res] (tags) {\textbf{Forward tag sweep} $(p_1,\alpha)$\\[0.4pt]
  computed only when the cover can read it};
\node[opt,   below=of tags] (bdec) {linear pre-filter \texttt{-{}-bdec}\\[0.4pt]
  replaces the stage below};
\node[stage, below=of bdec] (cover) {\textbf{Cover and technology map}\\[0.4pt]
  technology-priced (default) $|$ activity-priced\\[0.4pt]
  PO-rooted extraction, fanout recovery,\\[0.4pt]
  dead-block elimination};
\node[stage, below=of cover] (cap) {\textbf{Buffer insertion} to series bound\\[0.4pt]
  a design rule, never a cost term};
\node[stage, below=of cap] (er) {\textbf{Energy report, twice}\\[0.4pt]
  $T_1$ uncapped $\;\cdot\;$ $T_2$ capped};
\node[stage, below=of er] (out) {\textbf{Record and emit}\\[0.4pt]
  independent netlist $\cdot$ mapped \texttt{.tgn} $\cdot$ mapped Verilog};

\foreach \a/\b in {fe/prep, prep/res, res/tags, tags/bdec, bdec/cover,
                   cover/cap, cap/er, er/out}
  \draw[ar] (\a) -- (\b);

\node[note, right=1.4mm of tags] (n1)
  {inert under the default\\cover (\S\ref{sec:forward})};
\node[note, right=1.4mm of cover] (n2)
  {the one \textbf{backward} pass\\the flow depends on:\\reverse reachability,\\not justification\\(\S\ref{sec:backward})};
\node[note, right=1.4mm of er] (n3)
  {both reported; neither\\is the other's proxy};
\node[note, left=1.4mm of res, text width=0.17\columnwidth, align=right] (n4)
  {all passes off\\by default};


\end{tikzpicture}
\caption{The flow as the released tool executes it.  Dashed stages are off by
default and can be enabled with command line flag settings or within the UI: 
a default run parses, covers, maps, inserts buffers to the series
bound, and reports two energies.  The mapped network is equivalence-checked
after mapping and again after buffer insertion.  The
forward tag sweep is skipped entirely when the selected cover cannot consume
it, which is the shipped default, so the trial-count options are inert on a
default run.  The technology-priced cover is the default, and the activity-priced
one is the alternative and is the only consumer of the tags, and everything
above the cover is technology-independent.  And the backward
traversal the flow genuinely depends on is the primary-output-rooted extraction
inside the cover, which is reverse reachability over the chosen cuts, not
the justification pass of Eq.~\eqref{eq:justification}, which the
combinational flow does not execute at all.}
\label{fig:pipeline}
\end{figure}

%% file: figures/ex_frontend.tex
\begin{figure}[t]
\centering
\begin{tikzpicture}[scale=0.76, transform shape,
  ag/.style={and gate US, draw, logic gate inputs=nn, scale=0.58, fill=black!3},
  og/.style={or gate US, draw, logic gate inputs=nn, scale=0.58, fill=black!3},
  ng/.style={not gate US, draw, scale=0.58, fill=black!3},
  pin/.style={font=\scriptsize},
  nl/.style={font=\tiny, inner sep=1pt}]

  \node[nl] at (1.5,3.55) {before: six gates};
  \node[pin] (a) at (-0.15,3.05) {$a$};
  \node[pin] (b) at (-0.15,2.55) {$b$};
  \node[pin] (c) at (-0.15,1.05) {$c$};
  \node[ag] (g1) at (1.25,2.95) {};
  \node[ag] (g2) at (1.25,2.10) {};
  \node[ng] (g3) at (1.25,1.35) {};
  \node[ng] (g4) at (2.25,1.35) {};
  \node[og] (g5) at (3.45,2.55) {};
  \node[og] (g6) at (3.45,1.55) {};
  \coordinate (sa) at (0.45,3.05); \coordinate (sb) at (0.45,2.55);
  \coordinate (sc) at (0.45,1.05);
  \draw (a) -- (sa); \draw (b) -- (sb); \draw (c) -- (sc);
  \draw (sa) |- (g1.input 1); \draw (sb) |- (g1.input 2);
  \draw (sa) |- (g2.input 1); \draw (sb) |- (g2.input 2);
  \draw (sc) |- (g3.input);
  \draw (g3.output) -- (g4.input);
  \draw (g1.output) -| (2.85,2.95) |- (g5.input 1);
  \draw (g4.output) -| (2.95,1.35) |- (g5.input 2);
  \draw (g2.output) -| (2.75,2.10) |- (g6.input 1);
  \draw (sc) -| (3.05,1.05) |- (g6.input 2);
  \draw (g5.output) -- ++(0.42,0) node[right,pin]{$y_1$};
  \draw (g6.output) -- ++(0.42,0) node[right,pin]{$y_2$};
  \node[nl,above=0.5pt] at (g2.north) {duplicate};
  \node[nl,below=1pt] at (g4.south) {double negation};

  \draw[-{Latex[length=2.2mm]}, thick] (5.05,2.2) -- (5.85,2.2);
  \node[nl,above] at (5.45,2.25) {strash};

  \node[nl] at (7.6,3.55) {after: three gates};
  \node[pin] (a2) at (6.15,3.05) {$a$};
  \node[pin] (b2) at (6.15,2.55) {$b$};
  \node[pin] (c2) at (6.15,1.35) {$c$};
  \node[ag] (h1) at (7.35,2.80) {};
  \node[og] (h2) at (8.85,2.45) {};
  \node[og] (h3) at (8.85,1.55) {};
  \coordinate (t)  at (7.95,2.80); \coordinate (sc2) at (6.65,1.35);
  \draw (a2) |- (h1.input 1); \draw (b2) |- (h1.input 2);
  \draw (c2) -- (sc2);
  \draw (h1.output) -- (t);
  \draw (t) |- (h2.input 1); \draw (t) |- (h3.input 1);
  \draw (sc2) |- (h2.input 2); \draw (sc2) |- (h3.input 2);
  \draw (h2.output) -- ++(0.42,0) node[right,pin]{$y_1$};
  \draw (h3.output) -- ++(0.42,0) node[right,pin]{$y_2$};
\end{tikzpicture}
\caption{Structural preparation (Sec.~\ref{sec:frontend}) on a fixture built
to contain the two redundancies it removes.  Hash-consing on the function
and the content-sorted input tuple merges the duplicated conjunction, and
the double negation collapses to a wire, taking six gates to three.  The
stage is nonetheless \emph{off by default}, and this fixture shows why the
default is defensible rather than merely cautious: mapped and priced, the
two netlists above are identical, at $12$ devices and $0.008228$\,pJ on both
tables, because the cover of Sec.~\ref{sec:cover} absorbs the redundancy on
its own.  Preparation changes the input netlist rather than the
optimization, so leaving it on would make every reported ratio a statement
about two changes at once.  Measured, its usual benefit is a faster run
rather than a less costly circuit.}
\label{fig:ex-frontend}
\end{figure}

%% file: figures/ex_c17.tex
\begin{figure}[t]
\centering
\begin{tikzpicture}[scale=0.80, transform shape,
  nd/.style={nand gate US, draw, logic gate inputs=nn, scale=0.62,
             fill=black!3},
  pin/.style={font=\scriptsize},
  nl/.style={font=\tiny, inner sep=1.2pt},
  stem/.style={circle, draw, dashed, thick, gray, inner sep=1.4pt}]

  \node[pin] (i1) at (-0.15,3.30) {$1$};
  \node[pin] (i3) at (-0.15,2.55) {$3$};
  \node[pin] (i6) at (-0.15,1.80) {$6$};
  \node[pin] (i2) at (-0.15,1.05) {$2$};
  \node[pin] (i7) at (-0.15,0.30) {$7$};

  \node[nd] (g10) at (1.55,3.00) {};
  \node[nd] (g11) at (1.55,2.10) {};
  \node[nd] (g16) at (3.05,1.35) {};
  \node[nd] (g19) at (3.05,0.55) {};
  \node[nd] (g22) at (4.60,2.35) {};
  \node[nd] (g23) at (4.60,0.95) {};

  \coordinate (s3) at (0.55,2.55);
  \draw (i3) -- (s3);
  \draw (s3) |- (g10.input 2);
  \draw (s3) |- (g11.input 1);
  \draw (i1) -| (0.75,3.30) |- (g10.input 1);
  \draw (i6) -| (0.75,1.80) |- (g11.input 2);

  \coordinate (s11) at (2.20,2.10);
  \draw (g11.output) -- (s11);
  \draw (s11) |- (g16.input 2);
  \draw (s11) |- (g19.input 1);
  \draw (i2) -| (2.40,1.05) |- (g16.input 1);
  \draw (i7) -| (2.55,0.30) |- (g19.input 2);

  \coordinate (s16) at (3.60,1.35);
  \draw (g16.output) -- (s16);
  \draw (s16) |- (g22.input 2);
  \draw (s16) |- (g23.input 1);
  \draw (g10.output) -| (3.95,3.00) |- (g22.input 1);
  \draw (g19.output) -| (3.80,0.55) |- (g23.input 2);
  \draw (g22.output) -- ++(0.42,0) node[right,pin]{$22$};
  \draw (g23.output) -- ++(0.42,0) node[right,pin]{$23$};

  \node[nl,above=1pt] at (g10.north) {$10$};
  \node[nl,below=1pt] at (g11.south) {$11$};
  \node[nl,above=1pt] at (g16.north) {$16$};
  \node[nl,below=1pt] at (g19.south) {$19$};
  \node[stem] at (s3) {}; \node[stem] at (s11) {};
\end{tikzpicture}

\smallskip
{\footnotesize
\begin{tabular}{lrrr}
\toprule
net & exact $p_1$ & independence & error \\
\midrule
$10$ & $0.75000$ & $0.75000$ & $0.00\%$ \\
$11$ & $0.75000$ & $0.75000$ & $0.00\%$ \\
$16$ & $0.62500$ & $0.62500$ & $0.00\%$ \\
$19$ & $0.62500$ & $0.62500$ & $0.00\%$ \\
$22$ & $0.56250$ & $0.53125$ & $\mathbf{+5.56\%}$ \\
$23$ & $0.56250$ & $0.60938$ & $\mathbf{-8.33\%}$ \\
\bottomrule
\end{tabular}}

\caption{The running example, \texttt{c17}, and the forward tags of
Sec.~\ref{sec:forward}.  The exact column enumerates all $32$ input
vectors.  The independence column is the per-gate product rule.  They agree
to every printed digit on the four nets whose fan-ins are independent and
disagree only at nets $22$ and $23$, which are exactly the gates downstream
of a reconvergent stem: net $3$ reaches net $22$ along two paths and net
$11$ reaches net $23$ along two paths (circled).  The two errors carry
\emph{opposite} signs, so no single calibration constant absorbs them, which
is the whole case for measuring tags rather than propagating them
analytically.}
\label{fig:ex-tags}
\end{figure}

%% file: figures/ex_cover.tex
\begin{figure}[t]
\centering
\begin{tikzpicture}[scale=0.80, transform shape,
  nd/.style={nand gate US, draw, logic gate inputs=nn, scale=0.62,
             fill=black!3},
  pin/.style={font=\scriptsize},
  nl/.style={font=\tiny, inner sep=1.2pt},
  cone/.style={draw, dashed, rounded corners=3pt, thick}]

  \node[pin] (i1) at (-0.15,3.30) {$1$};
  \node[pin] (i3) at (-0.15,2.55) {$3$};
  \node[pin] (i6) at (-0.15,1.80) {$6$};
  \node[pin] (i2) at (-0.15,1.05) {$2$};
  \node[pin] (i7) at (-0.15,0.30) {$7$};

  \node[nd] (g10) at (1.55,3.00) {};
  \node[nd] (g11) at (1.55,2.10) {};
  \node[nd] (g16) at (3.05,1.35) {};
  \node[nd] (g19) at (3.05,0.55) {};
  \node[nd] (g22) at (4.60,2.35) {};
  \node[nd] (g23) at (4.60,0.95) {};

  \coordinate (s3) at (0.55,2.55);
  \draw (i3) -- (s3); \draw (s3) |- (g10.input 2); \draw (s3) |- (g11.input 1);
  \draw (i1) -| (0.75,3.30) |- (g10.input 1);
  \draw (i6) -| (0.75,1.80) |- (g11.input 2);
  \coordinate (s11) at (2.20,2.10);
  \draw (g11.output) -- (s11);
  \draw (s11) |- (g16.input 2); \draw (s11) |- (g19.input 1);
  \draw (i2) -| (2.40,1.05) |- (g16.input 1);
  \draw (i7) -| (2.55,0.30) |- (g19.input 2);
  \coordinate (s16) at (3.60,1.35);
  \draw (g16.output) -- (s16);
  \draw (s16) |- (g22.input 2); \draw (s16) |- (g23.input 1);
  \draw (g10.output) -| (3.95,3.00) |- (g22.input 1);
  \draw (g19.output) -| (3.80,0.55) |- (g23.input 2);
  \draw (g22.output) -- ++(0.42,0) node[right,pin]{$22$};
  \draw (g23.output) -- ++(0.42,0) node[right,pin]{$23$};
  \node[nl,above=1pt] at (g10.north) {$10$};
  \node[nl,below=1pt] at (g11.south) {$11$};
  \node[nl,above=1pt] at (g16.north) {$16$};
  \node[nl,below=1pt] at (g19.south) {$19$};

  \node[cone, fit=(g10)(g11)(g16)(g22), inner sep=5pt,
        label={[font=\tiny]above right:block $22$}] {};
  \node[cone, draw=black!55, fit=(g11)(g16)(g19)(g23), inner sep=8pt,
        label={[font=\tiny,black!55]below right:block $23$}] {};
\end{tikzpicture}

\smallskip
{\scriptsize
\begin{minipage}{0.96\columnwidth}
\raggedright
\texttt{g 22gat ph0 POS=P(S(+1,+3),S(+2,P(-3,-6)))}\\
\texttt{\phantom{g 22gat ph0 }NEG=S(P(-1,-3),P(-2,S(+3,+6)))}\\[2pt]
\texttt{g 23gat ph0 POS=P(S(+2,P(-3,-6)),S(P(-3,-6),+7))}\\
\texttt{\phantom{g 23gat ph0 }NEG=S(P(-2,S(+3,+6)),P(S(+3,+6),-7))}
\end{minipage}}

\caption{Covering (Sec.~\ref{sec:cover}) on the running example.  The six
NAND gates are absorbed into \emph{two} mapped blocks, one per primary
output, with cuts $\{1,2,3,6\}$ and $\{2,3,6,7\}$.  No internal net survives
as a block of its own.  Net $16$ lies inside both cones, so its logic is
duplicated rather than shared, a decision the cover's duplication discount
priced and accepted.  Below is the tool's own output for the two blocks,
written as series--parallel conduction expressions over signed literals:
\texttt{P} is parallel, \texttt{S} is series, and a sign is the literal's
polarity.  Reading the first: net $22$ conducts on its positive rail when
$(1\wedge3)$ or $(2\wedge(\bar3\vee\bar6))$, which is $\overline{10\wedge16}$
rewritten by De Morgan, and the two rails are complementary by
construction.  Counting literal occurrences gives ten pass devices here and
twelve in block $23$, the twenty-two the energy model bills.}
\label{fig:ex-cover}
\end{figure}

%% file: figures/ex_map.tex
\begin{figure}[t]
\centering
\resizebox{\columnwidth}{!}{\input{figures/c17_node22_gen}}
\caption{Technology mapping (Sec.~\ref{sec:map}) applied to block $22$ of
Fig.~\ref{fig:ex-cover}: the positive-rail pull network and its complement,
each drawn from power-clock phase $\phi_0$ to one rail of the node.  Every
device is a CMOS transmission gate driven by the named rail, and an overline
marks a literal that conducts when its signal is low.  The series--parallel
structure is exactly the conduction expression printed in
Fig.~\ref{fig:ex-cover}: the positive rail is a parallel pair of series
branches, and the negative rail is its DeMorgan dual, series of parallel.
Ten pass devices appear here, twelve in the circuit's other block, which is
the twenty-two the energy model bills.  This figure is generated from the
mapped netlist by the released tool rather than drawn, so it cannot drift
from what the mapper builds.}
\label{fig:ex-map}
\end{figure}

%% file: figures/c17_node22_gen.tex
\begin{circuitikz}[scale=1.0, transform shape, every node/.style={inner sep=1pt}]
  \draw (-16.696,4.050) to[Tnmos, n=tgP1] (-16.696,2.025);
  \draw (tgP1.G) -- ++(0.62,0)
      node[anchor=west, font=\scriptsize, inner sep=2pt] {$1gat$};
  \draw (-16.696,2.025) to[Tnmos, n=tgP2] (-16.696,0.000);
  \draw (tgP2.G) -- ++(0.62,0)
      node[anchor=west, font=\scriptsize, inner sep=2pt] {$3gat$};
  \draw (-9.898,4.050) to[Tnmos, n=tgP3] (-9.898,2.025);
  \draw (tgP3.G) -- ++(0.62,0)
      node[anchor=west, font=\scriptsize, inner sep=2pt] {$2gat$};
  \draw (-12.164,2.025) to[Tnmos, n=tgP4] (-12.164,0.000);
  \draw (tgP4.G) -- ++(0.62,0)
      node[anchor=west, font=\scriptsize, inner sep=2pt] {$\overline{3gat}$};
  \draw (-7.631,2.025) to[Tnmos, n=tgP5] (-7.631,0.000);
  \draw (tgP5.G) -- ++(0.62,0)
      node[anchor=west, font=\scriptsize, inner sep=2pt] {$\overline{6gat}$};
  \draw (-12.164,2.025) -- (-7.631,2.025);
  \draw (-12.164,0.000) -- (-7.631,0.000);
  \draw (-9.898,2.025) -- (-12.164,2.025);
  \draw (-9.898,0.000) -- (-12.164,0.000);
  \draw (-16.696,4.050) -- (-9.898,4.050);
  \draw (-16.696,0.000) -- (-9.898,0.000);
  \draw (-12.164,4.050) -- (-16.696,4.050);
  \draw (-12.164,0.000) -- (-16.696,0.000);
  \draw (-12.164,4.050) -- (-12.164,4.793) node[anchor=south, font=\small] {$\phi_{0}$};
  \draw (-12.164,0.000) -- (-12.164,-0.743) node[anchor=north, font=\small] {22gat$_{p}$};
  \node[font=\footnotesize\itshape] at (-12.164,5.603) {positive rail};
  \draw (9.898,4.050) to[Tnmos, n=tgN1] (9.898,2.700);
  \draw (tgN1.G) -- ++(0.62,0)
      node[anchor=west, font=\scriptsize, inner sep=2pt] {$\overline{1gat}$};
  \draw (14.430,4.050) to[Tnmos, n=tgN2] (14.430,2.700);
  \draw (tgN2.G) -- ++(0.62,0)
      node[anchor=west, font=\scriptsize, inner sep=2pt] {$\overline{3gat}$};
  \draw (9.898,4.050) -- (14.430,4.050);
  \draw (9.898,2.700) -- (14.430,2.700);
  \draw (12.164,4.050) -- (9.898,4.050);
  \draw (12.164,2.700) -- (9.898,2.700);
  \draw (9.898,2.700) to[Tnmos, n=tgN3] (9.898,0.000);
  \draw (tgN3.G) -- ++(0.62,0)
      node[anchor=west, font=\scriptsize, inner sep=2pt] {$\overline{2gat}$};
  \draw (14.430,2.700) to[Tnmos, n=tgN4] (14.430,1.350);
  \draw (tgN4.G) -- ++(0.62,0)
      node[anchor=west, font=\scriptsize, inner sep=2pt] {$3gat$};
  \draw (14.430,1.350) to[Tnmos, n=tgN5] (14.430,0.000);
  \draw (tgN5.G) -- ++(0.62,0)
      node[anchor=west, font=\scriptsize, inner sep=2pt] {$6gat$};
  \draw (9.898,2.700) -- (14.430,2.700);
  \draw (9.898,0.000) -- (14.430,0.000);
  \draw (12.164,2.700) -- (9.898,2.700);
  \draw (12.164,0.000) -- (9.898,0.000);
  \draw (12.164,4.050) -- (12.164,4.793) node[anchor=south, font=\small] {$\phi_{0}$};
  \draw (12.164,0.000) -- (12.164,-0.743) node[anchor=north, font=\small] {22gat$_{n}$};
  \node[font=\footnotesize\itshape] at (12.164,5.603) {negative rail};
\end{circuitikz}

%% file: figures/ex_cap.tex
\begin{figure}[t]
\centering
\begin{tikzpicture}[scale=0.78, transform shape,
  dev/.style={draw, minimum width=5.5mm, minimum height=4mm, font=\tiny,
              inner sep=0.5pt, fill=black!3},
  st/.style={draw, minimum width=5.5mm, minimum height=4mm, font=\tiny,
             inner sep=0.5pt, fill=black!12},
  lb/.style={font=\scriptsize},
  ph/.style={font=\scriptsize}]

  \node[lb] at (-1.15,4.6) {before};
  \draw (0,5.05) node[above,ph]{$\phi$} -- (0,4.85);
  \foreach \i/\n in {0/a,1/b,2/c,3/d,4/e,5/f,6/g,7/h}
    \node[dev] (d\i) at (0,4.55-0.52*\i) {$\n$};
  \draw (0,4.85) -- (d0);
  \foreach \i [evaluate=\i as \j using int(\i+1)] in {0,...,6} \draw (d\i) -- (d\j);
  \draw (d7) -- (0,0.35) node[below,ph]{$y$};
  \draw[<->,gray] (0.55,4.7) -- node[right,lb]{$\delta=8$} (0.55,0.5);

  \node[lb] at (3.05,4.6) {after, $c=3$};
  \draw (2.4,5.05) node[above,ph]{$\phi$} -- (2.4,4.85);
  \node[dev] (e0) at (2.4,4.55) {$a$};
  \node[dev] (e1) at (2.4,4.03) {$b$};
  \node[dev] (e2) at (2.4,3.51) {$c$};
  \draw (2.4,4.85) -- (e0); \draw (e0) -- (e1); \draw (e1) -- (e2);
  \draw (e2) -- (2.4,3.05) node[right=1pt,lb]{$s_1$};
  \draw (4.35,5.05) node[above,ph]{$\phi$} -- (4.35,4.85);
  \node[st]  (f0) at (4.35,4.55) {$s_1$};
  \node[dev] (f1) at (4.35,4.03) {$d$};
  \node[dev] (f2) at (4.35,3.51) {$e$};
  \draw (4.35,4.85) -- (f0); \draw (f0) -- (f1); \draw (f1) -- (f2);
  \draw (f2) -- (4.35,3.05) node[right=1pt,lb]{$s_2$};
  \draw (3.35,2.55) node[above,ph]{$\phi$} -- (3.35,2.35);
  \node[st]  (g0) at (3.35,2.05) {$s_2$};
  \node[dev] (g1) at (3.35,1.53) {$f$};
  \node[dev] (g2) at (3.35,1.01) {$g$};
  \node[dev] (g3) at (3.35,0.49) {$h$};
  \draw (3.35,2.35) -- (g0); \draw (g0) -- (g1); \draw (g1) -- (g2);
  \draw (g2) -- (g3);
  \draw (g3) -- (3.35,0.05) node[below,ph]{$y$};
  \node[lb,align=left] at (6.05,2.6)
    {two new\\ dual-rail\\ stages,\\ each $\delta\le3$};
\end{tikzpicture}
\caption{Series bounding (Sec.~\ref{sec:cap}) on a conduction chain of eight
pass devices, with the bound set to three.  The chain is partitioned into
consecutive segments, and each segment after the first is composed through
an \emph{accumulator literal} $s_i$, the buffered output of the previous
segment: a segment conducts exactly when the previous stage conducted and
this segment conducts, so the conjunction of the segments is the original
series function.  Every segment carrying an accumulator has budget $c-1$
rather than $c$, which is why three devices become three, two, and three.
Each extracted segment becomes a real dual-rail gate whose opposite rail is
the DeMorgan complement, so rail consistency survives insertion.  The bound
is a design rule and never a cost term: it makes a circuit realizable rather
than less costly, and on the small fixture of Fig.~\ref{fig:ex-davio} imposing
it raises the billed device count from $42$ to $56$.}
\label{fig:ex-cap}
\end{figure}

%% file: figures/ex_davio.tex
\begin{figure}[t]
\centering
\begin{tikzpicture}[scale=0.74, transform shape,
  nd/.style={nand gate US, draw, logic gate inputs=nn, scale=0.58,
             fill=black!3},
  xg/.style={xor gate US, draw, logic gate inputs=nn, scale=0.58,
             fill=black!8},
  pin/.style={font=\scriptsize},
  nl/.style={font=\tiny, inner sep=1pt}]

  \node[pin] (a) at (-0.2,3.05) {$a$};
  \node[pin] (b) at (-0.2,1.55) {$b$};
  \node[pin] (c) at (-0.2,0.15) {$c$};
  \node[nd] (n1) at (1.15,2.30) {};
  \node[nd] (n2) at (2.35,2.85) {};
  \node[nd] (n3) at (2.35,1.75) {};
  \node[nd] (n4) at (3.55,2.30) {};
  \node[nd] (m1) at (4.75,1.35) {};
  \node[nd] (m2) at (5.95,1.90) {};
  \node[nd] (m3) at (5.95,0.80) {};
  \node[nd] (m4) at (7.15,1.35) {};
  \coordinate (sa) at (0.35,3.05); \coordinate (sb) at (0.35,1.55);
  \draw (a) -- (sa); \draw (b) -- (sb);
  \draw (sa) |- (n1.input 1); \draw (sb) |- (n1.input 2);
  \draw (sa) |- (n2.input 1); \draw (sb) |- (n3.input 2);
  \coordinate (s1) at (1.72,2.30);
  \draw (n1.output) -- (s1);
  \draw (s1) |- (n2.input 2); \draw (s1) |- (n3.input 1);
  \draw (n2.output) -| (3.15,2.85) |- (n4.input 1);
  \draw (n3.output) -| (3.15,1.75) |- (n4.input 2);
  \coordinate (st) at (4.05,2.30);
  \draw (n4.output) -- (st);
  \coordinate (sc) at (0.72,0.15);
  \draw (c) -- (sc);
  \draw (st) |- (m1.input 1);
  \draw (sc) -| (4.32,0.15) |- (m1.input 2);
  \draw (st) -- (4.32,2.30);
  \coordinate (s2) at (5.32,1.35);
  \draw (m1.output) -- (s2);
  \draw (s2) |- (m2.input 2); \draw (s2) |- (m3.input 2);
  \draw (st) -| (5.05,2.30) |- (m2.input 1);
  \draw (sc) -| (5.05,0.15) |- (m3.input 1);
  \draw (m2.output) -| (6.72,1.90) |- (m4.input 1);
  \draw (m3.output) -| (6.72,0.80) |- (m4.input 2);
  \draw (m4.output) -- ++(0.45,0) node[right,pin]{$y$};
  \node[nl] at (2.35,3.55) {eight NAND gates};

  \draw[-{Latex[length=2.2mm]}, thick] (3.65,-0.55) -- (4.65,-0.55);
  \node[nl,above] at (4.15,-0.5) {\textsc{davio}};

  \node[pin] (a2) at (1.65,-1.55) {$a$};
  \node[pin] (b2) at (1.65,-2.10) {$b$};
  \node[pin] (c2) at (1.65,-2.75) {$c$};
  \node[xg] (x1) at (3.15,-1.82) {};
  \node[xg] (x2) at (4.65,-2.25) {};
  \draw (a2) |- (x1.input 1); \draw (b2) |- (x1.input 2);
  \draw (x1.output) -| (4.15,-1.82) |- (x2.input 1);
  \draw (c2) -| (4.15,-2.75) |- (x2.input 2);
  \draw (x2.output) -- ++(0.45,0) node[right,pin]{$y$};
  \node[nl] at (3.9,-1.15) {two XOR gates};
\end{tikzpicture}

\smallskip
{\footnotesize
\begin{tabular}{lrrrr}
\toprule
 & gates & devices & $T_1$ (pJ) & $T_2$ (pJ) \\
\midrule
default        & $8$ & $42$ & $0.004114$ & $0.006171$ \\
\texttt{--davio} & $2$ & $20$ & $0.004114$ & $\mathbf{0.004114}$ \\
\midrule
ratio          &     &      & $1.0000$   & $\mathbf{0.6667}$ \\

\bottomrule
\end{tabular}}

\caption{Affine-cut extraction (Sec.~\ref{sec:davio}) on a fixture in which
$a\oplus b\oplus c$ is written as two classic four-NAND clusters.  The pass
never matches that shape.  It takes the cut function, evaluates the Boolean
difference with respect to each variable, finds all three differences
constant at $1$, and concludes by Lemma~\ref{lem:affine} that the cut is
affine with $L=\{a,b,c\}$ and $c_0=f(0,0,0)$, so it re-emits the cut as a
parity tree over exactly those leaves.  Because the test reads the function
rather than the netlist (Corollary~\ref{cor:realization-invariant}), the
same detection fires on a NOR cluster, an AOI form, or a wide tree flattened
by an earlier pass.  The measured effect on this fixture is a fall from
eight gates to two, from $42$ devices to $20$, and from $0.006171$ to
$0.004114$\,pJ on the capped table, a ratio of $0.6667$.  The uncapped table
does not move, so the whole saving here is the series bound no longer having
work to do.}
\label{fig:ex-davio}
\end{figure}

%% file: figures/ex_elim.tex
\begin{figure}[t]
\centering
\begin{tikzpicture}[scale=0.74, transform shape,
  nn/.style={draw, rounded corners=1.5pt, minimum width=13mm,
             minimum height=5mm, font=\scriptsize, fill=black!4},
  lb/.style={font=\scriptsize},
  nl/.style={font=\tiny, inner sep=1pt}]

  \node[lb] at (-0.85,1.9) {before};
  \node[nn] (t) at (0.85,1.9) {$t=ab$};
  \node[nn] (u1) at (3.15,2.45) {$u_1=t\,c$};
  \node[nn] (u2) at (3.15,1.35) {$u_2=t\,d$};
  \draw[-{Latex[length=1.6mm]}] (t) -- (u1);
  \draw[-{Latex[length=1.6mm]}] (t) -- (u2);
  \draw[-{Latex[length=1.6mm]}] (-0.35,1.9) -- (t);
  \draw[-{Latex[length=1.6mm]}] (u1) -- ++(0.55,0);
  \draw[-{Latex[length=1.6mm]}] (u2) -- ++(0.55,0);
  \node[nl] at (2.0,0.85) {$t$ is charged and read twice};

  \draw[-{Latex[length=2.2mm]}, thick] (2.0,0.45) -- (2.0,-0.05);
  \node[nl,right] at (2.1,0.2) {collapse $t$ into its readers};

  \node[lb] at (-0.85,-0.75) {after};
  \node[nn] (v1) at (1.55,-0.30) {$u_1=abc$};
  \node[nn] (v2) at (1.55,-1.25) {$u_2=abd$};
  \draw[-{Latex[length=1.6mm]}] (0.35,-0.30) -- (v1);
  \draw[-{Latex[length=1.6mm]}] (0.35,-1.25) -- (v2);
  \draw[-{Latex[length=1.6mm]}] (v1) -- ++(0.6,0);
  \draw[-{Latex[length=1.6mm]}] (v2) -- ++(0.6,0);
  \node[nl,align=left] at (4.05,-0.78)
    {one charged net removed,\\ one literal occurrence added};
\end{tikzpicture}

\smallskip
{\scriptsize
\setlength{\tabcolsep}{4pt}
\begin{tabular}{lrrrr}
\toprule
\texttt{c880} & devices & capped dev. & $T_1$ (pJ) & $T_2$ (pJ) \\
\midrule
default          & $1202$ & $1338$ & $0.839256$ & $0.888624$ \\
\texttt{--elim single} & $1166$ & $1399$ & $\mathbf{0.732292}$ & $\mathbf{0.863940}$ \\
\midrule
ratio            &        &        & $0.8725$   & $0.9722$ \\
\bottomrule
\end{tabular}}

\caption{Bounded elimination (Sec.~\ref{sec:elim}).  A node whose output is
charged and read by two consumers is collapsed into them when that pays for
itself in literal occurrences: the internal net disappears from the ledger,
at the cost of one extra occurrence at each reader.  Whether that is a win
is arithmetic in the cost model of Eq.~\eqref{eq:cost} and not a matter of
gate count.  The measured effect on \texttt{c880} is a fall to $0.8725$ of
the uncapped figure and $0.9722$ of the capped one, and the device columns
show why no structural count could have predicted it: the pass \emph{lowers}
the uncapped device count from $1202$ to $1166$ while \emph{raising} the
capped count from $1338$ to $1399$, because the collapsed cubes are deeper
and the series bound then inserts more stages.  Energy falls on both tables
regardless, which is the only test the gate applies.}
\label{fig:ex-elim}
\end{figure}

%% file: figures/ex_prefix.tex
\begin{figure}[t]
\centering
\begin{tikzpicture}[scale=0.78, transform shape,
  op/.style={draw, circle, inner sep=0pt, minimum size=3.6mm, font=\tiny,
             fill=black!5},
  lb/.style={font=\scriptsize},
  nl/.style={font=\tiny, inner sep=1pt}]

  \node[lb] at (-0.75,2.55) {before};
  \foreach \i in {0,...,5} \node[op] (a\i) at (0.72*\i,2.55) {};
  \foreach \i [evaluate=\i as \j using int(\i+1)] in {0,...,4}
     \draw[-{Latex[length=1.6mm]}] (a\i) -- (a\j);
  \draw[-{Latex[length=1.6mm]}] (a5) -- ++(0.5,0) node[right,nl]{$c_{\mathrm{out}}$};
  \foreach \i in {0,...,5}
     \draw[-{Latex[length=1.6mm]}] (0.72*\i,3.15) node[above,nl]{$g_\i p_\i$} -- (a\i);
  \node[nl] at (1.8,1.95) {serial: depth grows with the chain};

  \draw[-{Latex[length=2.2mm]}, thick] (1.8,1.55) -- (1.8,1.05);
  \node[nl,right] at (1.9,1.3) {treeify, then re-window};

  \node[lb] at (-0.75,0.35) {after};
  \foreach \i in {0,...,5} \node[op] (b\i) at (0.72*\i,0.35) {};
  \foreach \i in {0,...,5}
     \draw[-{Latex[length=1.6mm]}] (0.72*\i,0.95) node[above,nl]{$g_\i p_\i$} -- (b\i);
  \foreach \i in {0,...,5} \node[op] (c\i) at (0.72*\i,-0.55) {};
  \foreach \i in {0,...,5} \node[op] (d\i) at (0.72*\i,-1.45) {};
  \foreach \i in {0,...,5} \draw[-{Latex[length=1.6mm]}] (b\i) -- (c\i);
  \foreach \i in {0,...,5} \draw[-{Latex[length=1.6mm]}] (c\i) -- (d\i);
  \foreach \i [evaluate=\i as \j using int(\i+1)] in {0,2,4}
     \draw[-{Latex[length=1.6mm]}] (b\i) -- (c\j);
  \foreach \i/\j in {0/2,1/3,2/4,3/5}
     \draw[-{Latex[length=1.6mm]},gray] (c\i) -- (d\j);
  \draw[-{Latex[length=1.6mm]}] (d5) -- ++(0.5,0) node[right,nl]{$c_{\mathrm{out}}$};
  \node[nl] at (1.8,-2.05) {prefix tree: depth $O(\log n)$, more gates};
\end{tikzpicture}

\smallskip
{\footnotesize
\begin{tabular}{lrr}
\toprule
state of the move & $T_1$ ratio & $T_2$ ratio \\
\midrule
treeified only            & $1.1333$ & $1.1500$ \\
treeified $+$ re-windowed & $1.0000$ & $1.0833$ \\
\midrule
gate verdict & \multicolumn{2}{r}{\emph{rejected, input returned}} \\
\bottomrule
\end{tabular}}

\caption{Parallel-prefix restructuring (Sec.~\ref{sec:prefix}) on the shipped
eight-bit ripple-carry fixture \texttt{rca8}.  The pass detects one carry
chain of length $16$, rebuilds it as a Brent--Kung all-prefix network, and
re-windows the result, pricing $512$ candidates of which $5$ are accepted
and $210$ are skipped by the overlap guard.  The table is the whole argument
for gating the move as a \emph{compound}: treeification alone is worse on
both tables, and re-windowing recovers most but not all of the loss, so the
two-table gate refuses the move and the tool returns the input unchanged.
That is a result rather than a failure.  On a circuit with denser
reconvergence around the chain the same mechanism is accepted: on
\texttt{c432} the recorded triple is $0.707608$\,pJ before, $0.761090$\,pJ
treeified, and $0.580074$\,pJ after re-windowing, a ratio of $0.8198$.
Because the intermediate state is worse than the input in both cases, gating
the two steps separately would reject the move at the first step and the
\texttt{c432} result would never be reachable.}
\label{fig:ex-prefix}
\end{figure}

%% file: figures/ex_window.tex
\begin{figure}[t]
\centering
\begin{tikzpicture}[scale=0.80, transform shape,
  bx/.style={draw, rounded corners=2pt, font=\scriptsize, align=center,
             fill=black!4, minimum height=6.5mm},
  enc/.style={bx, fill=black!10},
  lb/.style={font=\scriptsize},
  nl/.style={font=\tiny, inner sep=1pt},
  win/.style={draw, dashed, rounded corners=3pt, thick}]

  \node[lb] at (-0.75,1.55) {before};
  \node[bx, minimum width=15mm] (c0) at (1.35,1.55) {interior\\ $f(c)$};
  \foreach \i in {0,1,2} \draw[-{Latex[length=1.6mm]}]
      (0.15,1.95-0.4*\i) node[left,nl]{$c_\i$} -- (c0.west|-0,1.95-0.4*\i);
  \draw[-{Latex[length=1.6mm]}] (c0.east) -- ++(0.6,0) node[right,nl]{root};
  \node[win, fit=(c0), inner sep=4pt] {};
  \node[nl] at (1.35,0.72) {fanout-closed window};

  \draw[-{Latex[length=2.2mm]}, thick] (1.35,0.4) -- (1.35,-0.1);

  \node[lb] at (-0.75,-0.85) {after};
  \node[enc, minimum width=13mm] (A) at (1.05,-0.85) {encoder\\ $u=Ac\oplus m$};
  \node[bx, minimum width=15mm] (g) at (3.55,-0.85)
       {re-expressed\\ interior $g(u)$};
  \foreach \i in {0,1,2} \draw[-{Latex[length=1.6mm]}]
      (-0.15,-0.45-0.4*\i) node[left,nl]{$c_\i$} -- (A.west|-0,-0.45-0.4*\i);
  \draw[-{Latex[length=1.6mm]}] (A) -- node[above,nl]{$u$} (g);
  \draw[-{Latex[length=1.6mm]}] (g.east) -- ++(0.6,0) node[right,nl]{root};
  \node[nl,align=center] at (1.05,-1.75) {weight-one rows\\ are free rail aliases};
  \node[nl,align=center] at (3.55,-1.75) {priced by remapping\\ the whole netlist};
\end{tikzpicture}

\smallskip
{\footnotesize
\begin{tabular}{lrrr}
\toprule
\texttt{reconv24} & devices & $T_1$ (pJ) & $T_2$ (pJ) \\
\midrule
default            & $1208$ & $0.111078$ & $0.481338$ \\
\texttt{--linwin}  & $1188$ & $\mathbf{0.069938}$ & $\mathbf{0.440198}$ \\
\midrule
ratio              &        & $0.6296$   & $0.9145$ \\
\bottomrule
\end{tabular}}

\caption{Interior affine windows (Sec.~\ref{sec:windows}) on the shipped
\texttt{reconv24} fixture, three eight-input parity trees recombined through
majority logic.  The pass selects a fanout-closed region, applies an
invertible affine change of coordinates $u=Ac\oplus m$ to the cut variables,
and re-expresses the interior as $g(u)=f(A^{-1}(u\oplus m))$ in algebraic
normal form.  The encoder rows are emitted only for coordinates the
transformed function actually uses, and a weight-one uncomplemented row is
aliased directly to its leaf and costs no gate at all.  The local score only
\emph{orders} candidates.  Acceptance is global: the full netlist is
remapped and repriced, and the two-table gate decides.  One window is
accepted here, taking the uncapped figure to $0.6296$ of default and the
capped figure to $0.9145$.  The gap between the two ratios is the series
bound reclaiming part of the win, and is exactly why the tool reports both
tables rather than either one.}
\label{fig:ex-window}
\end{figure}

%% file: figures/ex_bdec.tex
\begin{figure}[t]
\centering
\begin{tikzpicture}[scale=0.80, transform shape,
  blk/.style={draw, rounded corners=2pt, minimum height=7mm, font=\scriptsize,
              align=center, fill=black!4},
  lb/.style={font=\scriptsize},
  nl/.style={font=\tiny, inner sep=1pt}]

  \node[lb] at (-0.55,1.35) {before};
  \node[blk, minimum width=30mm] (core0) at (1.65,1.35)
       {six 13-input parities\\ $y_i=\bigoplus_{j=0}^{12}x_{i+j}$};
  \draw[-{Latex[length=2mm]}] (-0.05,1.35) -- (core0.west);
  \draw[-{Latex[length=2mm]}] (core0.east) -- ++(0.55,0) node[right,lb]{$y$};
  \node[nl] at (1.65,0.72) {$72$ two-input XOR gates};

  \node[lb] at (-0.55,-0.75) {after};
  \node[blk, minimum width=17mm] (core1) at (1.15,-0.75)
       {re-encoded\\ core $h=Bf$};
  \node[blk, minimum width=15mm, fill=black!8] (dec) at (3.75,-0.75)
       {decoder\\ $y=B^{-1}h$};
  \draw[-{Latex[length=2mm]}] (-0.05,-0.75) -- (core1.west);
  \draw[-{Latex[length=2mm]}] (core1.east) -- node[above,nl]{$h$} (dec.west);
  \draw[-{Latex[length=2mm]}] (dec.east) -- ++(0.55,0) node[right,lb]{$y$};
  \node[nl,align=center] at (1.15,-1.55) {one 13-input parity\\ + five 2-input};
  \node[nl,align=center] at (3.75,-1.55) {$15$ XOR gates\\ (rows of weight $\ge2$)};
\end{tikzpicture}

\smallskip
{\footnotesize
\setlength{\tabcolsep}{4pt}
\begin{tabular}{lrrr}
\toprule
\texttt{bdtoy2} & devices & $T_1$ (pJ) & $T_2$ (pJ) \\
\midrule
default        & $230$ & $0.271524$ & $0.271524$ \\
\texttt{--bdec} & $192$ & $\mathbf{0.246840}$ & $\mathbf{0.246840}$ \\
\midrule
ratio          &       & $0.9091$   & $0.9091$ \\
\bottomrule
\end{tabular}}

\smallskip
{\scriptsize
$B=\begin{bmatrix}
1&0&0&0&0&0\\ 1&1&0&0&0&0\\ 0&1&1&0&0&0\\
0&0&1&1&0&0\\ 0&0&0&1&1&0\\ 0&0&0&0&1&1
\end{bmatrix}$\quad
$B^{-1}=\begin{bmatrix}
1&0&0&0&0&0\\ 1&1&0&0&0&0\\ 1&1&1&0&0&0\\
1&1&1&1&0&0\\ 1&1&1&1&1&0\\ 1&1&1&1&1&1
\end{bmatrix}$}

\caption{The linear pre-filter (Sec.~\ref{sec:bdec}) on the shipped fixture
\texttt{bdtoy2}, whose six outputs are overlapping thirteen-input parities,
$y_i=x_i\oplus\cdots\oplus x_{i+12}$.  A single elementary row addition,
$e_{i+1}\mathbin{+}=e_i$, cancels twelve of the thirteen terms, so the
re-encoded core computes $h_{i+1}=y_i\oplus y_{i+1}=x_i\oplus x_{i+13}$,
a two-input parity, and only $h_0$ retains a long chain.  The boundary is
restored by a decoder computing $y=B^{-1}h$.  The two matrices are where
Corollary~\ref{cor:asymmetry} becomes concrete: $B$ is light, every row of
weight at most two, while its inverse is the lower-triangular all-ones
matrix with row weights $1,\dots,6$, so by Lemma~\ref{lem:decoder} the
decoder costs $\sum_{w_r\ge2}(w_r-1)=15$ two-input XOR gates.  Bounding the
weight of $B$ alone would have admitted candidates whose inexpensive core is paid
for at the boundary with interest, which is why the search bounds both.
The matrices above are the fixture's ideal, the point the construction is
built around.  The recorded search reaches part of the way there, accepting
two row additions under the shipped budget and returning a circuit at
$0.9091$ of the default on both tables, at $192$ devices against $230$.
The gap between what the structure admits and what the search finds is a
tie-break problem, discussed in Sec.~\ref{sec:bdec}.  It is \emph{not} the
larger effect described in Sec.~\ref{sec:bdecrealize}, which this fixture
is too small to exhibit.}
\label{fig:ex-bdec}
\end{figure}

%% file: tab_passes.tex
\begin{tabular}{lrrrrrrrr}
\toprule
Circuit & \texttt{davio} & \texttt{linwin} & \texttt{mowin} & \texttt{prefix} & \texttt{factor} & \texttt{bdec} & \texttt{optall} & \texttt{lw+mw} \\
\midrule
xa & 1.000 & 1.000 & 1.000 & 1.000 & 1.000 & 1.000 & 1.000 & 1.000 \\
c17 & 1.000 & 1.000 & 1.000 & 1.000 & 1.000 & 1.000 & 1.000 & 1.000 \\
ctrl & 1.000 & \textbf{0.971} & \textbf{0.873} & \textbf{0.863} & \textbf{0.765} & 1.000 & \textbf{0.863} & \textbf{0.892} \\
t481 & 1.000 & 1.000$^\dagger$ & 1.000$^\dagger$ & 1.000$^\dagger$ & $\ddagger$ & 1.000 & 1.000$^\dagger$ & 1.000$^\dagger$ \\
reconv24 & 1.000 & \textbf{0.915} & \textbf{0.915} & 1.000 & 1.000 & 1.000 & \textbf{0.915} & \textbf{0.915} \\
c432 & 1.000 & 1.000 & \textbf{0.862} & \textbf{0.817} & 1.000 & 1.000 & \textbf{0.817} & \textbf{0.862} \\
c880 & \textbf{0.968} & \textbf{0.970} & \textbf{0.979} & 1.000 & \textbf{0.972} & 1.000 & \textbf{0.968} & \textbf{0.968} \\
router & 1.000 & 1.000 & \textbf{0.973} & \textbf{0.985} & $\ddagger$ & 1.000$^\dagger$ & \textbf{0.936} & \textbf{0.965} \\
8-bit hash & 1.000 & 1.000 & 1.000 & 1.000 & 1.000 & 1.000 & 1.000 & 1.000 \\
c499 & 1.000 & 1.000 & \textbf{0.971} & 1.000 & 1.000 & 1.000 & \textbf{0.971} & \textbf{0.971} \\
c1355 & \textbf{0.709} & \textbf{0.795} & \textbf{0.741} & \textbf{0.839} & 1.000 & 1.000 & \textbf{0.673} & \textbf{0.739} \\
crc8 & 1.000 & 1.000 & 1.000 & 1.000 & 1.000 & 1.000$^\dagger$ & 1.000 & 1.000 \\
c1908 & \textbf{0.990} & \textbf{0.992} & \textbf{0.947} & 1.000 & \textbf{0.959} & 1.000 & \textbf{0.964} & \textbf{0.956} \\
c2670 & \textbf{0.987} & \textbf{0.978} & \textbf{0.910} & \textbf{0.969} & $\ddagger$ & 1.000 & \textbf{0.941} & \textbf{0.944} \\
c3540 & \textbf{0.990} & \textbf{0.995} & \textbf{0.934} & 1.000 & \textbf{0.985} & 1.000 & \textbf{0.938} & \textbf{0.946} \\
dec & 1.000 & 1.000 & \textbf{0.882} & 1.000 & \textbf{0.677} & 1.000 & \textbf{0.882} & \textbf{0.882} \\
c5315 & \textbf{0.973} & \textbf{0.957} & \textbf{0.926} & 1.000 & 1.000 & 1.000 & \textbf{0.916} & \textbf{0.891} \\
c7552 & \textbf{0.954} & \textbf{0.945} & \textbf{0.911}$^\dagger$ & \textbf{0.940} & \textbf{0.998} & 1.000 & \textbf{0.898}$^\dagger$ & \textbf{0.891} \\
c6288 & 1.000 & 1.000 & \textbf{0.984}$^\dagger$ & 1.000$^\dagger$ & \textbf{0.987} & 1.000 & 1.000$^\dagger$ & \textbf{0.984}$^\dagger$ \\
12-bit hash & 1.000 & 1.000 & 1.000 & 1.000 & $\ddagger$ & 1.000 & 1.000 & 1.000 \\
\bottomrule
\end{tabular}

%% file: tab_heldout_default.tex
\begin{tabular}{lrrrr}
\toprule
Circuit & Gates & Devices & $T_1$ (pJ) & $T_2$ (pJ) \\
\midrule
s641 & 379 & 618 & 0.617100 & 0.631499 \\
s838 & 390 & 972 & 0.538934 & 0.538934 \\
s953 & 395 & 1402 & 0.789888 & 0.865997 \\
c1238 & 508 & 1828 & 0.678810 & 0.808401 \\
c1196 & 529 & 1636 & 0.650012 & 0.705551 \\
s1488 & 653 & 2012 & 0.481338 & 0.561561 \\
s1423 & 657 & 1746 & 1.353506 & 1.404931 \\
int2float & 805 & 640 & 0.119306 & 0.143990 \\
10-bit hash & 1044 & 19745 & 0.411400 & 5.195982 \\
10-bit hash B & 1044 & 19457 & 0.411400 & 5.082847 \\
cavlc & 2293 & 4928 & 0.374374 & 1.145749 \\
s5378 & 2779 & 4726 & 4.336156 & 4.545970 \\
priority & 3327 & 2888 & 2.353208 & 2.357322 \\
adder & 3561 & 3822 & 2.098140 & 2.098140 \\
i2c & 4086 & 3880 & 2.414918 & 2.610333 \\
s9234 & 5597 & 7836 & 7.211842 & 7.674667 \\
max & 8975 & 7496 & 8.602374 & 8.849214 \\
dalu & 9553 & 6562 & 5.076676 & 5.284433 \\
sin & 16853 & 17994 & 35.688950 & 36.589916 \\
arbiter & 35712 & 27692 & 48.598682 & 50.145546 \\
\midrule
median & & & 1.071697 & 2.227731 \\
\bottomrule
\end{tabular}

%% file: tab_heldout_arms.tex
\begin{tabular}{llrr}
\toprule
Circuit & best arm & $T_2$ & $T_1$ \\
\midrule
adder & \texttt{elim} & \textbf{0.9059} & 0.8255 \\
c1196 & \texttt{lw+mw} & \textbf{0.9388} & 0.9177 \\
c1238 & \texttt{lw+mw} & \textbf{0.8931} & 0.9394 \\
cavlc & \texttt{elim} & \textbf{0.4093} & 0.9560 \\
dalu & \texttt{linwin} & \textbf{0.9774} & 0.9878 \\
i2c & \texttt{elim} & \textbf{0.9149} & 0.8313 \\
int2float & \texttt{lw+mw} & \textbf{0.8143} & 0.7931 \\
max & \textit{(none)} & 1.0000 & 1.0000 \\
priority & \textit{(none)} & 1.0000 & 1.0000 \\
s1423 & \texttt{lw+mw} & \textbf{0.9253} & 0.9392 \\
s1488 & \texttt{opt-all} & \textbf{0.8645} & 0.8205 \\
s5378 & \texttt{elim} & \textbf{0.7593} & 0.7163 \\
s641 & \texttt{elim} & \textbf{0.8371} & 0.8333 \\
s838 & \texttt{lw+mw} & \textbf{0.9771} & 0.9695 \\
s9234 & \texttt{elim} & \textbf{0.7875} & 0.7353 \\
s953 & \texttt{opt-all} & \textbf{0.8409} & 0.8646 \\
sin & \texttt{davio} & \textbf{0.9971} & 0.9963 \\
10-bit hash & \textit{(none)} & 1.0000 & 1.0000 \\
10-bit hash B & \textit{(none)} & 1.0000 & 1.0000 \\
\midrule
\multicolumn{4}{l}{\emph{15 of 19 improve, best-arm $T_2$ median $0.9149$/19.}} \\
\bottomrule
\end{tabular}

%% file: tab_bdecset.tex
\begin{table}[t]
\caption{The boundary-structure set.  Ratios are the pass's result against
the same circuit synthesized without the flag, on the uncapped and capped
tables.  The set is not all winners by design: four members are expected to
decline, each for a different reason, and a set of only wins would say
nothing about where the pass stops.  Coverage is the number of candidate row
additions priced against the number legal at the shipped weight bound.}
\label{tab:bdecset}
\centering
\small
\setlength{\tabcolsep}{1.5pt}
\begin{tabular}{lrrrrrr}
\toprule
circuit & $m$ & coverage & acc. & devices & $T_1$ & $T_2$ \\
\midrule
\multicolumn{7}{l}{\emph{accepted}}\\
\texttt{bdtoy2}     &  6 & 49/30  & 2 & $230\!\to\!192$   & $0.9091$ & $0.9091$ \\
\texttt{bdslide}    &  8 & 53/56  & 3 & $2176\!\to\!1432$ & $\mathbf{0.8929}$ & $\mathbf{0.7955}$ \\
\texttt{gray2bin16} & 16 & 44/240 & 2 & $2082\!\to\!1626$ & $1.0526$ & $0.9355$ \\
\midrule
\multicolumn{7}{l}{\emph{declined, and why}}\\
\texttt{gray2bin32} & 32 & 25/992 & 0 & $314\!\to\!314$   & $1.0000$ & $1.0000$ \\
\texttt{hamsynd}    &  5 & 21/20  & 0 & $1320\!\to\!1320$ & $1.0000$ & $1.0000$ \\
\texttt{bdwin}      &  6 & 25/30  & 0 & $1128\!\to\!1128$ & $1.0000$ & $1.0000$ \\
\texttt{bdnull}     &  6 & 25/30  & 0 & $888\!\to\!888$   & $1.0000$ & $1.0000$ \\
\bottomrule
\end{tabular}
\end{table}

%% file: tab_wallclock.tex
\begin{tabular}{lr@{\qquad}lr}
\toprule
\multicolumn{2}{c}{default flow} & \multicolumn{2}{c}{pass study (\texttt{c880})} \\
circuit & s & arm & s \\
\midrule
\texttt{c17} & 0.09 & \texttt{--davio} & 7.69 \\
\texttt{xa} & 0.03 & \texttt{--factor} & 2.82 \\
\texttt{ctrl} & 0.46 & \texttt{--prefix} & 198.7 \\
\texttt{reconv24} & 3.63 & \texttt{--linwin} & 17.0 \\
\texttt{crc8} & 42.5 & \texttt{--mowin} & 162.2 \\
\texttt{c432} & 0.16 & \texttt{--bdec} & 72.6 \\
\texttt{c499} & 1.85 & five-pass chain & 212.4 \\
\texttt{c880} & 2.65 &  &  \\
\texttt{c1355} & 4.84 &  &  \\
\texttt{c1908} & 1.79 &  &  \\
\texttt{c2670} & 1.76 &  &  \\
\texttt{c3540} & 4.63 &  &  \\
\texttt{c5315} & 4.73 &  &  \\
\texttt{c6288} & 35.5 &  &  \\
\texttt{c7552} & 28.3 &  &  \\
\texttt{dec} & 0.24 &  &  \\
\texttt{router} & 70.5 &  &  \\
\texttt{hash8} & 1.06 &  &  \\
\texttt{t481} & 47.0 &  &  \\
\texttt{hash12} & 47.2 &  &  \\
\bottomrule
\end{tabular}

%% file: tab_families.tex
\begin{tabular}{lrrrr}
\toprule
Family & geo.\ mean $T_2$ ratio & best case & worst case & wins \\
\midrule
ecrl & 0.612 & 0.211 & 1.25 & 17 \\
pal & 0.612 & 0.211 & 1.25 & 0 \\
pfal & 0.722 & 0.239 & 1.50 & 0 \\
cal & 0.830 & 0.261 & 1.76 & 0 \\
tgate & 0.857 & 0.280 & 1.33 & 0 \\
tgate\_sl6 & 1.000 & 1.000 & 1.00 & 1 \\
spgal & 1.308 & 0.470 & 2.45 & 0 \\
2lal & 2.502 & 0.715 & 6.73 & 2 \\
s2lal & 5.004 & 1.430 & 13.45 & 0 \\
\bottomrule
\end{tabular}

%% file: tab_families_heldout.tex
\begin{tabular}{lrrrrr}
\toprule
Family & $n$ & geo.\ mean & best & worst & wins \\
\midrule
ecrl & 20 & 0.669 & 0.218 & 1.34 & 18 \\
pal & 20 & 0.669 & 0.218 & 1.34 & 18 \\
pfal & 20 & 0.804 & 0.289 & 1.74 & 0 \\
cal & 20 & 0.938 & 0.360 & 2.14 & 0 \\
tgate & 20 & 0.950 & 0.285 & 1.69 & 0 \\
tgate\_sl6 & 20 & 1.000 & 1.000 & 1.00 & 2 \\
spgal & 20 & 1.498 & 0.570 & 3.29 & 0 \\
2lal$^{\dagger}$ & 17 & 4.384 & 0.777 & 20.53 & 0 \\
s2lal$^{\dagger}$ & 17 & 8.767 & 1.553 & 41.06 & 0 \\
\bottomrule
\end{tabular}

%% file: tab_certified.tex
\begin{tabular}{lrl}
\toprule
Circuit & gap (\%) & certificate \\
\midrule
c17 & 0.00 & exact optimum computed \\
xa & 0.00 & exact optimum computed \\
reconv24 & 0.00 & exact optimum computed \\
t481 & 0.00 & \textbf{attains the cover-space floor} \\
c499 & 1.54 & anytime bound (ceiling) \\
c1355 & 6.08 & anytime bound (ceiling) \\
c1908 & 11.38 & anytime bound (ceiling) \\
c880 & 19.30 & anytime bound (ceiling) \\
router & 20.21 & anytime bound (ceiling) \\
c432 & 23.74 & anytime bound (ceiling) \\
ctrl & 24.39 & exact optimum computed \\
crc8 & 35.05 & anytime bound (ceiling) \\
c2670 & 40.61 & anytime bound (ceiling) \\
c5315 & 47.14 & anytime bound (ceiling) \\
dec & 50.00 & exact optimum computed \\
c3540 & 101.19 & anytime bound (ceiling) \\
c7552 & 103.56 & anytime bound (ceiling) \\
c6288 & 250.42 & anytime bound (ceiling) \\
8-bit hash & $-96.90$ & realizer outside the cover space \\
12-bit hash & $-98.37$ & realizer outside the cover space \\
\bottomrule
\end{tabular}

%% file: renesis_arxiv.bib
@inproceedings{zulehner2019adiabatic,
  author    = {Alwin Zulehner and Michael P. Frank and Robert Wille},
  title     = {Design Automation for Adiabatic Circuits},
  booktitle = {Proc. 24th Asia and South Pacific Design Automation Conf. (ASP-DAC)},
  year      = {2019}, publisher = {ACM},
  doi       = {10.1145/3287624.3287673}, note = {arXiv:1809.02421},
  pages     = {VERIFY}
}

@article{landauer1961, author={Rolf Landauer},
  title={Irreversibility and Heat Generation in the Computing Process},
  journal={IBM J. Research and Development}, volume={5}, number={3},
  pages={183--191}, year={1961}}

@article{reeb2014landauer, author={David Reeb and Michael M. Wolf},
  title={An improved {Landauer} principle with finite-size corrections},
  journal={New J. Physics}, volume={16}, pages={103011}, year={2014}}

@inproceedings{lindgren_lowpower_bdd,
  title={Low power optimization technique for BDD mapped circuits},
  author={Lindgren, Per and Kerttu, Mikael and Thornton, Mitchell and Drechsler, Rolf},
  booktitle={Proceedings of the 2001 Asia and South Pacific Design Automation Conference},
  pages={615--621},
  year={2001}
}

@incollection{thornton_rm_chapter1,
  author={Mitchell A. Thornton and David K. Houngninou and D. Michael Miller},
  title={Computing the {Reed-Muller} Spectrum / Algebraic Normal Form:
         Functional Methods},
  booktitle={Advances in the {Boolean} Domain},
  editor={Bernd Steinbach},
  publisher={Cambridge Scholars Publishing},
  address={Newcastle upon Tyne, UK}, year={2022}, chapter={1}}

@incollection{houngninou_rm_chapter2,
  author={David K. Houngninou and Mitchell A. Thornton and D. Michael Miller},
  title={Extracting the {Reed-Muller} Spectrum / Algebraic Normal Form from a
         Circuit Specification},
  booktitle={Advances in the {Boolean} Domain},
  editor={Bernd Steinbach},
  publisher={Cambridge Scholars Publishing},
  address={Newcastle upon Tyne, UK}, year={2022}, chapter={2}}

@article{arf1941untersuchungen,
  title={Untersuchungen {\"u}ber quadratische formen in k{\"o}rpern der charakteristik 2.(teil i.).},
  author={Arf, Cahit},
  year={1941},
  publisher={Walter de Gruyter, Berlin/New York Berlin, New York}
}

@book{hurst1985spectral,
  author={Stanley L. Hurst and D. Michael Miller and Jon C. Muzio},
  title={Spectral Techniques in Digital Logic},
  publisher={Academic Press}, address={London, UK}, year={1985}}

@incollection{steinbach2016chapter,
  author={Mitchell A. Thornton},
  title={A Vector Space Method for {Boolean} Networks},
  booktitle={Problems and New Solutions in the {Boolean} Domain},
  editor={Bernd Steinbach},
  publisher={Cambridge Scholars Publishing},
  address={Newcastle upon Tyne, UK}, year={2016},
  chapter={1.1}, pages={3--50}}

@book{thornton_vsim_book,
  author={Mitchell A. Thornton},
  title={Modeling Digital Switching Circuits with Linear Algebra},
  series={Synthesis Lectures on Digital Circuits and Systems},
  number={44},
  publisher={Morgan \& Claypool Publishers}, year={2014},
  doi={10.2200/S00579ED1V01Y201404DCS044},
  isbn={9781627052337}}

@article{bennett1973,
  author  = {Charles H. Bennett},
  title   = {Logical Reversibility of Computation},
  journal = {IBM Journal of Research and Development},
  volume  = {17},
  number  = {6},
  pages   = {525--532},
  year    = {1973},
  doi     = {10.1147/rd.176.0525}
}

@article{athas1994adiabatic,
  author  = {W. C. Athas and L. J. Svensson and J. G. Koller and
             N. Tzartzanis and E. Y.-C. Chou},
  title   = {Low-Power Digital Systems Based on Adiabatic-Switching
             Principles},
  journal = {IEEE Transactions on Very Large Scale Integration (VLSI)
             Systems},
  volume  = {2},
  number  = {4},
  pages   = {398--407},
  year    = {1994},
  doi     = {10.1109/92.335009}
}

@inproceedings{koller1992adiabatic,
  title={Adiabatic switching, low energy computing, and the physics of storing and erasing information},
  author={Koller, Jeffrey G and Athas, William C},
  booktitle={Workshop on Physics and Computation},
  pages={267--270},
  year={1992},
  organization={IEEE}
}

@article{moon1996ecrl,
  title={An efficient charge recovery logic circuit},
  author={Moon, Yong and Jeong, Deog-Kyoon},
  journal={IEICE transactions on electronics},
  volume={79},
  number={7},
  pages={925--933},
  year={1996},
  publisher={The Institute of Electronics, Information and Communication Engineers}
}

@article{vetuli1996pfal,
  title={Positive feedback in adiabatic logic},
  author={Vetuli, A and Di Pascoli, Stefano and Reyneri, LM and others},
  journal={Electronics Letters},
  volume={32},
  number={20},
  pages={1867--1868},
  year={1996},
  publisher={[Stevenage, etc., Institution of Electrical Engineers]}
}

@inproceedings{maksimovic1995cal,
  title={Clocked CMOS adiabatic logic with single AC power supply},
  author={Maksimovic, Dragan and Oklobdzija, Vojin G},
  booktitle={ESSCIRC'95: Twenty-first European Solid-State Circuits Conference},
  pages={370--373},
  year={1995},
  organization={IEEE}
}

@article{oklobdzija1997pal,
  title={Pass-transistor adiabatic logic using single power-clock supply},
  author={Oklobdzija, Vojin G and Maksimovic, Dragan and Lin, Fengcheng},
  journal={IEEE Transactions on Circuits and Systems II: Analog and Digital Signal Processing},
  volume={44},
  number={10},
  pages={842--846},
  year={1997},
  publisher={IEEE}
}

@inproceedings{aghoram2004twolal,
  title={Driving Fully-Adiabatic Logic Circuits Using Custom High-Q MEMS Resonators.},
  author={Anantharam, Venkiteswaran and He, Maojiao and Natarajan, Krishna and Xie, Huikai and Frank, Michael P},
  booktitle={ESA/VLSI},
  pages={5--11},
  year={2004}
}

@inproceedings{frank2020s2lal,
  title={Reversible computing with fast, fully static, fully adiabatic CMOS},
  author={Frank, Michael P and Brocato, Robert W and Tierney, Brian D and Missert, Nancy A and Hsia, Alexander H},
  booktitle={2020 International Conference on Rebooting Computing (ICRC)},
  pages={1--8},
  year={2020},
  organization={IEEE}
}

@article{kumar2019eespfal,
  author  = {S. Dinesh Kumar and Himanshu Thapliyal and Azhar Mohammad},
  title   = {{EE-SPFAL}: A Novel Energy-Efficient Secure Positive Feedback
             Adiabatic Logic for {DPA} Resistant {RFID} and Smart Card},
  journal = {IEEE Transactions on Emerging Topics in Computing},
  volume  = {7},
  number  = {2},
  pages   = {281--293},
  year    = {2019},
  doi     = {10.1109/TETC.2016.2645128}
}

@inproceedings{rauchenecker2017mixdes,
  title={Exploiting reversible logic design for implementing adiabatic circuits},
  author={Rauchenecker, Andreas and Ostermann, Timm and Wille, Robert},
  booktitle={2017 MIXDES-24th International Conference" Mixed Design of Integrated Circuits and Systems},
  pages={264--270},
  year={2017},
  organization={IEEE}
}

@article{morrison2014tcad,
  author  = {Matthew Morrison and Nagarajan Ranganathan},
  title   = {Synthesis of Dual-Rail Adiabatic Logic for Low Power Security
             Applications},
  journal = {IEEE Transactions on Computer-Aided Design of Integrated Circuits
             and Systems},
  volume  = {33},
  number  = {7},
  pages   = {975--988},
  year    = {2014},
  doi     = {10.1109/TCAD.2014.2313454}
}

@inproceedings{clark2024isvlsi,
  author    = {Jared Clark and Eric Raffel and Himanshu Thapliyal},
  title     = {Automated Generation of Dual Rail Adiabatic Gates from Binary
               Decision Diagrams},
  booktitle = {Proc. IEEE Computer Society Annual Symposium on VLSI (ISVLSI)},
  address   = {Knoxville, TN},
  pages     = {809--811},
  year      = {2024},
  doi       = {10.1109/ISVLSI61997.2024.00160}
}

@inproceedings{dhananjay2019iscas,
  author    = {Kaushik Dhananjay and Emre Salman},
  title     = {Special Session: Adiabatic Circuits for Energy-Efficient and
               Secure {IoT} Systems},
  booktitle = {Proc. IEEE Int. Symposium on Circuits and Systems (ISCAS)},
  year      = {2019}
}

@phdthesis{ye1997thesis,
  title={Design and synthesis of adiabatic circuits for low power},
  author={Ye, Yibin},
  year={1997},
  publisher={Purdue University}
}

@incollection{hanninen2014adiabatic,
  title={Adiabatic CMOS: limits of reversible energy recovery and first steps for design automation},
  author={H{\"a}nninen, Ismo and Snider, Gregory L and Lent, Craig S},
  booktitle={Transactions on Computational Science XXIV: Special Issue on Reversible Computing},
  pages={1--20},
  year={2014},
  publisher={Springer}
}

@inproceedings{bairamkulov2024date,
  title={Technology-aware logic synthesis for superconducting electronics},
  author={Bairamkulov, Rassul and Lee, Siang-Yun and Calvino, Alessandro Tempia and Marakkalage, Dewmini Sudara and Yu, Mingfei and De Micheli, Giovanni},
  booktitle={2024 Design, Automation \& Test in Europe Conference \& Exhibition (DATE)},
  pages={1--6},
  year={2024},
  organization={IEEE}
}

@inproceedings{tsui1993dac,
  author    = {Chi-Ying Tsui and Massoud Pedram and Alvin M. Despain},
  title     = {Technology Decomposition and Mapping Targeting Low Power
               Dissipation},
  booktitle = {Proc. 30th Design Automation Conference (DAC)},
  pages     = {68--73},
  year      = {1993},
  doi       = {10.1145/157485.164577}
}

@inproceedings{tiwari1993dac,
  author    = {Vivek Tiwari and Pranav Ashar and Sharad Malik},
  title     = {Technology Mapping for Low Power},
  booktitle = {Proc. 30th Design Automation Conference (DAC)},
  pages     = {74--79},
  year      = {1993},
  doi       = {10.1145/157485.164581}
}

@article{najm1994survey,
  title={A survey of power estimation techniques in VLSI circuits},
  author={Najm, Farid N},
  journal={IEEE transactions on very large scale integration (VLSI) systems},
  volume={2},
  number={4},
  pages={446--455},
  year={1994},
  publisher={IEEE}
}

@misc{mishchenko2011power,
  title={A power optimization toolbox for logic synthesis and mapping},
  author={Chung, Stephen Jang Kevin and Brayton, Alan Mishchenko Robert},
  year={2009}
}

@article{brent1982adder,
  title={A regular layout for parallel adders},
  author={Brent and Kung},
  journal={IEEE transactions on Computers},
  volume={100},
  number={3},
  pages={260--264},
  year={1982},
  publisher={IEEE}
}

@article{kogge1973recurrence,
  title={A parallel algorithm for the efficient solution of a general class of recurrence equations},
  author={Kogge, Peter M and Stone, Harold S},
  journal={IEEE transactions on computers},
  volume={100},
  number={8},
  pages={786--793},
  year={1973},
  publisher={IEEE}
}

@article{bryant1986bdd,
  author  = {Randal E. Bryant},
  title   = {Graph-Based Algorithms for {B}oolean Function Manipulation},
  journal = {IEEE Transactions on Computers},
  volume  = {C-35},
  number  = {8},
  pages   = {677--691},
  year    = {1986},
  doi     = {10.1109/TC.1986.1676819}
}

@inproceedings{brglez1985iscas85,
  author    = {Franc Brglez and Hideo Fujiwara},
  title     = {A Neutral Netlist of 10 Combinational Benchmark Circuits and a
               Target Translator in {F}ortran},
  booktitle = {Proc. IEEE Int. Symposium on Circuits and Systems (ISCAS)},
  year      = {1985}
}

@phdthesis{maslov2003,
  author={Dmitri Maslov},
  title={Reversible Logic Synthesis},
  school={University of New Brunswick},
  address={Fredericton, New Brunswick, Canada},
  year={2003}, month={September},
  note={Minimum garbage: Theorem~1, Sec.~3.1, p.~28}}

@book{thornton2001spectral,
  author={Mitchell A. Thornton and Rolf Drechsler and D. Michael Miller},
  title={Spectral Techniques in {VLSI} {CAD}},
  publisher={Kluwer Academic Publishers}, address={Boston, MA}, year={2001}}

@article{thornton2015tc,
  author={Mitchell A. Thornton},
  title={Simulation and Implication Using a Transfer Function Model for
         Switching Logic},
  journal={IEEE Trans. Computers},
  volume={64}, number={12}, pages={3580--3590}, year={2015}}

@inproceedings{avizienis1977nversion,
  author    = {Avi{\v z}ienis, Algirdas and Chen, Liming},
  title     = {On the Implementation of {N}-Version Programming for
               Software Fault Tolerance During Execution},
  booktitle = {Proc. IEEE COMPSAC},
  year      = {1977},
  pages     = {149--155}
}

@article{avizienis1985nversion,
  author  = {Avi{\v z}ienis, Algirdas},
  title   = {The {N}-Version Approach to Fault-Tolerant Software},
  journal = {IEEE Transactions on Software Engineering},
  volume  = {SE-11},
  number  = {12},
  pages   = {1491--1501},
  year    = {1985}
}

@article{knight1986independence,
  author  = {Knight, John C. and Leveson, Nancy G.},
  title   = {An Experimental Evaluation of the Assumption of Independence
             in Multiversion Programming},
  journal = {IEEE Transactions on Software Engineering},
  volume  = {SE-12},
  number  = {1},
  pages   = {96--109},
  year    = {1986}
}

@article{mckeeman1998differential,
  author  = {McKeeman, William M.},
  title   = {Differential Testing for Software},
  journal = {Digital Technical Journal},
  volume  = {10},
  number  = {1},
  pages   = {100--107},
  year    = {1998}
}

@book{boehm1981cocomo,
  author    = {Boehm, Barry W.},
  title     = {Software Engineering Economics},
  publisher = {Prentice-Hall},
  address   = {Englewood Cliffs, NJ},
  year      = {1981}
}

@article{cong1994flowmap,
  title={FlowMap: An optimal technology mapping algorithm for delay optimization in lookup-table based FPGA designs},
  author={Cong, Jason and Ding, Yuzheng},
  journal={IEEE Transactions on Computer-Aided Design of Integrated Circuits and Systems},
  volume={13},
  number={1},
  pages={1--12},
  year={1994},
  publisher={IEEE}
}

@misc{vaire2025whitepaper,
  author       = {{Vaire Computing}},
  title        = {Software Whitepaper},
  howpublished = {\url{https://vaire.co/uploads/Software-Whitepaper.pdf}},
  year         = {2025}
  }

@article{soeken2022epfl,
  author  = {Soeken, Mathias and Riener, Heinz and Haaswijk, Winston and
             Testa, Eleonora and Schmitt, Bruno and Meuli, Giulia and
             Mozafari, Fereshte and De Micheli, Giovanni},
  title   = {The {EPFL} Logic Synthesis Libraries},
  journal = {arXiv preprint arXiv:1805.05121},
  year    = {2022}
}

@inproceedings{meuli2019pebbling,
  author    = {Meuli, Giulia and Soeken, Mathias and Roetteler, Martin and
               Bj{\o}rner, Nikolaj and De Micheli, Giovanni},
  title     = {Reversible Pebbling Game for Quantum Memory Management},
  booktitle = {Design, Automation and Test in Europe Conference (DATE)},
  pages     = {288--291},
  year      = {2019},
  doi       = {10.23919/DATE.2019.8715092}
 }

@inproceedings{soeken2012revkit,
  author    = {Soeken, Mathias and Frehse, Stefan and Wille, Robert and
               Drechsler, Rolf},
  title     = {{RevKit}: An Open Source Toolkit for the Design of Reversible
               Circuits},
  booktitle = {Reversible Computation (RC 2011)},
  series    = {Lecture Notes in Computer Science},
  volume    = {7165},
  pages     = {64--76},
  publisher = {Springer},
  address   = {Berlin, Heidelberg},
  year      = {2012},
  doi       = {10.1007/978-3-642-29517-1_6}
 }

@inproceedings{schmitt2022tweedledum,
  author    = {Schmitt, Bruno and De Micheli, Giovanni},
  title     = {{tweedledum}: A Compiler Companion for Quantum Computing},
  booktitle = {Design, Automation and Test in Europe Conference (DATE)},
  pages     = {7--12},
  year      = {2022}
 }

@article{adarsh2022syrec,
  author  = {Adarsh, Smaran and Burgholzer, Lukas and Manjunath, Tanmay and
             Wille, Robert},
  title   = {{SyReC} Synthesizer: An {MQT} Tool for Synthesis of Reversible
             Circuits},
  journal = {Software Impacts},
  volume  = {14},
  pages   = {100451},
  year    = {2022},
  doi     = {10.1016/j.simpa.2022.100451}
}

@inproceedings{wille2010syrec,
  title={SyReC: A programming language for synthesis of reversible circuits},
  author={Wille, Robert and Offermann, Sebastian and Drechsler, Rolf},
  booktitle={2010 Forum on Specification \& Design Languages (FDL 2010)},
  pages={184--189},
  year={2010},
  organization={IET}
}

@inproceedings{amy2017reverc,
  author    = {Amy, Matthew and Roetteler, Martin and Svore, Krysta M.},
  title     = {Verified Compilation of Space-Efficient Reversible Circuits},
  booktitle = {Computer Aided Verification (CAV 2017), Part {II}},
  series    = {Lecture Notes in Computer Science},
  volume    = {10427},
  pages     = {3--21},
  publisher = {Springer},
  year      = {2017},
  doi       = {10.1007/978-3-319-63390-9_1}
}

@misc{nangate45,
  author       = {{Nangate Inc.}},
  title        = {Nangate 45\,nm Open Cell Library},
  howpublished = {Library \texttt{NangateOpenCellLibrary}, revision 1.0,
                  characterized 10 February 2011 with NGLibraryCharacterizer
                  v2011.01},
  year         = {2011},
  note         = {Typical corner, \SI{25}{\celsius}, \SI{1.1}{\volt}.
                  Generated with Nangate Library Creator from the FreePDK45
                  base kit of North Carolina State
                  University~\cite{stine2007freepdk}.  Now distributed by
                  Silvaco through Si2~\cite{si2opencell}.}
}
